\documentclass[12pt]{article}
\usepackage{amsmath}
\usepackage{graphicx}
\usepackage{enumerate}
\usepackage{natbib}
\usepackage{url} 
\usepackage{bbm}
\usepackage[caption=false]{subfig}
\usepackage{rotating}
\usepackage{xcolor}
\usepackage{hyperref}

\usepackage{xr}

\usepackage{verbatim}

\newcommand{\blind}{0}

\usepackage[export]{adjustbox}
\usepackage{array}

\usepackage{enumitem}
\usepackage{setspace}
\usepackage{multirow, booktabs}
\usepackage{amsmath, amsthm, amssymb, amsfonts}
\usepackage{bm, bbm}
\usepackage[ruled, vlined]{algorithm2e}
\usepackage{color}
\usepackage[T1]{fontenc}
\usepackage{mathtools, nccmath}
\usepackage{tabularx}

\allowdisplaybreaks

\setlist[enumerate]{itemsep=0mm}

\DeclareMathAlphabet\mathbfcal{OMS}{cmsy}{b}{n}

\newcommand{\mbf}{\mathbf}
\newcommand{\mc}{\mathcal}
\newcommand{\bmx}{\begin{bmatrix}}
\newcommand{\emx}{\end{bmatrix}}

\newcommand{\vep}{\varepsilon}

\renewcommand{\l}{\left}
\renewcommand{\r}{\right}
\newcommand{\lvertiii}[1]{{\left\vert\kern-0.25ex\left\vert\kern-0.25ex\left\vert #1 
    \right\vert\kern-0.25ex\right\vert\kern-0.25ex\right\vert}}
\newcommand{\vertiii}[1]{{\vert\kern-0.25ex\vert\kern-0.25ex\vert #1 
    \vert\kern-0.25ex\vert\kern-0.25ex\vert}}
    
\def\wh{\widehat}
\def\wt{\widetilde}

\newcommand{\E}[0]{\mathsf{E}}
\newcommand{\Var}[0]{\mathsf{Var}}

\newcommand{\p}{\mathsf{P}}
\newcommand{\R}{\mathbb{R}}
\newcommand{\Z}{\mathbb{Z}}
\newcommand{\C}{\mathbb{C}}
\newcommand{\N}{\mathbb{N}}
\newcommand{\iid}{\mbox{\scriptsize{iid}}}
\newcommand{\nn}{\nonumber}

\newcommand{\cp}{\theta}  
\newcommand{\Cp}{\Theta}  
\renewcommand{\k}{[k]} 
\renewcommand{\lll}{[l]}  
\newcommand{\bbG}{\mathbbm{G}}
\newcommand{\bbg}{\mathbbm{g}}

\theoremstyle{definition}
\newtheorem{thm}{Theorem}[section]
\theoremstyle{definition}
\newtheorem{cor}[thm]{Corollary}
\theoremstyle{definition}
\newtheorem{lem}[thm]{Lemma}
\theoremstyle{definition}
\newtheorem{prop}[thm]{Proposition}
\theoremstyle{definition}
\newtheorem{assum}{Assumption}[section]
\theoremstyle{remark}
\newtheorem{rem}{Remark}[section]
\theoremstyle{definition}

\theoremstyle{definition}
\newtheorem{ex}{Example}[section]

\newif\ifJASA
\JASAtrue 

\newif\ifnotblind
\notblindtrue 

\ifJASA
\else

\fi

\def\spacingset#1{\renewcommand{\baselinestretch}%
{#1}\small\normalsize} \spacingset{1}

{
  \title{\bf High-dimensional time series segmentation
via factor-adjusted vector autoregressive modelling}
\ifnotblind
  \author{  Haeran Cho\\
  School of Mathematics, University of Bristol.
  \\~\\ 
  Hyeyoung Maeng\\
  Department of Mathematical Sciences, Durham University.
  \\~\\
  Idris A.\ Eckley and Paul Fearnhead\\
  Department of Mathematics and Statistics, Lancaster University.
  }
\fi
}

\begin{document}

\maketitle





\begin{abstract}
Vector autoregressive (VAR) models are popularly adopted for modelling high-dimensional time series, and their piecewise extensions allow for structural changes in the data. In VAR modelling, the number of parameters grow quadratically with the dimensionality which necessitates the sparsity assumption in high dimensions. However, it is debatable whether such an assumption is adequate for handling datasets exhibiting strong serial and cross-sectional correlations. We propose a piecewise stationary time series model that simultaneously allows for strong correlations as well as structural changes, where pervasive serial and cross-sectional correlations are accounted for by a time-varying factor structure, and any remaining idiosyncratic dependence between the variables is handled by a piecewise stationary VAR model. We propose an accompanying two-stage data segmentation methodology which fully addresses the challenges arising from the latency of the component processes. Its consistency in estimating both the total number and the locations of the change points in the latent components, is established under conditions considerably more general than those in the existing literature. We demonstrate the competitive performance of the proposed methodology on simulated datasets and an application to US blue chip stocks data.
\end{abstract}

\noindent%
{\it Keywords: data segmentation, vector autoregression, high dimensionality, factor model}  

\ifJASA
\spacingset{1.45} 
\fi

\section{Introduction}

Vector autoregressive (VAR) models 
are popular for modelling cross-sectional and serial correlations in multivariate, possibly 
high-dimensional time series. With, for example, applications in 
finance \citep{barigozzi2017},
biology \citep{shojaie2010discovering} and genomics \citep{michailidis2013autoregressive}.
Within such settings, the importance of data segmentation is well-recognised, 
and several methods exist for detecting change points in VAR models in both fixed \citep{kirch2015eeg} and high dimensions \citep{safikhani2020, wang2019, bai2020, maeng2022}.

\begin{figure}[!htb]
\centering
\begin{tabular}{ccc}
{\scriptsize (a) 18/03/2008--07/07/2009} & {\scriptsize (b) 18/03/2008--07/07/2009} & {\scriptsize (c) 18/03/2008--07/07/2009} 
\\ [-.2pt]
\includegraphics[width = .25\textwidth]{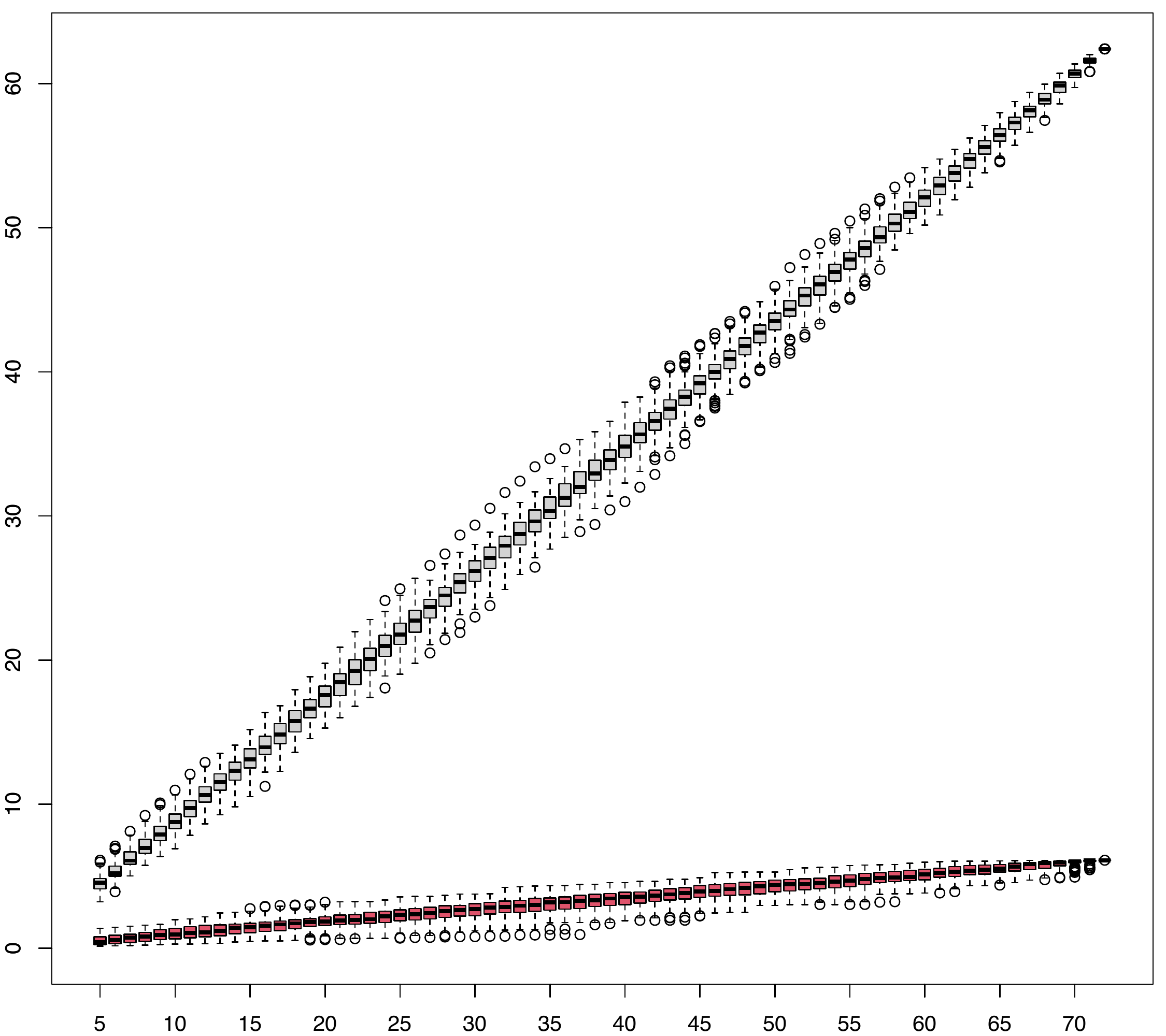} &
\includegraphics[width = .25\textwidth]{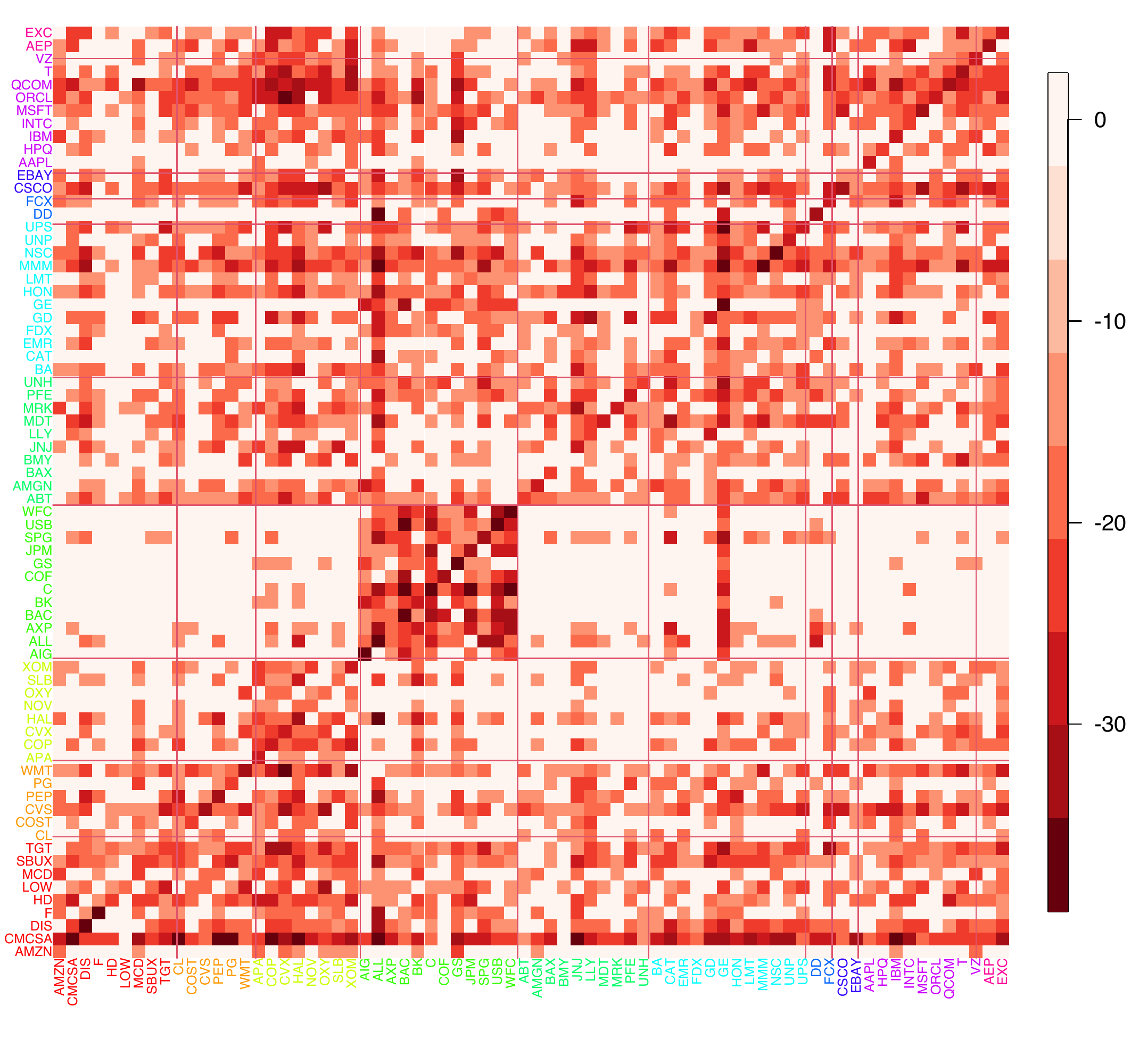} &
\includegraphics[width = .25\textwidth]{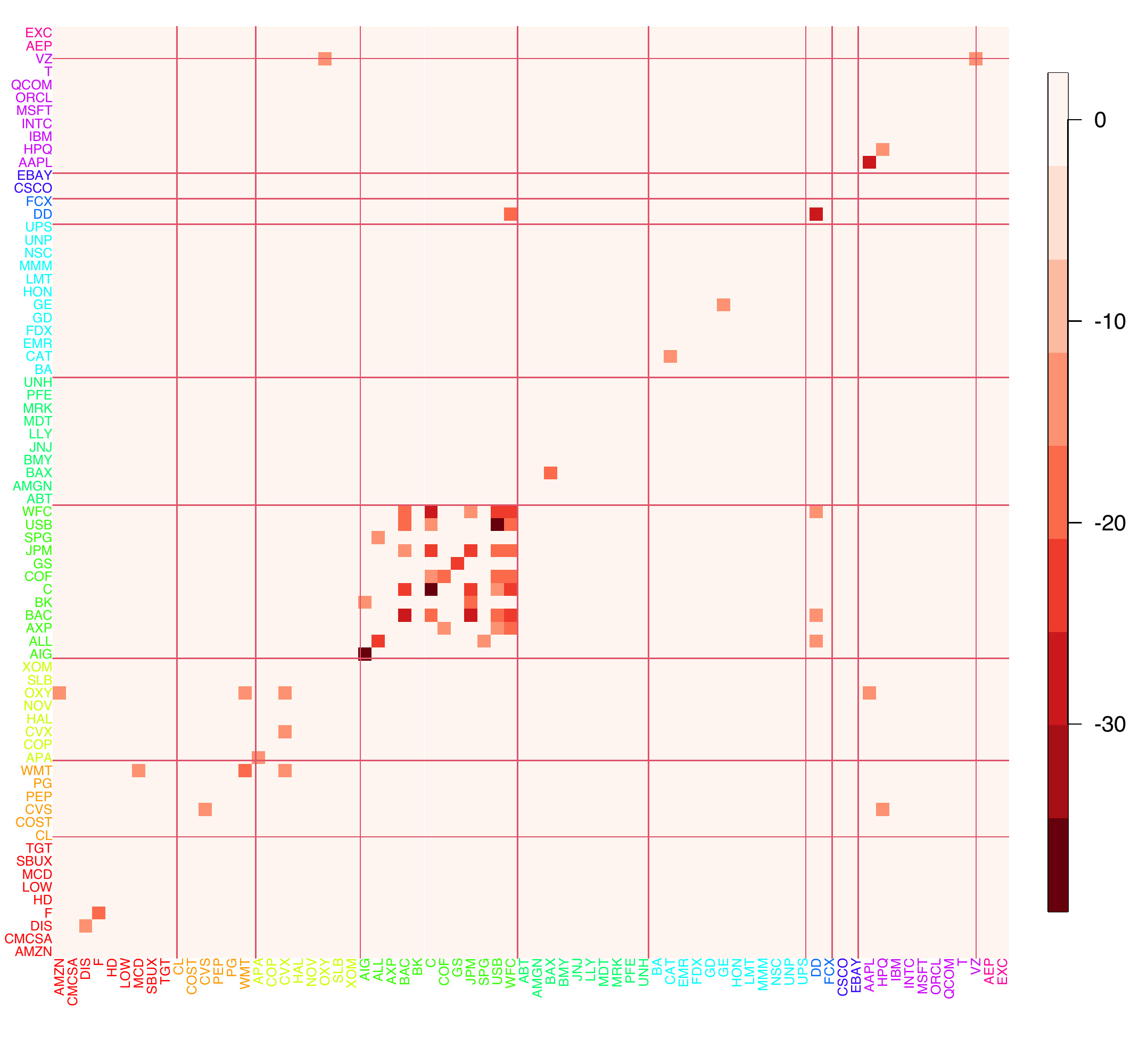} 
\\ [-.2pt]
{\scriptsize (d) 15/11/2003--07/06/2006} &
{\scriptsize (e) 08/06/2006--17/03/2008} &
{\scriptsize (f) 08/07/2009--28/07/2011} 
\\ [-.2pt]
\includegraphics[width = .25\textwidth]{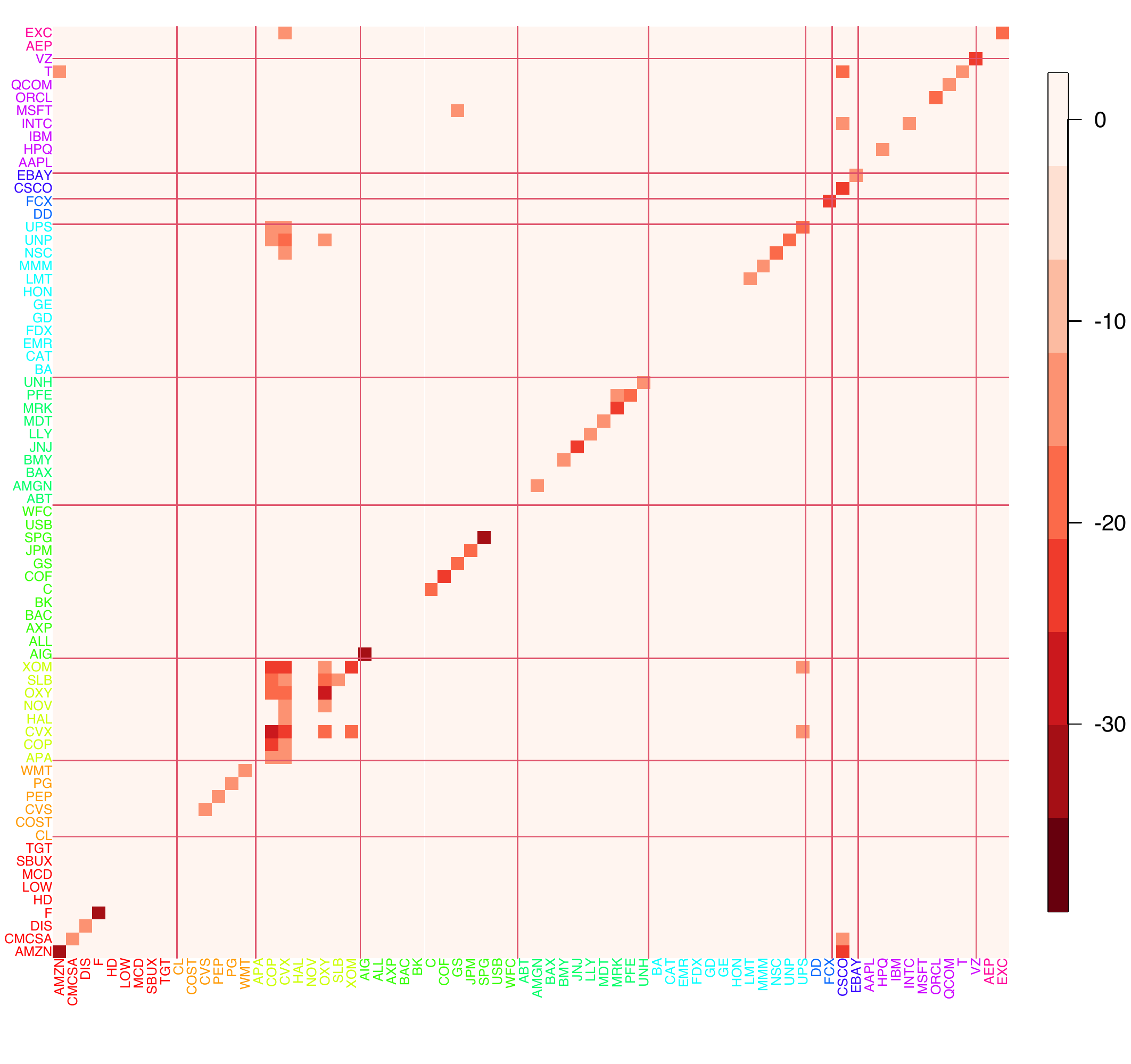} &
\includegraphics[width = .25\textwidth]{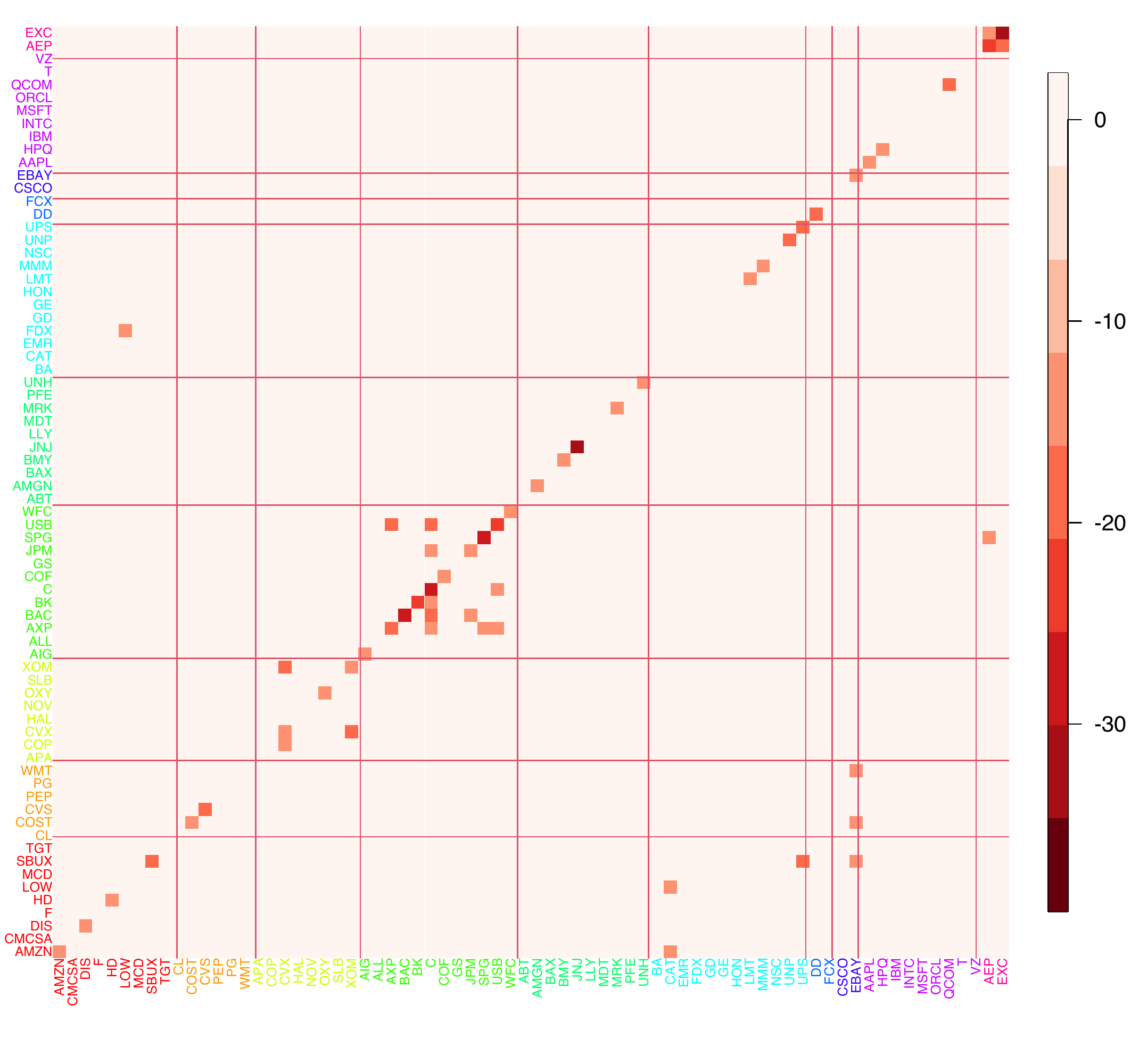} &
\includegraphics[width = .25\textwidth]{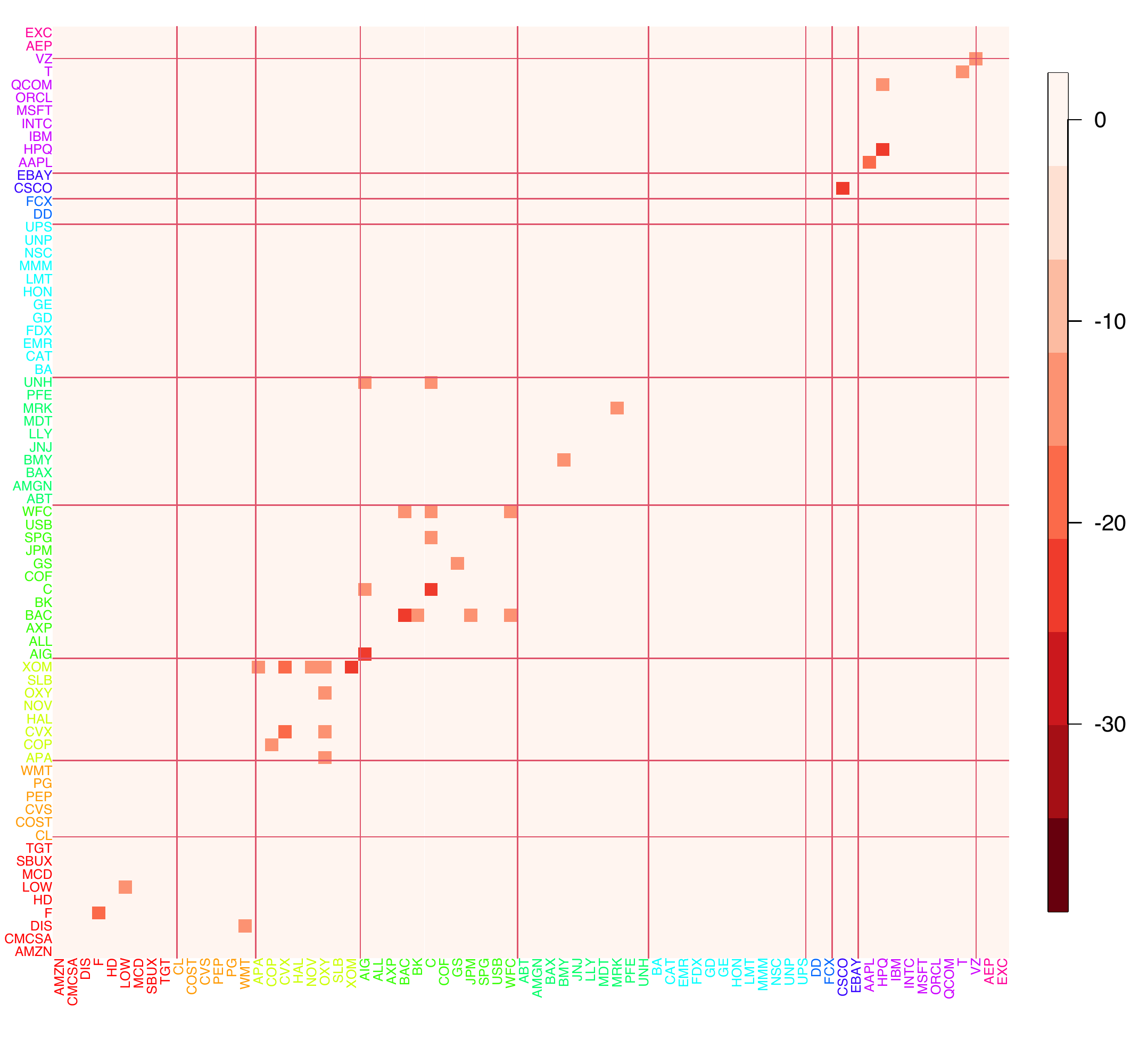} 
\end{tabular}

\caption{\small (a) The two largest eigenvalues of the long-run covariance matrix estimated from the volatility panel analysed in Section~\ref{sec:app}
(18/03/2008--07/07/2009,
$n = 223$)
with subsets of cross-sections randomly sampled $100$ times
for each given dimension $p \in \{5, \ldots, 72\}$ ($x$-axis).
(b) and (c): logged and truncated $p$-values 
from fitting a VAR($5$) model to the same dataset {\it without} and {\it with} factor-adjustment.
(d)--(f): logged and truncated $p$-values similarly obtained {\it with} factor-adjustment from the same variables over different periods.
In (b)--(f), for each pair of variables, the minimum $p$-value over the five lags is reported. Corresponding tickers are given in $x$- and $y$-axes and industrial sectors are indicated by the colours and boundaries drawn.}
\label{fig:ex}
\end{figure}

VAR modelling quickly becomes a high-dimensional problem as the number of parameters grows quadratically with the dimensionality. 
Accordingly, most existing methods for detecting change points in high-dimensional, piecewise stationary VAR processes assumes sparsity \citep{basu2015}.
However, it is debatable whether highly sparse models are appropriate for some applications. 
For example, \cite{giannone2021economic} note the difficulty of identifying sparse predictive representations 
for several macroeconomic applications.

We illustrate the inadequacy of the sparsity assumption on a volatility panel dataset (see Section~\ref{sec:app} for its description).
Figure~\ref{fig:ex}~(a) shows that as the dimensionality increases, the leading eigenvalue of the spectral density matrix at frequency~$0$ (i.e.\ the long-run covariance) estimated from the data also increases linearly. 
This indicates the presence of strong serial and cross-sectional correlations that cannot be accommodated by sparse VAR models. 
In Figure~\ref{fig:ex}~(b), we report the logged and truncated $p$-values obtained from fitting a VAR($5$) model to the same dataset (truncation level chosen at $\log(3.858 \times 10^{-6})$ by Bonferroni correction with the significance level $0.1$) via ridge regression, see \cite{cule2011significance}.
Strong dependence observed from most pairs of the variables further confirms that we cannot infer a sparse pairwise relationship from such data.
On the other hand, Figure~\ref{fig:ex}~(c) shows that once we estimate factors driving the strong correlations and adjust for their presence, there is evidence that the remaining dependence in the data can be modelled as being sparse.
Together, the plots (d), (e), (c) and~(f) display that the relationship between a pair of variables (after factor-adjustment) varies over time, particularly at the level of industrial sectors. Here, the intervals are chosen according to the data segmentation result reported in Section~\ref{sec:app}.
This example highlights the importance of (i)~accounting for the dominant correlations prior to fitting a model under the sparsity assumption, and (ii)~detecting structural changes when analysing time series datasets covering a long period.

Motivated by the aforementioned characteristics of high-dimensional time series data, factor-adjusted regression modelling has increasingly gained popularity \citep{fan2020factor, fan2021bridging, krampe2021}. 
The factor-adjusted VAR model proposed by \citet{barigozzi2022fnets} assumes that a handful of common factors capture strong serial and cross-sectional correlations, such that it is reasonable to assume a sparse VAR model on the remaining component to capture idiosyncratic, variable-specific dependence.
We extend this framework by proposing a new, piecewise stationary factor-adjusted VAR model and develop FVARseg, an accompanying change point detection methodology.
Below we summarise the methodological and theoretical contributions made in this paper.
\vspace{-12pt}

\paragraph{Generality of the modelling framework.}
We decompose the data into two piecewise stationary latent processes: one is driven by factors and accounts for dominant serial and cross-sectional correlations, and the other models sparse pairwise dependence via a VAR model.
We adopt the most general approach to factor modelling and allow both components to undergo changes which, in the case of the latter, are attributed to shifts in the VAR parameters.
To the best of our knowledge, such a general model simultaneously permitting the presence of common factors and change points, has not been studied in the literature previously. 
Accordingly, we are not aware of any method that can comprehensively address the data segmentation problem considered in this paper.
\vspace{-12pt}

\paragraph{Methodological novelty.} 
The idea of scanning the data for changes over moving windows, has successfully been applied to a variety of data segmentation problems \citep{preuss2015, eichinger2018, chen2019}.
We propose FVARseg, a two-stage methodology that combines this idea with statistics carefully designed to have good detection power against different types of changes in the two latent components.
In Stage~1 of FVARseg, motivated by that dominant factor-driven correlations appear as leading eigenvalues in the frequency domain, see e.g.\ Figure~\ref{fig:ex}~(a), we propose a detector statistic that contrasts the local spectral density matrix estimators from neighbouring moving windows in operator norm, which is well-suited to detect changes in the factor-driven component.

In Stage~2 for detecting change points in the latent piecewise stationary VAR process, we deliberately avoid estimating the latent process which may incur large errors.
Instead, we make use of (i)~the Yule-Walker equation that relates autocovariances (ACV) and VAR parameters, and (ii)~the availability of local ACV estimators of the latent VAR process after Stage~1. 
Combining these ingredients, we propose a novel detector statistic that enjoys methodological simplicity as well as statistical efficiency. 
Further, through sequential evaluation of the detector statistic, the second-stage procedure requires the estimation of local VAR parameters at selected locations only. 
Consequently it is highly competitive computationally when both the sample size and the dimensionality are large.
\vspace{-12pt}

\paragraph{Theoretical consistency.} 
FVARseg achieves consistency in estimating the total number and locations of the change points in both of the piecewise stationary factor-driven and VAR processes.
Our theoretical analysis is conducted in a setting considerably more general than those commonly adopted in the literature, permitting dependence across stationary segments and heavy-tailedness of the data.
We also derive the rate of localisation for each stage of FVARseg where we make explicit the influence of tail behaviour and the size of changes.
In particular, under Gaussianity, the estimators from Stage~1 nearly matches the minimax optimal rate derived for the simpler, covariance change point detection problem. 

The rest of the paper is structured as follows.
Section~\ref{sec:model} introduces the piecewise stationary factor-adjusted VAR model. 
Section~\ref{sec:method} describes the two stages of FVARseg, the proposed data segmentation methodology, and Section~\ref{sec:theory} establishes its theoretical consistency.
Section~\ref{sec:numeric} demonstrates the good performance of FVARseg empirically. 
\if0\blind{R code implementing our method is available from \url{https://github.com/haeran-cho/fvarseg}}.\fi
\vspace{-10pt}

\paragraph{Notation.}
Let $\mbf I$ and $\mbf O$ denote an identity matrix and a matrix of zeros whose dimensions depend on the context. 
For a random variable $X$ and $\nu \ge 1$, denote $\Vert X \Vert_\nu = (\E|X|^\nu)^{1/\nu}$.
Given $\mbf A = [a_{ii'}, \, 1 \le i \le m, \, 1 \le i' \le n]$, 
we denote by $\mbf A^*$ its transposed complex conjugate.
We define its element-wise $\ell_\infty$, $\ell_1$ and $\ell_2$-norms
by $\vert \mbf A \vert_\infty = \max_{i, i'} \vert a_{ii'} \vert$,
$\vert \mbf A \vert_1 = \sum_{i, i'} \vert a_{ii'} \vert$
and
$\vert \mbf A \vert_2 = \sqrt{\sum_{i, i'} \vert a_{ii'} \vert^2}$,
and its spectral and induced $L_1$, $L_\infty$-norms by
$\Vert \mbf A \Vert$,
$\Vert \mbf A \Vert_1 = \max_{1\le i' \le n} \sum_{i = 1}^m \vert a_{ii'}\vert$ and $\Vert \mbf A \Vert_\infty = \max_{1\le i \le n} \sum_{i' = 1}^m \vert a_{ii'}\vert$, respectively.
For positive definite $\mbf A$, we denote its minimum eigenvalue by $\Vert \mbf A \Vert_{\min}$.
For two real numbers, $a \vee b = \max(a, b)$ and $a \wedge b = \min(a, b)$. For two sequences $\{a_n\}$ and $\{b_n\}$, we write $a_n \asymp b_n$ if, for some constants $C_1, C_2 > 0$,
there exists $N \in \N$ such that $C_1 \le a_n b_n^{-1} \le C_2$ for all $n \ge N$.

\section{Piecewise stationary factor-adjusted VAR model}
\label{sec:model}

\subsection{Background}

A zero-mean, $p$-variate process $\bm\xi_t$ follows a VAR($d$) model if it satisfies
\begin{align}
\label{eq:var}
\bm\xi_t = \mbf A_1 \bm\xi_{t - 1} + \ldots + \mbf A_d \bm\xi_{t - d} + (\bm\Gamma)^{1/2} \bm\vep_t,
\end{align}
where 
$\mbf A_\ell \in \R^{p \times p}, \, 1 \le \ell \le d$, determine how future values of the series depend on their past. The $p$-variate random vector $\bm\vep_t = (\vep_{1t}, \ldots, \vep_{pt})^\top$ has $\vep_{it}$ which are independently and identically distributed (i.i.d.) for all $i$ and $t$ with $\E(\vep_{it}) = 0$ and $\Var(\vep_{it}) = 1$. 
The positive definite matrix $\bm\Gamma \in \R^{p \times p}$ is the covariance matrix of the innovations for the VAR process.

A factor-driven component exhibits strong cross-sectional and/or serial correlations by `loading' finite-dimensional factors linearly.
Among many, 
the generalised dynamic factor model (GDFM, \citeauthor{FHLR00}, \citeyear{FHLR00}, \citeyear{forni2015}) provides the most general approach (see Appendix~\ref{app:gdfm} for further discussions), and 
defines the $p$-variate factor-driven component $\bm\chi_t$~as
\begin{align}
\label{eq:gdfm}
\bm\chi_t = \mc B(L) \mbf u_t = \sum_{\ell=0}^\infty \mbf B_\ell \mbf u_{t-\ell}.
\end{align}
For fixed $q$, the $q$-variate random vector $\mbf u_t = (u_{1t}, \ldots, u_{qt})^\top$ contains the common factors which are shared across the variables and time, and $u_{jt}$ are assumed to be i.i.d.\ for all $j$ and $t$ with $\E(u_{jt}) = 0$ and $\Var(u_{jt}) = 1$. 
The matrix of square-summable filters $\mc B(L) = \sum_{\ell = 0}^\infty \mbf B_\ell L^\ell$ with the lag-operator $L$ and $\mbf B_\ell \in \R^{p \times q}$, serves the role of loadings under~\eqref{eq:gdfm}.

\cite{barigozzi2022fnets} propose a factor-adjusted VAR model, 
where the observations are assumed to be decomposed as a sum of the two latent components $\bm\xi_t$ and $\bm\chi_t$ in~\eqref{eq:var}--\eqref{eq:gdfm}, with pervasive correlations in the data are accounted for by $\bm\chi_t$ and the remaining dependence captured by $\bm\xi_t$.
In the next section, we introduce its piecewise stationary extension where both the factor-driven and VAR processes are allowed to undergo structural changes.

\subsection{Model}

We observe a zero-mean, $p$-variate piecewise stationary process $\mbf X_t = \bm\chi_t + \bm\xi_t$ where
\begin{align}
\label{eq:model}
& \l\{\begin{array}{ll}
\bm\chi_t = \bm\chi_t^{\k} = \mc B^{\k}(L) \mbf u_t 
& \text{for \ } \cp_{\chi, k} + 1 \le t \le \cp_{\chi, k + 1}, \, 0 \le k \le K_\chi, 
\\
\bm\xi_t = \bm\xi_t^{\k} = \sum_{\ell = 1}^d \mbf A^{\k}_\ell \bm\xi_{t - \ell} + (\bm\Gamma^{\k})^{1/2} \bm\vep_t & \text{for \ } \cp_{\xi, k} + 1 \le t \le \cp_{\xi, k + 1}, \, 0 \le k \le K_\xi.
\end{array}\r.
\end{align}
Here, $\cp_{\chi, k}, \, 1 \le k \le K_\chi$, denote the change points in the piecewise stationary factor-driven component $\bm\chi_t$ such that at each $\cp_{\chi, k}$, the filter of loadings $\mc B^{\k}(L)$ undergoes a change. 
We permit the factor number to vary over time as $q_k \le q$, with the factor $\mbf u^{\k}_t \in \R^{q_k}$ associated with $\bm\chi^{\k}_t$ being a sub-vector of $\mbf u_t \in \R^q$.
Similarly, $\cp_{\xi, k}, \, 1 \le k \le K_\xi$, denote the change points in the piecewise stationary VAR process $\bm\xi_t$ at which the VAR parameters undergo shifts; we permit the VAR innovation covariance matrix to vary as $\bm\Gamma^{\k}$ but our interest lies in detecting changes in VAR parameters, and the VAR order may vary over time as $d_k \le d$ with $\mbf A^{\k}_\ell = \mbf O$ for $\ell \ge d_k + 1$.
By convention, we denote $\cp_{\chi, 0} = \cp_{\xi, 0} = 0$ and $\cp_{\chi, K_\chi + 1} = \cp_{\xi, K_\xi + 1} = n$.
In line with the factor modelling literature, we assume that $\bm\chi_t$ and $\bm\xi_t$ are uncorrelated through having $\E(u_{jt} \vep_{it'})~=~0$ for any $i, j, t$ and $t'$.

The model~\eqref{eq:model} does not require that the change points in $\bm\chi_t$ and $\bm\xi_t$ are aligned, or that $K_\chi = K_\xi$.
Our goal is to estimate the total number and locations of the change points 
for both of the piecewise stationary latent processes.
Importantly, we allow $\{\bm\xi^{\k}_t, \, t \in \Z\}$ (resp. $\{\bm\chi^{\k}_t, \, t \in \Z\}$) to be dependent across $k$ through sharing the innovations $\bm\vep_t$ (resp. $\mbf u_t$).
This makes our model considerably more general than those found in the literature on (high-dimensional) data segmentation under VAR models \citep{wang2019, safikhani2020, bai2021} which assume independence across the segments.
Data segmentation under factor models has been considered by \cite{barigozzi2018} and \cite{li2019detection}
but they adopt a static approach to factor modelling.

\subsection{Assumptions}
\label{sec:assumptions}

We introduce assumptions that ensure the (asymptotic) identifiability of the two latent processes in~\eqref{eq:model} which are framed in terms of spectral properties, as well as controlling the degree of dependence in the data.
Denote by $\bm\Gamma^{\k}_\chi(\ell) = \E(\bm\chi^{\k}_{t - \ell} (\bm\chi^{\k}_t)^\top)$ the ACV matrix of $\bm\chi^{\k}_t$ at lag $\ell \in \Z$, and its spectral density matrix at frequency $\omega \in [-\pi, \pi]$ by $\bm\Sigma^{\k}_\chi(\omega) = (2\pi)^{-1} \sum_{\ell = -\infty}^\infty \bm\Gamma^{\k}_\chi(\ell) e^{-\iota \ell \omega}$ with $\iota = \sqrt{-1}$.
Then, $\mu^{\k}_{\chi, j}(\omega), \, 1 \le j \le q_k$, denote
the real, positive eigenvalues of $\bm\Sigma^{\k}_\chi(\omega)$ ordered by decreasing size.
We similarly define $\bm\Gamma^{\k}_\xi(\ell)$, $\bm\Sigma^{\k}_\xi(\omega)$ and $\mu^{\k}_{\xi, j}(\omega)$ for $\bm\xi^{\k}_t$.

\begin{assum} 
\label{assum:factor}
For each $0 \le k \le K_\chi$, the following holds:
There exist a positive integer $p_0 \ge 1$, pairs of functions
$\omega \mapsto \alpha^{\k}_j(\omega)$
and $\omega \mapsto \beta^{\k}_j(\omega)$
for $\omega \in [-\pi, \pi]$ and $1 \le j \le q_k$,
and $r_{k, j} \in (0, 1]$ satisfying $r_{k, 1} \ge \ldots \ge r_{k, q_k}$ such that
for all $p \ge p_0$,
\begin{align*}
& \beta^{\k}_1(\omega) \ge \frac{\mu^{\k}_{\chi, 1}(\omega)}{p^{r_{k, 1}}} \ge \alpha^{\k}_1(\omega) >
\ldots 
> \beta^{\k}_{q_k}(\omega) \ge \frac{\mu^{\k}_{q_k}(\omega)}{p^{r_{k, q_k}}} \ge \alpha^{\k}_{q_k}(\omega) > 0.
\end{align*}
\end{assum}
If $r_{k, j} = 1$ for all $1 \le j \le q_k$ as frequently assumed in the literature \citep{fan2013large, forni2015}, we are in the presence of $q_k$ factors 
which are equally pervasive for the whole cross-sections of $\bm\chi^{\k}_t$.
If $r_{k, j} < 1$ for some $j$, we permit the presence of `weak' factors.
Since our primary interest lies in change point analysis, 
we later introduce a related but distinct condition 
on the size of change in $\bm\chi_t$ in Assumption~\ref{assum:common:size}.
\begin{assum} 
\label{assum:idio}
\begin{enumerate}[noitemsep, wide, labelwidth=0pt, labelindent=0pt, label = (\roman*)]
\item \label{cond:idio:transition} 
$\det(\sum_{\ell = 1}^d \mbf A^{\k}_\ell z^\ell) \ne 0$ for all $\vert z \vert \le 1$ and $0 \le k \le K_\xi$.

\item \label{cond:idio:innov} 
$m_\vep \le \min_{0 \le k \le K_\xi} \Vert \bm\Gamma^{\k} \Vert_{\min} 
\le \max_{0 \le k \le K_\xi} \Vert \bm\Gamma^{\k} \Vert \le M_\vep$
for some constants $0 < m_\vep \le M_\vep$.

\item \label{cond:idio:coef} Consider the Wold decomposition 
$\bm\xi^{\k}_t = \sum_{\ell = 0}^\infty \mbf D^{\k}_\ell (\bm\Gamma^{\k})^{1/2} \bm\vep_{t - \ell}$
where $\mbf D^{\k}_\ell = [ D^{\k}_{\ell, ij}, \, 1 \le i, j \le p]$.
Then, there exist constants $\Xi > 0$ and $\varsigma > 2$ such that 
we have $C_{ij}, \, 1 \le i, j \le p$, satisfying
$\max\{ \max_{1 \le j \le p} \sum_{i = 1}^p C_{ij}, \max_{1 \le i \le p} \sum_{j = 1}^p C_{ij}, \max_{1 \le i \le p}\sqrt{\sum_{j = 1}^p C_{ij}^2} \} \le \Xi$
with which $\max_{0 \le k \le K_\xi}
\vert D^{\k}_{\ell, ij} \vert \le C_{ij} (1 + \ell)^{-\varsigma}$ for all $\ell \ge 0$.

\item \label{cond:idio:minspec}
$\min_{0 \le k \le K_\xi} \inf_{\omega \in [-\pi, \pi]} \mu^{\k}_{\xi, p}(\omega) \ge m_\xi$
for some fixed constant $m_\xi > 0$.
\end{enumerate}
\end{assum}

\begin{assum} 
\label{assum:common}
There exist constants $\Xi > 0$ and $\varsigma > 2$ such that for all $\ell \ge 0$,
\begin{center}
$\displaystyle{
\max_{0 \le k \le K_\chi} \max_{1 \le i \le p} \vert \mbf B^{\k}_{\ell, i \cdot} \vert_2 \le \Xi (1 + \ell)^{-\varsigma}
\text{ \ and \ }
\max_{0 \le k \le K_\chi}  \sqrt{\sum_{j = 1}^{q_k} \vert \mbf B^{\k}_{\ell, \cdot j} \vert_\infty^2} \le \Xi (1 + \ell)^{-\varsigma}.}$
\end{center}
\end{assum}

Assumption~\ref{assum:idio}~\ref{cond:idio:transition}--\ref{cond:idio:innov}
are standard conditions in the literature \citep{lutkepohl2005, basu2015}.
Under condition~\ref{cond:idio:coef} and Assumption~\ref{assum:common}, we have time-varying serial dependence in $\mbf X_t$ (across all segments) decay at an algebraic rate according to the functional dependence measure of \cite{zhang2021}, which is required for controlling the error in locally estimating spectral density and ACV matrices of $\mbf X_t$.
Assumption~\ref{assum:idio}~\ref{cond:idio:coef} allows for mild cross-correlations in $\bm\xi^{\k}_t$ while ensuring that $\mu^{\k}_{\xi, 1}(\omega)$ is uniformly bounded:

\begin{prop}
\label{prop:idio:eval}
Under Assumption~\ref{assum:idio}, uniformly over all $\omega \in [-\pi, \pi]$,
there exists some $M_\xi > 0$ depending only on $M_\vep$, $\Xi$ and $\varsigma$ 
such that $\max_{0 \le k \le K_\xi} \sup_{\omega \in [-\pi, \pi]} \mu^{\k}_{\xi, 1}(\omega) \le M_\xi$.
\end{prop}

\begin{rem}
\label{rem:idio:wold}
Proposition~\ref{prop:idio:eval}, together with Assumption~\ref{assum:idio}~\ref{cond:idio:minspec}, establishes the boundedness of the eigenvalues of $\bm\Sigma^{\k}_\xi(\omega)$,
which is commonly assumed in the high-dimensional VAR literature for the consistency of Lasso estimators. 
Assumption~\ref{assum:idio}~\ref{cond:idio:minspec} holds if there exists some constant $\Xi < \infty$ satisfying $\max( \max_{1 \le j \le p} \sum_{\ell = 1}^d \vert \mbf A^{\k}_{\ell, \cdot j} \vert_1, \max_{1 \le i \le p} \sum_{\ell = 1}^d \vert \mbf A^{\k}_{\ell, i \cdot} \vert_1 ) \le \Xi$ \citep{basu2015}.
When $d = 1$, we have $\mbf D^{\k}_\ell = (\mbf A^{\k}_1)^\ell$ such that if $\vert \mbf A^{\k}_1 \vert_\infty \le \gamma < 1$, Assumption~\ref{assum:idio}~\ref{cond:idio:coef} is readily satisfied with $\max(\Vert \mbf D^{\k}_\ell \Vert_1, \Vert \mbf D^{\k}_\ell \Vert_\infty) \le \Xi \gamma^{\ell - 1}$.
\end{rem}

From Assumption~\ref{assum:factor} and Proposition~\ref{prop:idio:eval},
the latent components in~\eqref{eq:model} are asymptotically identifiable as $p \to \infty$, thanks to the gap between $\mu^{\k}_{\chi, q_k}(\omega)$ diverging with $p$ and $\mu^{\k}_{\xi, 1}(\omega)$ which is uniformly bounded, which agrees with the phenomenon observed in Figure~\ref{fig:ex}~(a).




\section{Methodology}
\label{sec:method}


\subsection{Stage~1: Factor-driven component segmentation}
\label{sec:common:method}

\subsubsection{Change point detection}
\label{sec:common:method:seg}

The spectral density matrix of $\bm\chi_t$ is given by 
$\bm\Sigma^{\k}_\chi(\omega) = (2\pi)^{-1} \mc B^{\k}(e^{-\iota \omega}) (\mc B^{\k}(e^{-\iota \omega}))^*$ for $\cp_{\chi, k} + 1 \le t \le \chi_{\chi, k + 1}$, i.e.\ it varies over time in a piecewise constant manner with change points at $\cp_{\chi, k}, \, 1 \le k \le K_\chi$.
By Weyl's inequality, Assumption~\ref{assum:factor} and Proposition~\ref{prop:idio:eval} jointly indicate a gap in the eigenvalues of (time-varying) spectral density matrix of $\mbf X_t$, i.e.\ those attributed to the factor-driven component diverges with $p$ while the remaining ones are bounded for all $p$.
This suggests an approach that looks for changes in $\bm\chi_t$ from the behaviour of $\mbf X_t$ in the frequency domain which we further justify below.

\begin{ex}
\label{ex:common}
Suppose that $\bm\chi_t$ contains a single change point at $t = \cp_{\chi, 1}$ at which a new factor is introduced, i.e.\
$\bm\chi^{[0]}_t = \mc B^{[0]}(L) \mbf u_t^{[0]}$ and
$\bm\chi^{[1]}_t = \mc B^{[1]}(L) \mbf u_t^{[1]} = \mc B^{[0]}(L) \mbf u_t^{[0]} + \mbf b(L) v_t$
with $\mbf u^{[1]}_t = ((\mbf u_t^{[0]})^\top, v_t)^\top$,
which leads to $\bm\Sigma^{[1]}_\chi(\omega) - \bm\Sigma^{[0]}_\chi(\omega) = \mbf b(e^{-\iota \omega}) \mbf b^*(e^{-\iota \omega})/(2\pi)$.
Then, from the uncorrelatedness between $\bm\chi_t$ and $\bm\xi_t$ and Proposition~\ref{prop:idio:eval}, the time-varying spectral density of $\mbf X_t$, $\bm\Sigma_{x, t}(\omega)$, satisfies
$\Vert \sum_{t = 1}^{\cp_{\chi, 1}} \bm\Sigma_{x, t}(\omega) / \cp_{\chi, 1}
- \sum_{t = \cp_{\chi, 1} + 1}^n \bm\Sigma_{x, t}(\omega)/(n - \cp_{\chi, 1}) \Vert
= \Vert \mbf b(e^{-\iota \omega}) \mbf b^*(e^{-\iota \omega})\Vert/(2\pi) + O(1)$.
That is, the change in the spectral density of $\bm\chi_t$ is detectable
as a change in time-varying spectral density matrix of $\mbf X_t$ in operator norm, with the size of change diverging with $p$
as $\Vert (\mbf b(e^{-\iota \omega}))^* \mbf b(e^{-\iota \omega}) \Vert$ does so under Assumption~\ref{assum:factor}. 
\end{ex}

Thus, we detect changes in $\bm\chi_t$ by scanning for any large change in the spectral density matrix of~$\mbf X_t$ measured in operator norm, and propose the following moving window-based approach.
Given a bandwidth $G$, we estimate the local spectral density matrix of $\mbf X_t$ by
\begin{align}
\label{eq:local:spec}
\wh{\bm\Sigma}_{x, v}(\omega, G) &= \frac{1}{2\pi}
\sum_{\ell = -m}^m K\l(\frac{\ell}{m}\r) \wh{\bm\Gamma}_{x, v}(\ell, G) \exp(-\iota \ell \omega)
\quad \text{for} \quad G \le v \le n,
\end{align}
where $K(\cdot)$ denotes the Bartlett kernel, 
$m = G^\beta$ the kernel bandwidth with $\beta \in (0, 1)$,~and 
\begin{align}
\wh{\bm\Gamma}_{x, v}(\ell, G) = 
\frac{1}{G} \sum_{t = v - G + 1 + \ell}^v \mbf X_{t - \ell} \mbf X_t^\top
\text{ for } \ell \ge 0, \text{ \ and \ }
\wh{\bm\Gamma}_{x, v}(\ell, G) = \wh{\bm\Gamma}_{x, v}^\top(-\ell, G)
\text{ for } \ell < 0.
\label{eq:local:acv}
\end{align}
Then the following statistic
\begin{align}
\label{eq:common:test:stat}
T_{\chi, v}(\omega, G) = \l\Vert \wh{\bm\Sigma}_{x, v}(\omega, G) - \wh{\bm\Sigma}_{x, v + G}(\omega, G) \r\Vert, \quad G  \le v \le n - G,
\end{align}
serves as a good proxy of the difference in local spectral density matrices of $\bm\chi_t$ over $I_v(G) = \{v - G + 1, \ldots, v\}$ and $I_{v + G}(G) = \{v + 1, \ldots, v + G\}$.
To make it more precise, let $\bm\Sigma_{\chi, v}(\omega, G)$ denote a weighted average $\sum_{k = 0}^{K_\chi} w_{\chi, k}(v) \bm\Sigma^{\k}_\chi(\omega)$ with weights $w_{\chi, k}(v)$ corresponding to the proportion of $\bm\chi_t, \, t \in I_v(G)$, belonging to $\bm\chi^{\k}_t$ (see~\eqref{eq:common:tv:spec}). 
Then, $T^*_{\chi, v}(\omega, G) = \Vert \bm\Sigma_{\chi, v}(\omega, G) - 
\bm\Sigma_{\chi, v + G}(\omega, G) \Vert$, as a function of $v$, 
linearly increases and then decreases around the change points 
with a peak of size $\Vert \bm\Sigma^{\k}_\chi(\omega) - \bm\Sigma^{[k + 1]}_\chi(\omega) \Vert$ 
formed at $v = \cp_{\chi, k}$ for all $1 \le k \le K_\chi$,
provided that the bandwidth $G$ is not too large 
(in the sense of Assumption~\ref{assum:common:size}~\ref{cond:common:spacing} below).
The detector statistic $T_{\chi, v}(\omega, G)$ is designed to approximate $T^*_{\chi, v}(\omega, G)$
when $\bm\chi_t$ is not directly observed,
and thus is well-suited to detect and locate the change points therein.
Unlike other methods for detecting changes in the factor structure \cite[e.g.][]{li2019detection}, we do not require the number of factors, either for each segment or for the whole dataset, as an input for the construction of $T_{\chi, v}(\omega, G)$.

Once $T_{\chi, v}(\omega_l, G)$ is evaluated at the Fourier frequencies $\omega_l = 2\pi l/(2m + 1), \, 0 \le l \le m$, we adapt the maximum-check of \cite{eichinger2018} for simultaneous detection of the multiple change points.
Taking the pointwise maximum over the frequencies at each given location $v$, we check if $T_{\chi, v}(\omega(v), G)$ exceeds some threshold $\kappa_{n, p}$ where $\omega(v)$ denotes the frequency at which $T_v(\omega_l, G)$ is maximised, i.e.\ $\omega(v) = \arg\max_{\omega_l: \, 0 \le l \le m} T_v(\omega_l, G)$.
If so, it provides evidence that a change point $\cp_{\chi, k}$ is located near the time point $v$, but some care is needed to avoid detecting duplicate estimators, since the detector statistic is expected to take a large value over an interval containing $\cp_{\chi, k}$.
Therefore, denoting by $\mc I \subset \{G, \ldots, n - G\}$ the set containing all time points at which $T_{\chi, v}(\omega(v), G) > \kappa_{n, p}$, we regard $\wh\cp = \arg\max_{v \in \mc I} T_{\chi, v}(\omega(v), G)$ as a change point estimator if it is a local maximiser of $T_{\chi, v}(\omega(\wh\cp), G)$ within an interval of radius $\eta G$ centred at $\wh\cp$ with some $\eta \in (0, 1)$, i.e.\ $T_{\chi, \wh\cp}(\omega(\wh\cp), G) \ge \max_{\wh\cp - \eta G < v \le \wh\cp + \eta G} T_{\chi, v}(\omega(\wh\cp), G)$.
Once $\wh\cp$ is added to the set of final estimators, say $\wh{\Cp}_\chi$, in order to avoid the risk of duplicate estimators, we remove the interval of radius $G$ centred at $\wh\cp$ from $\mc I$, and repeat the same procedure with the maximiser of $T_{\chi, v}(\omega(v), G)$ at time points $v$ remaining in $\mc I$ until the set $\mc I$ is empty.
Algorithm~\ref{alg:one} in Appendix~\ref{sec:alg} outlines the steps of Stage~1 of FVARseg.



\subsubsection{Post-segmentation factor adjustment}
\label{sec:common:post}

Following the detection of change points in $\bm\chi_t$, we are able to estimate the segment-specific quantities related to $\bm\chi^{\k}_t$.
In view of the second-stage of FVARseg detecting change points in $\bm\xi_t$,
we describe how to estimate $\bm\Gamma^{\k}_\chi(\ell)$ with which we can estimate the ACV of $\bm\xi_t$.

For each $k = 0, \ldots, \wh{K}_\chi$, we first estimate the spectral density of $\mbf X_t$ over the segment $\{\wh\cp_{\chi, k} + 1, \ldots, \wh\cp_{\chi, k + 1}\}$ by $\wh{\bm\Sigma}_x^{\k}(\omega)$
as in~\eqref{eq:local:spec} using the sample ACV computed from the segment
(we use the same kernel bandwidth $m$ for simplicity).
Then noting that the spectral density matrix of $\bm\chi^{\k}_t$ is of rank $q_k$ under~\eqref{eq:model}, we estimate it from the eigendecomposition of $\wh{\bm\Sigma}^{\k}_x(\omega_l)$ by retaining only the $q_k$ largest eigenvalues, say $\wh\mu^{\k}_{x, j}(\omega_l)$, and the associated eigenvectors $\wh{\mbf e}^{\k}_{x, j}(\omega_l)$, and then estimate the ACV of $\bm\chi^{\k}_t$ by inverse Fourier transform, i.e.\
\begin{align}
\wh{\bm\Sigma}_\chi^{\k}(\omega_l) = \sum_{j = 1}^{q_k} \wh\mu^{\k}_{x, j}(\omega_l) \wh{\mbf e}^{\k}_{x, j}(\omega_l) \l(\wh{\mbf e}^{\k}_{x, j}(\omega_l)\r)^*
\text{ and }
\wh{\bm\Gamma}_\chi^{\k}(\ell) = \frac{2\pi}{2m + 1}\sum_{l = -m}^m \wh{\bm\Sigma}^{\k}_\chi(\omega_l) e^{\iota \omega_l \ell}.
\label{eq:common:est}
\end{align}
The estimators in~\eqref{eq:common:est} require the factor number $q_k$ as an input. 
We refer to \cite{hallin2007} for an information criterion (IC)-based estimator of $q_k$ that make use of the postulated eigengap in the spectral density matrix of $\mbf X_t$.

\subsection{Stage~2: Piecewise VAR process segmentation}
\label{sec:idio:method}


Applying the existing VAR segmentation methods in our setting requires estimating the $np$ elements of the latent piecewise stationary VAR process $\bm\xi_t$, 
which introduces additional errors and possibly results in the loss of statistical efficiency.
In addition, as discussed in Appendix~\ref{sec:comparison}, the existing methods tend to be computationally demanding, e.g.\ by evaluating the Lasso estimators $O(n^2)$ times in a dynamic programming algorithm, or solving a large fused Lasso objective function of dimension~$np^2d$. 
Instead, since we can estimate the local AVC of $\bm\xi_t$ from the post-segmentation factor-adjustment in Stage~1, our proposed methodology for segmenting the latent VAR component avoids estimating $\bm\xi_t$ directly. 
Also, as described below, the proposed method evaluates the local VAR parameters at carefully selected locations only, and thus is computationally efficient.

Specifically, our approach makes use of the Yule-Walker equation \citep{lutkepohl2005}.
Let $\bm\beta^{\k} = [\mbf A^{\k}_1, \ldots, \mbf A^{\k}_d]^\top \in \R^{(pd) \times p}$ contain all VAR parameters in the $k$th segment.
Then, it is related to the ACV matrices $\bm\Gamma^{\k}_\xi(\ell) = \E(\bm\xi^{\k}_{t - \ell}(\bm\xi^{\k}_t)^\top)$ as $\bbG^{\k} \bm\beta^{\k} = \bbg^{\k}$, where 
\begin{align}
\label{eq:idio:gamma}
\bbG^{\k} = \bmx 
\bm\Gamma^{\k}_\xi(0) & \bm\Gamma^{\k}_\xi(-1) & \ldots & \bm\Gamma^{\k}_\xi(-d + 1) 
\\
& & \ddots & 
\\
\bm\Gamma^{\k}_\xi(d - 1) & \bm\Gamma^{\k}_\xi(d - 2) & \ldots & \bm\Gamma^{\k}_\xi(0)
\emx 
\quad \text{and} \quad
\bbg^{\k} = \bmx
\bm\Gamma^{\k}_\xi(1)
\\
\vdots
\\
\bm\Gamma^{\k}_\xi(d)
\emx,
\end{align}
with $\bbG^{\k}$ being invertible due to Assumption~\ref{assum:idio}~\ref{cond:idio:minspec}. 
We propose to utilise this estimating equation in combination with the local ACV estimators of $\bm\xi_t$ obtained as described below.

For a given bandwidth $G$ and the interval $I_v(G) = \{v - G + 1, \ldots, v\}$, 
we estimate the ACV of $\bm\xi_t$ for $t \in I_v(G)$, by $\wh{\bm\Gamma}_{\xi, v}(\ell, G) = \wh{\bm\Gamma}_{x, v}(\ell, G) - \wh{\bm\Gamma}_{\chi, v}(\ell, G)$. 
Here, $\wh{\bm\Gamma}_{x, v}(\ell, G)$ is defined in~\eqref{eq:local:acv} and
$\wh{\bm\Gamma}_{\chi, v}(\ell, G)$ is a weighted average of $\wh{\bm\Gamma}^{\k}_\chi(\ell), \, 0 \le k \le \wh{K}_\chi$, the estimators of ACV of $\bm\chi^{\k}_t$ in~\eqref{eq:common:est}, with the weights given by the proportion of $I_v(G)$ covered by the $k$th segment 
(see~\eqref{eq:common:tv:acv:est} for the precise definition).
Replacing $\bm\Gamma^{\k}_\xi(\ell)$ 
with $\wh{\bm\Gamma}_{\xi, v}(\ell, G)$, we obtain $\wh{\bbG}_v(G)$ estimating a weighted average of $\bbG^{\k}$, 
and similarly $\wh{\bbg}_v(G)$. Then, we propose to scan
$T_{\xi, v}(\wh{\bm\beta}, G) = 
\vertiii{ (\wh{\bbG}_v(G) \wh{\bm\beta} - \wh{\bbg}_v(G)) -
(\wh{\bbG}_{v + G}(G) \wh{\bm\beta} - \wh{\bbg}_{v + G}(G) ) }$
with some inspection parameter $\wh{\bm\beta} \in \R^{(pd) \times p}$ and a matrix norm $\vertiii{\cdot}$.
We motivate this statistic by considering
$T^*_{\xi, v}(\wh{\bm\beta}, G) = \vertiii{(\bbG_v(G) \wh{\bm\beta} - \bbg_v(G)) - (\bbG_{v + G}(G) \wh{\bm\beta} - \bbg_{v + G}(G) )}$,
its population counterpart.
With appropriately chosen $G$ (see Assumption~\ref{assum:idio:size}~\ref{cond:idio:spacing} below),
$T^*_{\xi, v}(\wh{\bm\beta}, G) = 0$ if $v$ is far from all the change points in $\bm\xi_t$, i.e.\ $\min_k \vert v - \cp_{\xi, k} \vert \ge G$, while it is tent-shaped near the change points with a local maximum at $v = \cp_{\xi, k}$, provided that
\begin{align}
\label{eq:beta:cond}
\bbG^{[k - 1]}(\wh{\bm\beta} - \bm\beta^{[k - 1]}) \ne \bbG^{\k}(\wh{\bm\beta} - \bm\beta^{\k}).
\end{align}

For the inspection parameter, we adopt an $\ell_1$-regularised Yule-Walker estimator of the VAR parameters first considered by \cite{barigozzi2022fnets} in stationary settings.
At given $v_\circ \in \{G, \ldots, n\}$, we solve the constrained $\ell_1$-minimisation problem
\begin{align}
\label{eq:ds}
\wh{\bm\beta}_{v_\circ}(G) = {\arg\min}_{\bm\beta \in \R^{pd \times p}} \vert \bm\beta \vert_1
\quad \text{subject to} \quad
\l\vert \wh{\bbG}_{v_\circ}(G) \bm\beta - \wh{\bbg}_{v_\circ}(G) \r\vert_\infty \le \lambda_{n, p},
\end{align}
with a tuning parameter $\lambda_{n, p} > 0$.
The $\ell_\infty$-constraint in~\eqref{eq:ds} naturally leads to the choice $\vertiii{\cdot} = \vert \cdot \vert_\infty$, resulting in the following detector statistic:
\begin{align*}
T_{\xi, v}(\wh{\bm\beta}, G) = 
\l\vert \l(\wh{\bbG}_v(G) \wh{\bm\beta} - \wh{\bbg}_v(G)\r) -
\l(\wh{\bbG}_{v + G}(G) \wh{\bm\beta} - \wh{\bbg}_{v + G}(G) \r) \r\vert_\infty.
\end{align*}
For good detection power, the condition in~\eqref{eq:beta:cond} suggests using an estimator of $\bm\beta^{[k - 1]}$ or $\bm\beta^{\k}$ in place of $\wh{\bm\beta}$ for detecting $\cp_{\xi, k}$.
Therefore, we propose to evaluate $T_{\xi, v}(\wh{\bm\beta}_{v_\circ}(G), G)$ for $v \ge G$, with $\wh{\bm\beta}_{v_\circ}(G)$ updated sequentially at locations strategically selected as below.

\begin{figure}[htb]
\centering
\includegraphics[width = 1\textwidth]{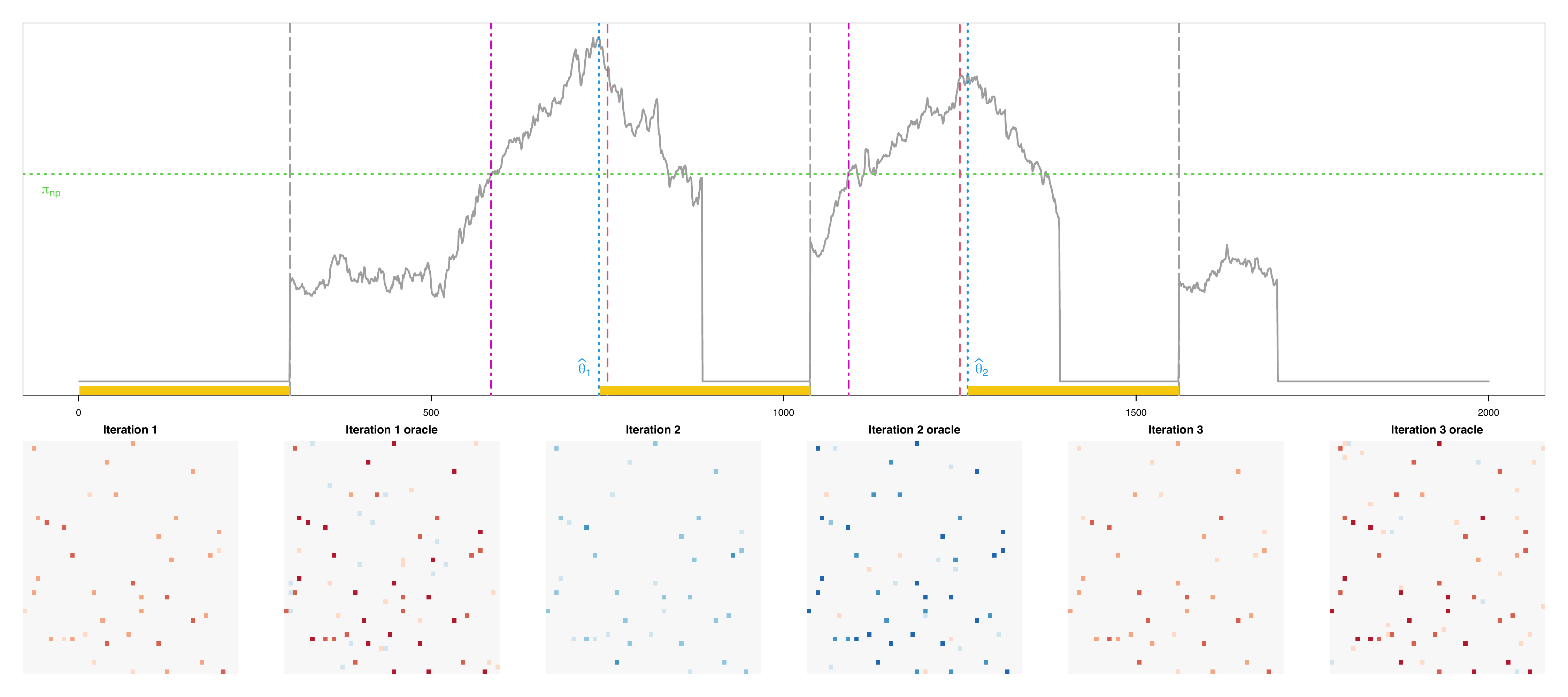}
\caption{\small Illustration of Stage~2 applied to a realisation from (M1) of Section~\ref{sec:sim} with $G = 300$ and $d = 1$.
Top: The solid curve represents $T_{\xi, v}(\wh{\bm\beta}, G), \, v_\circ \le v \le \check{\cp} + G$, computed at the three iterations of Steps~1--3 of Algorithm~\ref{alg:two}.
At each iteration, we use $\wh{\bm\beta} = \wh{\bm\beta}_{v_\circ}(G)$ estimated from each of the subsections the data highlighted in the $x$-axis (left to right); the corresponding estimators are plotted in the bottom panel and for comparison, we also plot the estimators obtained in the oracle setting where $\bm\xi_t$ is observable (all plots have the identical $z$-axis range).
The locations of $v_\circ$, $\check{\cp}$ and $\wh{\cp}$ in Algorithm~\ref{alg:two}, and
$\cp_{\xi, k}$ are denoted by the vertical long-dashed, dot-dashed, dotted and dashed lines, respectively. 
The horizontal line represents $\pi_{n, p}$ chosen as described in Section~\ref{sec:numeric:detail}.}
\label{fig:stage:two}
\end{figure}

First we estimate $\bm\beta^{[0]}$ by $\wh{\bm\beta} = \wh{\bm\beta}_G(G)$ in~\eqref{eq:ds} with $v_\circ = G$ 
and scan the data using $T_{\xi, v}(\wh{\bm\beta}, G), \, v \ge v_\circ$.
When $T_{\xi, v}(\wh{\bm\beta}, G)$ exceeds some threshold, say $\pi_{n, p}$, at $v = \check\cp$ for the first time, it signifies that a change has occurred in the neighbourhood. 
Reducing the search for a change point to $\{\check\cp, \ldots, \check\cp + G\}$, we identify a change point estimator as the local maximiser $\wh\cp_{\xi, 1} = \arg\max_{\check\cp \le v \le \check\cp + G} T_{\xi, v}(\wh{\bm\beta}, G)$.
Then updating $\wh{\bm\beta}$ with $\wh{\bm\beta}_{v_\circ}(G)$ obtained at $v_\circ = \wh{\cp}_{\xi, 1} + (\eta + 1) G$ for some $\eta \in (0, 1]$ (i.e.\ only using an interval of length $G$ located strictly to the right of $\wh\cp_{\xi, 1}$ for its computation), we continue screening $T_{\xi, v}(\wh{\bm\beta}, G), \, v \ge v_\circ$, until it next exceeds $\pi_{n, p}$.
These steps of screening $T_{\xi, v}(\wh{\bm\beta}, G)$ and updating $\wh{\bm\beta}$ are repeated iteratively until the end of the data sequence is reached.
Algorithm~\ref{alg:two} in Appendix~\ref{sec:alg} outlines the steps of the Stage~2 methodology.

Figure~\ref{fig:stage:two} illustrates that although $\bm\xi_t$ is latent, at each iteration, $\wh{\bm\beta}_{v_\circ}(G)$ does as well as its oracle counterpart (obtained as in~\eqref{eq:ds} with the sample ACV of $\bm\xi_t$ replacing $\wh{\bm\Gamma}_{\xi, v}(\ell, G)$).
Computationally, this strategy benefits from that the costly solution to the $\ell_1$-minimisation problem in~\eqref{eq:ds} is required (at most) $K_\xi + 1$ times with an appropriately chosen threshold $\pi_{n, p}$ (see Theorem~\ref{thm:idio} below).
We further demonstrate numerically the competitiveness of Stage~2 as a standalone method for VAR time series segmentation in Section~\ref{sec:sim}, and provide an in-depth comparative study with the existing methods in Appendix~\ref{sec:comparison}.

\section{Theoretical properties}
\label{sec:theory}


\subsection{Consistency of Stage~1 of FVARseg}

We carry out our theoretical investigation under two different regimes with respect to the tail behaviour of $\mbf u_t$ and $\bm\xi_t$; in particular, the weaker condition in Assumption~\ref{assum:innov}~\ref{cond:moment} permits heavy-tailed innovations, while the existing literature on (piecewise stationary) VAR modelling in high dimensions, commonly adopts the Gaussianity as in~\ref{cond:gauss}.

\begin{assum} 
\label{assum:innov} We assume {\it either} of the following conditions.
\vspace{-10pt}
\begin{enumerate}[noitemsep, wide, labelwidth=0pt, labelindent=0pt, label = (\roman*)] 
\item \label{cond:moment} There exists $\nu > 4$ such that
$\max\l\{ \E(\vert u_{jt} \vert^\nu), 
\E(\vert \vep_{it} \vert^\nu) \r\} \le \mu_\nu < \infty$.

\item \label{cond:gauss} $\mbf u_t \sim_{\iid} \mc N_q(\mbf 0, \mbf I)$ and 
$\bm\vep_t \sim_{\iid} \mc N_p(\mbf 0, \mbf I)$.
\end{enumerate}
\end{assum}

In establishing the consistency of Stage~1, 
we opt to measure the size of changes in $\bm\chi_t$ using ${\bm\Delta}_{\chi, k}(\omega) = \bm\Sigma_\chi^{\k}(\omega) - \bm\Sigma_\chi^{[k - 1]}(\omega)$, $1 \le k \le K_\chi$, the difference in spectral density matrices of $\bm\chi_t$ from neighbouring segments.
As ${\bm\Delta}_{\chi, k}(\omega)$ is Hermitian,
we can always find the $j$th largest (in modulus), real-valued eigenvalue of ${\bm\Delta}_{\chi, k}(\omega)$
which we denote by $\mu_j({\bm\Delta}_{\chi, k}(\omega))$, with $\mu_1({\bm\Delta}_{\chi, k}(\omega)) = \Vert {\bm\Delta}_{\chi, k}(\omega) \Vert$.
Recall that $m = G^\beta$ for some $\beta \in (0, 1)$, denotes the bandwidth used in local spectral density estimation, see~\eqref{eq:local:spec}.

\begin{assum}
\label{assum:common:size}
\begin{enumerate}[noitemsep, wide, labelwidth=0pt, labelindent=0pt, label = (\roman*)]
\item \label{cond:common:jump}
For each $1 \le k \le K_\chi$, the following holds:
There exist a positive integer $p_0 \ge 1$
and pairs of functions
$\omega \mapsto a^{\k}_j(\omega)$
and $\omega \mapsto b^{\k}_j(\omega)$
for $\omega \in [-\pi, \pi]$ and $j = 1, 2$,
and $r^\prime_{k, 1} \in (0, 1]$ and $r^\prime_{k, 2} \in [0, 1]$ 
satisfying $r^\prime_{k, 1} \ge r^\prime_{k, 2}$, 
such that 
\begin{align*}
& b^{\k}_1(\omega) \ge \frac{\mu_1({\bm\Delta}_{\chi, k}(\omega))}{p^{r^\prime_{k, 1}}} \ge a^{\k}_1(\omega) >
b^{\k}_2(\omega) \ge \frac{\mu_2({\bm\Delta}_{\chi, k}(\omega))}{p^{r^\prime_{k, 2}}} \ge a^{\k}_2(\omega) \ge 0
\end{align*}
for all $p \ge p_0$.
Besides, we assume that the functions $\omega \mapsto p^{-r^\prime_{k, 1}} \mu_1({\bm\Delta}_{\chi, k}(\omega))$
are Lipschitz continuous with bounded Lipschitz constants.
Then for $\Delta_{\chi, k} = \max_{\omega \in [-\pi, \pi]} \mu_1({\bm\Delta}_{\chi, k}(\omega))$, we have
$\max_{1 \le k \le K_\chi} \Delta_{\chi, k}^{-1} \cdot p(\psi_n \vee m^{-1}) = o(1)$, where 
\begin{align}
& \psi_n = \l\{\begin{array}{ll}
\frac{n^{2/\nu} m \log^{2 + 2/\nu}(G)}{G} \vee \sqrt{\frac{m\log(n)}{G}} & \text{under Assumption~\ref{assum:innov}~\ref{cond:moment}},
\\
\sqrt{\frac{m\log(n)}{G}} & \text{under Assumption~\ref{assum:innov}~\ref{cond:gauss}.}
\end{array}\r.
\label{eq:psi}
\end{align}

\item \label{cond:common:spacing} 
The bandwidth $G = G_n$ satisfies $G_n \to \infty$ as $n \to \infty$ while fulfilling 
\begin{align}
\label{eq:min:spacing}
\min\l\{\min_{0 \le k \le K_\chi} (\cp_{\chi, k + 1} - \cp_{\chi, k}), 
\min_{0 \le k \le K_\xi} (\cp_{\xi, k + 1} - \cp_{\xi, k}) \r\} \ge 2G.
\end{align}
\end{enumerate}
\end{assum}

Assumption~\ref{assum:common:size} specifies the detection lower bound which is determined by $\min_k \Delta_{\chi, k}$ and $\min_k(\cp_{\chi, k + 1} - \cp_{\chi, k})$ (through $G$),
for all $K_\chi$ change points $\bm\chi_t$ to be detectable by Stage~1.
Condition~\ref{cond:common:jump} requires $\mu_1({\bm\Delta}_{\chi, k}(\omega))$
to be distinct from the rest.
In fact, the remaining $\mu_j(\bm\Delta_{\chi, k}(\omega)), \, j \ge 2$, are allowed to be exactly zero, which is the case in Example~\ref{ex:common}; 
here, we have $\Delta_{\chi, 1} = \max_\omega (2\pi)^{-1} \Vert (\mbf b(e^{-\iota \omega}))^* \mbf b(e^{-\iota \omega}) \Vert$ where $\mbf b(z)$ is a $p$-variate vector of factor loading filters.
The rate $p(\psi_n \vee m^{-1})$ represents the bias-variance trade-off when estimating the local spectral density matrix of $\bm\chi_t$ by $\wh{\bm\Sigma}_{x, v}(G, \omega)$ (see Proposition~\ref{prop:loc:spec}).
It is possible to find the rate of kernel bandwidth $m$ that minimises this rate depending on the tail behaviour of $X_{it}$ (e.g.\ $m \asymp (G/\log(n))^{1/3}$ under Gaussianity),
but we choose to explicitly highlight the role of this tuning parameter on our results. 


\begin{thm}
\label{thm:common}
Suppose that Assumptions~\ref{assum:factor}--\ref{assum:common}, \ref{assum:innov} and~\ref{assum:common:size} hold.
Let $\kappa_{n, p}$ satisfy 
\begin{align*}
2 M p\l( \psi_n \vee  \frac{1}{m} \vee \frac{1}{p} \r) < \kappa_{n, p} < 
\frac{1}{2} \min_{1 \le k \le K_\chi} \Delta_{\chi, k} - M p\l( \psi_n \vee  \frac{1}{m} \vee \frac{1}{p} \r)
\end{align*}
for some constant $M > 0$. Then, 
there exists a set $\mc M^\chi_{n, p}$ with $\p(\mc M^\chi_{n, p}) \to 1$ as $n, p \to \infty$,
such that the following holds for $\wh\Cp_\chi = \{\wh\cp_{\chi, k}, \, 1 \le k \le \wh K_\chi: \,
\wh\cp_{\chi, 1} < \ldots < \wh\cp_{\chi, \wh K_\chi}\}$
returned by Stage~1 of FVARseg,
on $\mc M^\chi_{n, p}$ for large enough $n$ and $p$:
\begin{enumerate}[noitemsep, wide, labelwidth=0pt, labelindent=0pt, label = (\alph*)]
\item \label{thm:common:one}
$\wh K_\chi = K_\chi$ and 
$\max_{1 \le k \le K_\chi} \vert \wh\cp_{\chi, k} - \cp_{\chi, k} \vert \le \epsilon_0 G$
for some $\epsilon_0 \in (0, 1/2)$ with $\eta \in (2\epsilon_0, 1]$.

\item \label{thm:common:two}
There exists a constant $c_0 > 0$ such that for all $1 \le k \le K_\chi$, $\vert \wh\cp_{\chi, k} - \cp_{\chi, k} \vert \le c_0\rho^{\k}_{n, p}$~where
\begin{align*}
& \rho^{\k}_{n, p} = 
\l(\frac{\Delta_{\chi, k}}{p}\r)^{-2} \times
\l\{\begin{array}{ll}
m^{\frac{\nu}{\nu - 2}} (G K_\chi)^{\frac{2}{\nu - 2}} & 
\text{under Assumption~\ref{assum:innov}~\ref{cond:moment}},
\\
m\log(G K_\chi) & \text{under Assumption~\ref{assum:innov}~\ref{cond:gauss}}.
\end{array}\r.
\end{align*}
\end{enumerate}
\end{thm}

\begin{rem}
\label{rem:thm:common}
\begin{enumerate}[noitemsep, wide, labelwidth=0pt, labelindent=0pt, label = (\roman*)]

\item \label{rem:thm:common:two}
In Theorem~\ref{thm:common}~\ref{thm:common:two},
$\rho^{\k}_{n, p}$ reflects the difficulty associated with estimating the individual change point $\cp_{\chi, k}$
manifested by $(p^{-1}\Delta_{\chi, k})^{-2}$.
In the Gaussian case (Assumption~\ref{assum:innov}~\ref{cond:gauss}),
the localisation rate $\rho^{\k}_{n, p}$ is always sharper than $G$
due to Assumption~\ref{assum:common:size}~\ref{cond:common:jump}.
Considering the problem of covariance change point detection
in independent, sub-Gaussian random vectors in high dimensions,
\cite{wang2021optimal} derive the minimax lower bound on the localisation rate in their Lemma~3.2, 
and $\rho^{\k}_{n, p}$ matches this rate up to $m \log(n)$;
here, the dependence on the kernel bandwidth $m$ is attributed to that we consider a time series segmentation problem, i.e.\ a change may occur in the ACV of $\bm\chi_t$ at lags other than zero.
If heavier tails are permitted (Assumption~\ref{assum:innov}~\ref{cond:moment}), 
$\rho^{\k}_{n, p}$ can be tighter than $\epsilon_0 G$,
e.g.\ when $\Delta_{\chi, k} \asymp p$, $K_\chi$ is fixed and 
$m \asymp G^\beta$ for some $\beta \in (0,  1 - 4/\nu)$.

\item Empirically, replacing $\wh\cp$ with $\wt\cp = \arg\max_{v \in \mc I} \, \text{avg}_l T_{\chi, v}(\omega_l, G)$ returns a more stable location estimator, where $\text{avg}_l$ denotes the average operator over $l = 0, \ldots, m$.
We can derive the localisation rate for $\wt\cp$ similarly as in Theorem~\ref{thm:common}~\ref{thm:common:two} with $\wt\Delta_{\chi, k} = \pi^{-1} \int_0^{\pi} \Vert {\bm\Delta}_{\chi, k}(\omega) \Vert d\omega$ in place of $\Delta_{\chi, k}$. 
Our numerical results in Section~\ref{sec:sim} are based on this estimator.

\end{enumerate}
\end{rem}

Next, we establish the consistency of $\wh{\bm\Gamma}^{\k}_\chi(\ell)$ in~\eqref{eq:common:est} estimating the segment-specific ACV of $\bm\chi^{\k}_t$ under the following assumption on the strength of factors.
\begin{assum}
\label{assum:factor:two}
Assumption~\ref{assum:factor} holds with $r_{k, j} = 1$ for all $1 \le j \le q_k$ and $0 \le k \le K_\chi$.
\end{assum}

\begin{thm}
\label{thm:common:est}
Suppose that Assumption~\ref{assum:factor:two} holds
in addition to the assumptions made in Theorem~\ref{thm:common},
and define $\rho_{n, p} = \max_{1 \le k \le K_\chi} \min(\epsilon_0 G, \rho^{\k}_{n, p})$.
Also let
\begin{align}
\vartheta_{n, p} &= \l\{\begin{array}{ll}
\frac{m (np)^{2/\nu} \log^{7/2}(p)}{G} 
\vee \sqrt{\frac{m\log(np)}{G}} 
& \text{under Assumption~\ref{assum:innov}~\ref{cond:moment}},
\\
\sqrt{\frac{m\log(np)}{G}} & 
\text{under Assumption~\ref{assum:innov}~\ref{cond:gauss}}.
\end{array}\r.
\nn 
\end{align}
Then on $\mc M^\chi_{n, p}$ defined in Theorem~\ref{thm:common}, 
for some finite integer $d \in \N$, we have 
\begin{align*}
\max_{0 \le k \le K_\chi} \max_{0 \le \ell \le d} \l\vert \wh{\bm\Gamma}^{\k}_\chi(\ell) - \bm\Gamma^{\k}_\chi(\ell) \r\vert_\infty
= O_p\l( \vartheta_{n, p} \vee \frac{1}{m} \vee \frac{\rho_{n, p}}{G} \vee \frac{1}{\sqrt p}\r).
\end{align*}
\end{thm}
It is possible to work under the weaker Assumption~\ref{assum:factor} and trace the effect of weak factors or bound estimation errors measured in different norms.
Corollary~C.16 of \cite{barigozzi2022fnets} derives such results in the stationary setting, where an additional multiplicative factor of $p^{2(1 - \min_k r_{k, q_k})}$ appears in the $O_P$-bound in Theorem~\eqref{thm:common:est}.
We work under the stronger Assumption~\ref{assum:factor:two} as it simplifies the presentation of Theorem~\ref{thm:common:est} which plays an important role in the investigation into Stage~2 of FVARseg, and since only Assumption~\ref{assum:factor:two} is compatible with 
the cross-sectional ordering often being completely arbitrary.

\subsection{Consistency of Stage~2 of FVARseg}
\label{sec:idio:theory}

Suppose that the tuning parameter for the $\ell_1$-regularised Yule-Walker estimation problem in~\eqref{eq:ds}, is set with some constant $M > 0$ and $\vartheta_{n, p}$ and $\rho_{n, p}$ defined in Theorem~\ref{thm:common:est}, as
\begin{align}
\lambda_{n, p} &= M\l(\max_{0 \le k \le K_\xi} \Vert \bm\beta^{\k} \Vert_1 + 1\r)
\l(\vartheta_{n, p} \vee \frac{1}{m} \vee \frac{\rho_{n, p}}{G} \vee \frac{1}{\sqrt p}\r).
\label{eq:lambda}
\end{align}
This choice reflects the error in $\wh{\bm\Gamma}_{\xi, v}(\ell, G)$ estimating the local ACV of $\bm\xi_t$ over all $v$ and $\ell$.

The following assumption imposes conditions on the size of the changes in VAR parameters
and the minimum spacing between the change points.

\begin{assum}
\label{assum:idio:size}
\begin{enumerate}[noitemsep, wide, labelwidth=0pt, labelindent=0pt, label = (\roman*)]
\item \label{cond:idio:jump}
For each $1 \le k \le K_\xi$, 
let ${\bm\Delta}_{\xi, k} = \bbG^{\k}(\bm\beta^{\k} - \bm\beta^{[k - 1]})$.
Then, 
\begin{align*}
\max_{1 \le k \le K_\xi} 
\frac{(1 \vee \Vert \bbG^{\k} (\bbG^{[k - 1]})^{-1} \Vert_1) \lambda_{n, p}}{\vert {\bm\Delta}_{\xi, k} \vert_\infty} 
= o(1).
\end{align*}


\item \label{cond:idio:spacing} 
The bandwidth $G$ fulfils~\eqref{eq:min:spacing}, i.e.\
$\min_{0 \le k \le K_\xi} (\cp_{\xi, k + 1} - \cp_{\xi, k}) \ge 2G$.
\end{enumerate}
\end{assum}

\begin{rem}
\label{rem:cond:idio}
We choose to measure the size of change using $\vert \bm\Delta_{\xi, k} \vert_\infty$.
From Assumption~\ref{assum:idio}~\ref{cond:idio:minspec}, we have ${\bm\Delta}_{\xi, k} = \mbf O$ iff $\bm\beta^{\k} - \bm\beta^{[k - 1]} = \mbf O$.
In the related literature, the $\ell_2$-norm $\vert \bm\beta^{\k} - \bm\beta^{[k - 1]} \vert_2$ scaled by the global sparsity (given by the union of the supports of all $\bm\beta^{\k}, \, 0 \le k \le K_\xi$), is used to measure the size of change where this global sparsity may be much greater than that of $\bm\Delta_{\xi, k}$ when $K_\xi$ is large, see Appedix~\ref{sec:comparison}.
In some instances, we have $\bbG^{\k} (\bbG^{[k - 1]})^{-1} = \mbf I$, e.g.\ when $d = 1$ and $\mbf A_1^{\k} = - \mbf A_1^{[k - 1]}$ such that Assumption~\ref{assum:idio:size}~\ref{cond:idio:jump} becomes $\lambda_{n, p} = o(\min_k \vert \bm\Delta \vert_\infty)$.
More generally, bounding $\Vert \bbG^{\k} (\bbG^{[k - 1]})^{-1} \Vert_1$ implicitly assumes (approximate) sparsity on the second-order structure of $\bm\xi_t$.
When $d = 1$, we have $\bbG^{\k} = \sum_{\ell = 0}^\infty (\mbf A^{\k}_1)^\ell \bm\Gamma^{\k} [(\mbf A^{\k}_1)^\top]^\ell$ such that the boundedness of $\Vert \bbG^{\k} \Vert_1$ and $\Vert (\bbG^{\k})^{-1} \Vert_1$ follows when $\mbf A^{\k}_1$ and $\bm\Gamma^{\k}$ are block diagonal with fixed block size \citep{wang2021robust}.
For general $d \ge 1$, we have $\Vert \bbG^{\k} (\bbG^{[k - 1]})^{-1} \Vert_1$ bounded if $\bbG^{\k}$ are strictly diagonally dominant (see Definition~6.1.9 of \cite{horn1985} and \cite{han2015direct}), which is met e.g.\ when $\mbf A^{\k}_\ell$ are diagonal with their diagonal entries fulfilling $\gamma^{\k}_{\xi, ii}(0) > 2 \sum_{\ell = 1}^{d - 1} \vert \gamma^{\k}_{\xi, ii}(\ell) \vert$ (where $\bm\Gamma^{\k}_\xi(\ell) = [\gamma^{\k}_{\xi, ii'}(\ell)]_{i, i'}$); this trivially holds when $d = 1$.
\end{rem}

\begin{thm}
\label{thm:idio}
Suppose that Assumption~\ref{assum:idio:size} holds
in addition to the assumptions made in Theorem~\ref{thm:common:est}.
With $\lambda_{n, p}$ chosen as in~\eqref{eq:lambda}, we set $\pi_{n, p}$ to satisfy
\begin{align}
& 2\lambda_{n, p} < \pi_{n, p} < \frac{1}{2} \min_{1 \le k \le K_\xi} \vert {\bm\Delta}_{\xi, k} \vert_\infty.
\nn
\end{align}
Then, there exists a set $\mc M^\xi_{n, p}$ with $\p(\mc M^\xi_{n, p}) \to 1$ as $n, p \to \infty$,
such that the following holds for $\wh\Cp_\xi = \{\wh\cp_{\xi, k}, \, 1 \le k \le \wh K_\xi: \,
\wh\cp_{\xi, 1} < \ldots < \wh\cp_{\xi, \wh K_\xi}\}$
returned by Stage~2 of FVARseg,
on $\mc M^\xi_{n, p}$ for large enough $n$:
\begin{enumerate}[noitemsep, wide, labelwidth=0pt, labelindent=0pt, label = (\alph*)]
\item \label{thm:idio:one} $\wh K_\xi = K_\xi$ and 
$\max_{1 \le k \le K_\xi} \vert \wh\cp_{\xi, k} - \cp_{\xi, k} \vert \le \epsilon_0 G$ 
for some $\epsilon_0 \in (0, 1/2)$ with $\eta \in (\epsilon_0, 1]$.

\item \label{thm:idio:two} There exists a constant $c_0 > 0$ such that
for all $1 \le k \le K_\xi$ satisfying
$\{\cp_{\xi, k} - 2G + 1, \ldots, \cp_{\xi, k} + 2G\}  \cap \Cp_\chi = \emptyset$,
we have $\vert \wh\cp_{\xi, k} - \cp_{\xi, k} \vert \le c_0 \varrho^{\k}_{n, p}$, where
\begin{align*}
& \varrho^{\k}_{n, p} = \vert {\bm\Delta}_{\xi, k} \vert_\infty^{-2} \l(1 + \max_{0 \le k \le K_\xi} \Vert \bm\beta^{\k} \Vert_1 \r) \times
\l\{\begin{array}{l}
(G K_\xi p)^{\frac{2}{\nu - 2}} \log^{\frac{3\nu}{\nu - 2}}(p) 
\\
\qquad \text{under Assumption~\ref{assum:innov}~\ref{cond:moment}},
\\
\log(G K_\xi p) \\
\qquad \text{under Assumption~\ref{assum:innov}~\ref{cond:gauss}}.
\end{array}\r.
\end{align*}
\end{enumerate}
\end{thm}

Due to the sequential nature of FVARseg, the success of Stage~2 is conditional on that of Stage~1 which occurs on an asymptotic one-set, see Theorem~\ref{thm:common}.
Theorem~\ref{thm:idio}~\ref{thm:idio:one} establishes that Stage~2 of FVARseg consistently detects all $K_\xi$ change points within the distance of $\epsilon_0 G$ where 
$\epsilon_0$ can be made arbitrarily small as $n, p \to \infty$ under Assumption~\ref{assum:idio:size}~\ref{cond:idio:jump}.
Theorem~\ref{thm:idio}~\ref{thm:idio:two} shows that a further refined localisation rate can be derived for $\cp_{\xi, k}$ when it is sufficiently distanced away from the change points in the factor-driven component.
If, say, $\cp_{\xi, k}$ lies close to $\cp_{\chi, k'}$, a change point in $\bm\chi_t$, the error from estimating the local ACV of $\bm\xi_t$ due to the bias in $\wh\cp_{\chi, k'}$, prevents applying the arguments involved in the refinement to such $\cp_{\xi, k}$.
The refined rate $\varrho^{\k}_{n, p}$ is always tighter than $G$ under Gaussianity.

It is of independent interest to consider the cases where $\bm\chi_t$ is stationary (i.e.\ $K_\chi = 0$) or where we directly observe the piecewise stationary VAR process (i.e.\ $\mbf X_t = \bm\xi_t$). 
Consistency of the Stage~2 of FVARseg readily extends to such settings and the improved localisation rates in Theorem~\ref{thm:idio}~\ref{thm:idio:two} apply to {\it all} the estimators.
Also, further improvement is attained in the heavy-tailed situations (Assumption~\ref{assum:innov}~\ref{cond:moment}) if $\bm\xi_t$ is directly observable.
For the full statement of the results, we refer to Corollary~\ref{cor:idio} in Appendix~\ref{sec:idio:ext} where we also provide a detailed comparison between Stage~2 of FVARseg and existing VAR segmentation methods (that do not take into the possible presence of factors), both theoretically and numerically.

\section{Empirical results}
\label{sec:numeric}


\subsection{Numerical considerations}
\label{sec:numeric:detail}

\paragraph{Multiscale extension.} 
The bandwidth $G$ is required to be large enough to provide a good local estimators of spectral density of $\bm\chi_t$ (Stage~1) and VAR parameters (Stage~2).
However, if $G$ is too large, we may have windows that contain two or more changes when scanning the data for change points, which violates Assumptions~\ref{assum:common:size}~\ref{cond:common:spacing} and~\ref{assum:idio:size}~\ref{cond:idio:spacing}.
\cite{cho2020} note the lack of adaptivity of a single-bandwidth moving window procedure in the presence of multiscale change points (a mixture of large changes over short intervals
and smaller changes over long intervals), and advocates the use of multiple bandwidths.
Accordingly we also propose to apply FVARseg with a range of bandwidths and prune down the outputs using a `bottom-up' method \citep{messer2014, meier2018}.
Let $\wh\Cp(G)$ denote the output from Stage~1 or~2 with a bandwidth $G$. 
Given a set of bandwidths $\mc G = \{G_h, \, 1 \le h \le H: \, G_1 < \ldots < G_H\}$,
we accept all estimators from the finest $G_1$ to the set of final estimators $\wh\Cp$
and sequentially for $h \ge 2$, accept $\wh\cp \in \wh\Cp(G_h)$ iff $\min_{\check\cp \in \wh\Cp} \vert \wh\cp - \check\cp \vert \ge G/2$. 
In simulation studies, we use $\mc G_\chi = \{ [n/10], [n/8], [n/6], [n/4] \}$ for Stage~1, and $\mc G_\xi$ generated as an equispaced sequence between $[2.5p]$ and $[n/4]$ of length~$4$ for Stage~2.
The choice of $\mc G_\xi$ is motivated by the simulation results of \cite{barigozzi2022fnets}
under the stationarity,  where the $\ell_1$-regularised estimator in~\eqref{eq:ds} was observed to performs well when the sample size exceeds $2p$. 
\vspace{-10pt}

\paragraph{Speeding up Stage~1.}
The computational bottleneck of FVARseg is the computation of $T_{\chi, v}(\omega_l, G)$ in Stage~1, which involves singular value decomposition (SVD) of a $p \times p$-matrix
at multiple frequencies and over time.
We propose to evaluate $T_{\chi, v}(\omega_l, G)$ on a grid $v \in \{ G + a b_n: \, 0 \le a \le \lfloor (n - 2G)/b_n \rfloor\}$ with $b_n = \lfloor 2\log(n) \rfloor$.
This may incur additional bias of at most $b_n/2 \le \log(n)$ in change point location estimation
which is asymptotically negligible in view of Theorem~\ref{thm:common},
but reduce the computational load by the factor of $b_n$.
\vspace{-10pt}

\paragraph{Selection of thresholds.}
The theoretically permitted ranges of $\kappa_{n, p}$ and $\pi_{n, p}$ (see Theorems~\ref{thm:common} and~\ref{thm:idio}) depend on constants which are not accessible or difficult to estimate in practice.
This is an issue commonly encountered by data segmentation methods which involve localised testing, and often a reasonable solution is found by large-scale simulations, an approach we also take.
We use simulations to derive a simple rule for selecting the threshold as a function of $n$, $p$ and $G$.
For this, we (i)~propose a scaling for each of the two detector statistics adopted in Stages~1 and~2 which reduces its dependence on the data generating process, and (ii)~fit a linear model for an appropriate percentile of the scaled detector statistics obtained from simulated datasets.
Specifically, we simulate $B = 100$ time series following~\eqref{eq:model} with $K_\chi = K_\xi = 0$
using the models considered in Section~\ref{sec:sim}, and record the maximum of the scaled detector statistics $T^\circ_{\chi, v}(G)$ and $T^\circ_{\xi, v}(G)$ over $v$ on each realisation.
Here, the scaling terms are obtained from the first $G$ observations only, as

\begin{align*}
T^\circ_{\chi, v}(G) = \max_{0 \le l \le m} \frac{T_{\chi, v}(\omega_l, G)}{T_{\chi, G}(\omega_l, G)} \text{ \ and \ }
T^\circ_{\xi, v}(G) = \frac{T_{\xi, v}(\wh{\bm\beta}_G(G), G)}{\max_{0 \le \ell \le d} \vert \wh{\bm\Gamma}_{\xi, [\frac{G}{2}]}(\ell, [\frac{G}{2}]) - \wh{\bm\Gamma}_{\xi, G}(\ell, [\frac{G}{2}]) \vert_\infty}.
\end{align*}
Generating the data with varying $(n, p, q, d)$ and repeating the above procedure with multple choices of $G$, we fit a linear model to the $100(1 - \tau)$th percentile of
$\log(\max_v T^\circ_{\chi, v}(G))$ with $\log\log(n)$ and $\log(G)$ as regressors ($R^2_{\text{adj}} = 0.9651$),
and use the fitted model to derive a threshold for given $n$ and $G$ that is then applied to the similarly scaled $T^\circ_{\chi, v}(\omega_l, G)$.
Analogously, we regress the $100(1 - \tau)$th percentile of $\log(\max_v T^\circ_{\xi, v}(G))$ onto $\log\log(n)$, $\log\log(p)$ and $\log(G)$ ($R^2_{\text{adj}} = 0.985$), and find a threshold applied to the scaled $T^\circ_{\xi, v}(\wh{\bm\beta}, G)$ given $n$, $p$ and $G$ from the fitted model.
The choice of the regressors is motivated by the definitions of $\psi_n$ and $\psi_{n, p}$ which appear in Theorems~\ref{thm:common} and~\ref{thm:idio}.
The high values of $R^2_{\text{adj}}$ indicate the excellent fit of the linear models and consequently, that the threshold selection rule is insensitive to the data generating processes.
When Stage~2 is used as a standalone method for segmenting observed VAR processes, a smaller threshold is recommended which is in line with Corollary~\ref{cor:idio}, and we find that $\pi_{n, p} = 1$ works well with the proposed scaling.
\vspace{-10pt}

\paragraph{Other tuning parameters.} 
While data-adaptive methods exist for selecting the kernel window size $m$ in~\eqref{eq:local:spec} \citep{politis2003adaptive}, we find that setting it simply at $m = \max(1, \lfloor G^{1/3} \rfloor)$ for given $G$, works well for the purpose of data segmentation.
The results are not highly sensitive to the choice of $\eta$ in Stage~1 and use $\eta = 0.5$ throughout. 
In Stage~2, we find that not trimming off the data when estimating the VAR parameters by setting $\eta = 0$, does not hurt the numerical performance.
In factor-adjustment, we select the segment-specific factor number $q_k$ using the IC-based approach of \cite{hallin2007}.
\cite{krampe2021} propose to jointly select the (static) factor number and the VAR order using an IC but generally, the validity of IC is not well-understood for VAR order selection in high dimensions. 
In our simulations, following the practice in the literature on VAR segmentation, we regard $d$ as known but also investigate the sensitivity of FVARseg when $d$ is mis-specified.
In analysing the panel of daily volatilities (Section~\ref{sec:app}), we use $d = 5$ which has the interpretation of the number of trading days per week.
Finally, we select $\lambda_{n, p}$ in~\eqref{eq:ds} via cross validation as in \cite{barigozzi2022fnets}.

\subsection{Simulation studies}
\label{sec:sim}

In the simulations, we consider the cases when the factor-driven component is present ($\bm\chi_t \ne 0$) and when it is not ($\bm\chi_t = 0$).
For the former, we consider two models for generating $\bm\chi_t$ with $q = 2$.
In the first model, referred to as~\ref{c:one}, $\bm\chi_t$ admits a static factor model representation while in the second model~\ref{c:two}, it does not;
empirically, the task of factor structure estimation is observed to be more challenging under~\ref{c:two} \citep{forni2017dynamic, barigozzi2022fnets}.
We generate $\bm\xi_t$ as piecewise stationary Gaussian VAR($d$) processes with $d \in \{1, 2\}$ and a parameter $\beta$ that controls the size of the change (with smaller $\beta$ indicating the smaller change).
We refer to Appendix~\ref{sec:common:dgp} for the full descriptions of simulation models and Table~\ref{table:sim} for an overview of the $24$ data generating processes which also contains information about the sets of change points $\Cp_\chi$ and $\Cp_\xi$;
under each setting, we generate $100$ realisations.
Below we provide a summary of the findings from the simulation studies, and Tables~\ref{Tab:com1}--\ref{Tab:oracle} reporting the results can be found in Appendix~\ref{sec:sim:res}.

\begin{table}[ht] 
\caption{\small Data generating processes for simulation studies.}
\label{table:sim}
\centering
\resizebox{\columnwidth}{!}{\footnotesize
\begin{tabular}{c cccccc} 
\toprule
& $n$ & $p$ & $\bm\chi_t$ & $(d, \beta)$ & $\Cp_\chi$ & $\Cp_\xi$ \\
\cmidrule{1-1}\cmidrule(lr){2-2}\cmidrule(lr){3-3}\cmidrule(lr){4-4}\cmidrule(lr){5-5}\cmidrule(lr){6-6}\cmidrule(lr){7-7}
(M1) & $2000$ & $50, 100, 150$  & \ref{c:one} & $(1, 1)$ &
$\emptyset$, $\{[n/4], [n/2], [3n/4]\}$ & $\{[3n/8], [5n/8]\}$ \\
(M2) & $2000$ & $50, 100, 150$  & \ref{c:two} & $(1, 1)$ & 
$\emptyset$, $\{[n/3], [2n/3]\}$ & $\{[n/3], [2n/3]\}$ \\
(M3) & $2000$ & $50, 100, 150$ & $\mbf 0$ & $(1, 0.6)$, $(2, 0.8)$ & 
$\emptyset$ & $\emptyset$, $\{[3n/8], [5n/8]\}$ \\
\bottomrule
\end{tabular}} 
\end{table} 
To the best of our knowledge, there does not exist a methodology that comprehensively addresses the change point problem under the model~\eqref{eq:model}.
Therefore under (M1)--(M2), we compare the Stage~1 of FVARseg with a method proposed in \cite{barigozzi2018}, referred to as BCF hereafter, on their performance at detecting changes in $\bm\chi_t$.
While BCF has a step for detecting change points in the remainder component, it does so nonparametically unlike the Stage~2 of FVARseg, which may lead to unfair comparison.
Hence we separately consider (M3) with $\mbf X_t = \bm\xi_t$ where we compare the Stage~2 method with VARDetect \citep{safikhani2020b}, a block-wise variant of \cite{safikhani2020}.
\vspace{-10pt}

\paragraph{Results under (M1)--(M2).} Overall, FVARseg achieves good accuracy in estimating the total number and locations of the change points for both $\bm\chi_t$ and $\bm\xi_t$ across different data generating processes. 
Under (M1) adopting the static factor model for generating $\bm\chi_t$, FVARseg shows similar performance as BCF in detecting $\Cp_\chi$ when the dimension is small ($p = 50$), 
but the latter tends to over-estimate the number of change points as $p$ increases.
Also, FVARseg outperforms the binary segmentation-based BCF
in change point localisation.
BCF requires as an input the upper bound on the number of global factors, say $q'$, that includes the ones attributed to the change points, and its performance is sensitive to its choice.
In (M1), we have $q' \le 3 q (K_\chi + 1)$ (which is supplied to BCF) while in (M2), $\bm\chi^{\k}_t$ does not admit a static factor representation and 
accordingly such $q'$ does not exist (we set $q' = 2q$ for BCF).
Accordingly, BCF tends to under-estimate the number of change points under (M2).
Generally, the task of detecting change points in $\bm\xi_t$ is aggravated by the presence of change points in $\bm\chi_t$ due to the sequential nature of FVARseg, and the Stage~2 performs better when $K_\chi = 0$ both in terms of detection and localisation accuracy, which agrees with the observations made in Corollary~\ref{cor:idio}~\ref{cor:idio:station}.

Between~(M1) and~(M2), the latter poses a more challenging setting for the Stage~2 methodology.
This may be attributed to (i)~the difficulty posed by the data generating scenario~\ref{c:two}, which is observed to make the estimation tasks related to the latent VAR process more difficult \citep{barigozzi2022fnets}, and (ii)~that $\Cp_\chi = \Cp_\xi$ where the estimation bias from Stage~1 has a worse effect on the performance of Stage~2 compared to when $\Cp_\chi$ and $\Cp_\xi$ do not overlap, see the discussion below Theorem~\ref{thm:idio}.
\vspace{-10pt}

\paragraph{Results under (M3).} Table~\ref{Tab:oracle} shows that the Stage~2 of FVARseg outperforms VARDetect in all criteria considered, particularly as $p$ increases.
VARDetect struggles to detect any change point  when the change is weak (recall that $\beta = 0.6$ is used when $d = 1$ which makes the size of change at $\cp_{\xi, 2}$ small) or when $d = 2$.
FVARseg is faster than VARDetect in most situations except for when $(d, p, K_\xi) = (1, 50, 0)$, sometimes more than $10$ times e.g.\ when $d = 2$ and there is no change point in the data. 
Additionally, Stage~2 of FVARseg is insensitive to the over-specification of the VAR order ($d = 2$ is used when in fact $d = 1$). When it is under-specified, there is slight loss of detection power as expected.
Compared to the results obtained under (M1)--(M2), the localisation performance of the Stage~2 method improves in the absence of the factor-driven component, even though the size of changes under (M3) tends to be smaller. 
This confirms the theoretical findings reported in Corollary~\ref{cor:idio}~\ref{cor:idio:no} in Appendix~\ref{sec:idio:ext}.
Although not reported here, when the full FVARseg methodology is applied to the data generated under~(M3), the Stage~1 method does not detect any spurious change point estimators as desired.

\subsection{Application: US blue chip data}
\label{sec:app}

We consider daily stock prices from $p = 72$ US blue chip companies 
across industry sectors
between January 3, 2000 and February 16, 2022 ($n = 5568$ days),
retrieved from the Wharton Research Data Services;
the list of companies and their corresponding sectors can be found in Appendix~\ref{app:ticker}.
Following \cite{diebold2014network}, we measure the volatility using 
$\sigma_{it}^2 = 0.361 (p_{it}^{\text{high}} - p_{it}^{\text{low}})^2$
where $p_{it}^{\text{high}}$ (resp.\ $p_{it}^{\text{low}}$) denotes the maximum (resp.\ minimum) log-price of stock $i$ on day $t$, and set $X_{it} = \log(\sigma_{it}^2)$.

We apply FVARseg to detect change points in the panel of volatility measures $\{X_{it}, \, 1 \le i \le p; \, 1 \le t \le n\}$.
With $n_0 = 252$ denoting the number of trading days per year,
we apply Stage~1 with bandwidths chosen as an equispaced sequence between $[n_0/4]$ and $2n_0$ of length $4$, implicitly setting the minimum distance between two neighbouring change points to be three months.
Based on the empirical sample size requirement for VAR parameter estimation (see Section~\ref{sec:numeric:detail}), we apply Stage~2 with bandwidths chosen as an equispaced sequence between $2.5p$ and $2n_0$ of length~$4$.
The VAR order is set at $d = 5$ which corresponds to the number of trading days in each week,
and the rest of the tuning parameters are selected as in Section~\ref{sec:numeric:detail}.
Table~\ref{table:real} reports the segmentation results.

\begin{table}[h!t!b!]
\caption{\small Sets of change point estimators returned by FVARseg.}
\label{table:real}
\resizebox{\columnwidth}{!}{
\scriptsize
\begin{tabular}{cc cc ccc ccc}
\toprule
\multicolumn{4}{c}{$\wh{\Cp}_\chi$ returned by Stage~1} & \multicolumn{6}{c}{$\wh{\Cp}_\xi$ returned by Stage~2} \\
\cmidrule(lr){1-4} \cmidrule(lr){5-10}
2002-06-06 & 2007-12-10 & 2008-09-12 & 2008-12-16 &
2002-02-04 & 2003-03-18 & 2003-11-25 & 2006-06-07 & 2008-03-17 & 2009-07-07 \\
2009-05-11 & 2020-02-18 & 2020-05-20 & &
2011-07-28 & 2013-05-30 & 2015-06-25 & 2017-10-03 & 2020-02-27 & \\
\bottomrule
\end{tabular}}
\end{table}


Stage~1 detects four change points around the Great Financial Crisis between 2007 and 2009, and the last two estimators from Stage~1 correspond to the onset (2020-02-20) and the end (2020-04-07) of the stock market crash brought in by the instability due to the COVID-19 pandemic.
Given the clustering of change points between 2007 and 2009, an alternative approach is to adopt a locally stationary factor model as in \cite{barigozzi2020}.
However, such a model does not allow for the number of factors to vary over time, whereas we observe the contrary to be the case when applying the IC-based method of \cite{hallin2007} to each segment defined by $\wh{\Cp}_\chi$, see Table~\ref{tab:fn}.
This supports that it is more appropriate to model the changes in the factor-driven component of this dataset as abrupt changes rather than as smooth transitions.

\begin{table}[h!t!b!]
\caption{\small Estimated number of factors $\wh q_k$ from $\{\mbf X_t, \, \wh\cp_{\chi, k} + 1 \le t \le \wh\cp_{\chi, k + 1}\}$, $k = 0, \ldots, 7$.}
\label{tab:fn}
    \centering
    {\scriptsize
    \begin{tabular}{c cccccccc}
    \toprule
     Segment $k$ & 0 & 1 & 2 & 3 & 4 & 5 & 6 & 7 \\
     \cmidrule(lr){1-1} \cmidrule(lr){2-9}
     $\wh q_k$ & 3 & 4 & 2 & 7 & 2 & 5 & 1 & 2 \\
     \bottomrule
    \end{tabular}}
\end{table}
The estimators from Stage~2 are spread across the period in consideration. Figure~\ref{fig:ex}~(c)--(f) illustrate how the linkages between different companies vary over the four segments identified between $2003$ and $2011$ particularly at the level of industrial sectors, although this information is not used by FVARseg.

\begin{table}[h!t!b!]
\caption{\small Mean and standard errors of $\text{FE}^{\text{avg}}_t$ and 
$\text{FE}^{\text{max}}_t$ for $t \in \mc T$ where $\vert \mc T \vert = 1600$.}
\label{table:real:forecast}
\centering
{\footnotesize
\begin{tabular}{cc cc cc}
\toprule										
&	&	\multicolumn{2}{c}{$\text{FE}^{\text{avg}}$} &		\multicolumn{2}{c}{$\text{FE}^{\text{max}}$}		\\	
&	Forecasting method &	Mean &	SE &	Mean &	SE	\\	\cmidrule(lr){1-2} \cmidrule(lr){3-4} \cmidrule(lr){5-6} 
\ref{method:one} &	Restricted &	0.7671 &	0.3729 &	0.9181 &	0.1898	\\	
&	Unrestricted &	0.7746 &	0.4123 &	0.9204 &	0.2007	\\	\cmidrule(lr){1-2} \cmidrule(lr){3-4} \cmidrule(lr){5-6} 
\ref{method:two} &	Restricted &	0.7831 &	0.4011 &	0.9217 &	0.1962	\\	
&	Unrestricted &	0.8138 &	0.4666 &	0.9279 &	0.2008	\\	\bottomrule
\end{tabular}}
\end{table}

To further validate the segmentation obtained by FVARseg, we perform a forecasting exercise.
Two approaches, referred to as~\ref{method:one} and~\ref{method:two} below, are adopted to build forecasting models where the difference lies in how a sub-sample of $\{\mbf X_u, \, u \le t - 1\}$, is chosen to forecast $\mbf X_t$.
Simply put, \ref{method:one} uses the observations belonging to the same segment as~$\mbf X_t$ only,
for constructing the forecast of $\bm\chi_t$ (resp. $\bm\xi_t$) according to the segmentation defined by $\wh\Cp_\chi$ (resp. $\wh\Cp_\xi$),
while \ref{method:two} ignores the presence of the most recent change point estimator.
We expect~\ref{method:one} to give more accurate predictions if the data undergoes structural changes at the detected change points.
On the other hand, if some of the change point estimators are spurious, \ref{method:two} is expected to produce better forecasts since it makes use of more observations.
We select $\mc T$, the set of time points at which to perform forecasting, such that each $t \in \mc T$ does not belong to the first two segments (i.e.\ $t \ge \max(\wh\cp_{\chi, 2}, \wh\cp_{\xi, 2}) + 1$), and there are at least $n_0$ of observations to build a forecast model separately for $\bm\chi_t$ and $\bm\xi_t$, respectively.
Denoting by $\wh L_\chi(v) = \max\{0 \le k \le \wh K_\chi: \, \wh\cp_{\chi, k} + 1 \le v\}$ the index of $\wh\cp_{\chi, k}$ nearest to and strictly left of $v$ and similarly defining $\wh L_\xi(v)$, this means that $\min(\wh L_\chi(t), \wh L_\xi(t)) \ge 2$ and $\min(t - \wh\cp_{\chi, \wh L_\chi(t)}, t - \wh\cp_{\xi, \wh L_\xi(t)}) \ge n_0$ for all $t \in \mc T$. We have $\vert \mc T \vert = 1600$.
For such $t \in \mc T$, we obtain $\wh{\mbf X}_t(\mbf N) = 
\wh{\bm\chi}_t(N_1) + \wh{\bm\xi}_t(N_2)$ for some $\mbf N = (N_1, N_2)$, where $\wh{\bm\chi}_t(N_1)$ denotes an estimator of the best linear predictor of $\bm\chi_t$ given $\mbf X_{t - \ell}, \, 1 \le \ell \le N_1$, and $\wh{\bm\xi}_t(N_2)$ is defined analogously.
The difference between the two approaches we take lies in the selection of $\mbf N$.
\begin{enumerate}[noitemsep, wide, labelwidth=0pt, labelindent=0pt, label = (F\arabic*)]
\item \label{method:one} We set $N_1 = t - \wh{K}_{\chi, \wh L_\chi(t)} - 1$ and $N_2 = t - \wh{K}_{\xi, \wh L_\xi(t)} - 1$.

\item \label{method:two} We set $N_1 = t - \wh{K}_{\chi, \wh L_\chi(t) - 1} - 1$ 
and $N_2 = t - \wh{K}_{\xi, \wh L_\xi(t) - 1} - 1$.
\end{enumerate}
\cite{barigozzi2022fnets} propose two methods for estimating the best linear predictors
of $\bm\chi_t$ and $\bm\xi_t$ under a stationary factor-adjusted VAR model,
one based on a more restrictive assumption on the factor structure (`restricted') than the other (`unrestricted');
we refer to the paper for their detailed descriptions.
Both estimators are combined with the two approaches~\ref{method:one} and~\ref{method:two}.
Table~\ref{table:real:forecast} reports the summary of the forecasting errors measured as
$\text{FE}^{\text{avg}}_t = \vert \mbf X_t - \wh{\mbf X}_t(\mbf N) \vert_2^2 / \vert \mbf X_t \vert_2^2$ and
$\text{FE}^{\text{max}}_t = \vert \mbf X_t - \wh{\mbf X}_t(\mbf N) \vert_\infty /
\vert \mbf X_t \vert_\infty$,
obtained from combining different best linear predictors with~\ref{method:one}--\ref{method:two}.
According to all evaluation criteria, \ref{method:one} produces forecasts
that are more accurate than \ref{method:two} regardless of the forecasting methods,
which supports the validity of the change point estimators returned by FVARseg.

%


\bibliographystyle{apalike}
\bibliography{fbib}

\clearpage

\appendix

\numberwithin{equation}{section}
\numberwithin{figure}{section}
\numberwithin{table}{section}

\section{Further discussions on Stage~2 of FVARseg}
\label{sec:idio:ext}

\subsection{Extension of Theorem~\ref{thm:idio}}
\label{sec:thm:idio:ext}

We consider the performance of the Stage~2 of FVARseg when applied to some special cases under the model~\eqref{eq:model}
where \ref{cor:idio:station} $\bm\chi_t$ is stationary (i.e.\ $\Cp_\chi = \emptyset$) and
\ref{cor:idio:no} we directly observe $\mbf X_t = \bm\xi_t$ (i.e.\ $\bm\chi_t = \mbf 0$).
Theorem~\ref{thm:common} indicates that in both cases,
the Stage~1 of FVARseg returns $\wh{\Cp}_\chi = \emptyset$.
The results reported in Theorem~\ref{thm:idio} readily extend to such settings.

\begin{cor}
\label{cor:idio}
Suppose that the assumptions of Theorem~\ref{thm:idio} hold,
including Assumption~\ref{assum:idio:size}~\ref{cond:idio:jump}
with $\lambda_{n, p}$ specified below.
Then, with $\varrho^{\k}_{n, p}$ defined as in Theorem~\ref{thm:idio},
i.e.\ 
\begin{align*}
& \varrho^{\k}_{n, p} = \vert {\bm\Delta}_{\xi, k} \vert_\infty^{-2} \l(1 + \max_{0 \le k \le K_\xi} \Vert \bm\beta^{\k} \Vert_1 \r) \times
\l\{\begin{array}{l}
(G K_\xi p)^{\frac{2}{\nu - 2}} \log^{\frac{3\nu}{\nu - 2}}(p) 
\\
\qquad \text{under Assumption~\ref{assum:innov}~\ref{cond:moment}},
\\
\log(G K_\xi p) \\
\qquad \text{under Assumption~\ref{assum:innov}~\ref{cond:gauss}},
\end{array}\r.
\end{align*}
there exist a set $\mc M^\xi_{n, p}$ with $\p(\mc M^\xi_{n, p}) \to 1$ as $n, p \to \infty$
and constants $\epsilon_0, c_0 > 0$ such that on $\mc M^\xi_{n, p}$,
we have 
\begin{align*}
& \wh K_\xi = K_\xi \quad \text{and} \quad 
\l\vert \wh\cp_{\xi, k} - \cp_{\xi, k} \r\vert \le \min\l(\epsilon_0 G, \, c_0 \varrho^{\k}_{n, p}\r)
\text{ for all } 1 \le k \le K_\xi
\end{align*}
for $n$ large enough, in the following situations.
\begin{enumerate}[label = (\alph*)]
\item \label{cor:idio:station} There is no change point in the factor-driven component, 
i.e.\ $K_\chi = 0$, and we set
\begin{align*}
\lambda_{n, p} = M\l(\max_{0 \le k \le K_\xi} \Vert \bm\beta^{\k} \Vert_1 + 1\r)
\l(\vartheta_{n, p} \vee \frac{1}{m} \vee \frac{1}{\sqrt p}\r).
\end{align*}

\item \label{cor:idio:no} We directly observe the piecewise stationary VAR process, 
i.e.\ $\mbf X_t = \bm\xi_t$ for all $t$, and we set
$\lambda_{n, p} = M(\max_{0 \le k \le K_\xi} \Vert \bm\beta^{\k} \Vert_1 + 1) \bar{\vartheta}_{n, p}$ with
\begin{align}
\bar{\vartheta}_{n, p} = \l\{\begin{array}{ll}
\frac{ (np)^{2/\nu} \log^3(p) \log^{2 + 2/\nu}(G) }{G} \vee 
\sqrt{\frac{\log(np)}{G}} 
& \text{under Assumption~\ref{assum:innov}~\ref{cond:moment}},
\\
\sqrt{\frac{\log(np)}{G}} & \text{under Assumption~\ref{assum:innov}~\ref{cond:gauss}.}
\end{array}\r.
\label{eq:bar:vartheta}
\end{align}
\end{enumerate}
\end{cor}

When compared to the methods dedicated to the setting 
corresponding to Corollary~\ref{cor:idio}~\ref{cor:idio:no},
our Stage~2 methodology achieves comparative theoretical performance
in terms of the detection lower bound imposed on the size of changes 
for their detection, and the rate of localisation achieved.
We provide a comprehensive comparison of the Stage~2 methodology
with the existing VAR segmentation methods in the next section,
both on their theoretical and computational properties. 

\subsection{Comparison with the existing VAR segmentation methods}
\label{sec:comparison}

There are a few methods proposed for time series segmentation 
under piecewise stationary, Gaussian VAR models,
a setting that corresponds to Corollary~\ref{cor:idio}~\ref{cor:idio:no} under Gaussianity.
In this setting, we compare the Stage~2 of FVARseg with 
those proposed by \cite{wang2019} and \cite{safikhani2020}.
\begin{table}[htb!]
\caption{\small Comparison of change point methods developed under piecewise stationary VAR models
on their theoretical performance (under Gaussianity) and computational complexity.
Here, $\mathfrak{g} = \max_{1 \le k \le K_\xi} \Vert \bbG^{\k}(\bbG^{[k - 1]})^{-1} \Vert_1$,
$s_\circ = \vert \mc S \vert$ denotes the global sparsity defined 
with $\mc S \in \{1, \ldots, pd\} \times \{1, \ldots, p\}$ where
$[\bm\beta^{\k}]_{i, i'} = 0$ iff $(i, i') \notin \mc S$ for all $k$.
LP($a$) denotes the complexity of solving a linear program with $a$ variables
and Lasso($a$, $b$) that of solving a Lasso problem
with sample size $a$ and dimensionality $b$.}
\label{table:comparison}
\centering
{\small \begin{tabular}{l  c c c}
\toprule
Methods & Separation & Localisation & Complexity \\
\cmidrule(lr){1-1}\cmidrule(lr){2-2}\cmidrule(lr){3-3}\cmidrule(lr){4-4}
Stage~2 of FVARseg & $(1 \vee \mathfrak{g})^2 \log(n\vee p)$ & $\log(n \vee p)$
& $O(G^{-1} n p \;\text{LP}(pd) + np^2)$ \\
\cite{wang2019} & $s_\circ K_\xi \log(n\vee p)$ & $\log(n\vee p)$ & $O(n^2 p \; \text{Lasso}(n, pd))$ \\
\cite{safikhani2020} & $s_\circ^4 K_\xi^2 \log(p)$ & $s_\circ^4 K_\xi \log(p)$ & Not available \\ 
\bottomrule
\end{tabular}}
\end{table}

Table~\ref{table:comparison} summarises the comparative study
in terms of their theoretical and computational properties.
Denoting by $\bar{\Delta}_k$ the size of change between the $(k - 1)$th and the $k$th segments
(measured differently for different methods),
the separation rate refers to some $\nu_{n, p} \to \infty$ such that if 
$\nu_{n, p}^{-1} \; \min_k \bar{\Delta}_k^2 \cdot \min_k (\cp_{\xi, k + 1} - \cp_{\xi, k}) \to \infty$, 
the corresponding method correctly detects all $K_\xi$ change points;
for Stage~2, we set $\bar{\Delta}_k =  \vert {\bm\Delta}_{\xi, k} \vert_\infty$
and for the others, $\bar{\Delta}_k = s_\circ^{-1/2} \vert \bm\beta^{\k} - \bm\beta^{[k - 1]} \vert_2$ (see the caption of Table~\ref{table:comparison} for the definition of $s_\circ$).
The localisation rate refers to some $\varrho_{n, p} \to \infty$ satisfying
$\max_{1 \le k \le K_\xi} w_k^2 \vert \bar{\cp}_k - \cp_{\xi, k} \vert = O_P(\varrho_{n, p})$
for the estimators $\bar{\cp}_k$ returned by respective methods.
For FVARseg, the weights $w_k$ reflect the difficulty associated with locating individual change points, i.e.\ $w_k = \bar{\Delta}_k$,
while for \cite{wang2019} and \cite{safikhani2020}, the weights are global
with $w_k = \min_k \bar{\Delta}_k$ and $w_k = 1$, respectively.
\cite{safikhani2020} further assume that
$\min_k \vert \bm\beta^{\k} - \bm\beta^{[k - 1]} \vert_2$ is bounded away from zero.
We suppose that $\max_k \Vert \bm\beta^{\k} \Vert_1 = O(1)$,
a sufficient condition for the boundedness of $\Vert \bm\Sigma^{\k}_\xi(\omega) \Vert$ for each segment-specific VAR process \citep[Proposition~2.2]{basu2015},
which is required by all the methods in consideration for their theoretical consistency.

Immediate comparison of the theoretical results is difficult due to different definitions of $\bar{\Delta}_k$: Observe that
\begin{align*}
\vert {\bm\Delta}_{\xi, k} \vert_0 \cdot \vert {\bm\Delta}_{\xi, k} \vert_\infty^2
\ge \vert {\bm\Delta}_{\xi, k} \vert_2^2
\ge (2\pi m_\xi)^2 \vert \bm\beta^{\k} - \bm\beta^{[k - 1]} \vert_2^2
\end{align*}
from Assumption~\ref{assum:idio}~\ref{cond:idio:minspec}, where $\vert \cdot \vert_0$ denotes the element-wise $\ell_0$-norm.
Noting that $\vert \bm\Delta_{\xi, k} \vert_0 = \vert \bm\beta^{\k} - \bm\beta^{[k - 1]} \vert_0$, the requirement of Stage~2 of FVARseg may be stronger than that made in \cite{wang2019} if $s_\circ \asymp \vert \bm\beta^{\k} - \bm\beta^{[k - 1]} \vert_0$.
On the other hand, we can have $s_\circ$ much greater than $\vert \bm\beta^{\k} - \bm\beta^{[k - 1]} \vert_0$ if $K_\xi$ is large or when the sparsity pattern of $\bm\beta^{\k}$ varies greatly from one segment to another.
The method proposed by \cite{safikhani2020} is generally worse than the other two both in terms of separation and localisation rates.

The $\ell_1$-regularised Yule-Walker estimation problem in~\eqref{eq:ds}
can be solved in parallel and further, it needs to be performed only $K_\xi + 1$ times with large probability,
which makes the Stage~2 methodology more attractive.
By comparison, the dynamic programming methodology of \cite{wang2019} requires the Lasso estimation to be performed $O(n^2)$ times,
and the multi-stage procedure of \cite{safikhani2020} solves a fused Lasso problem of dimension $np^2d$ to obtain pre-estimators of the change points,
and then exhaustively searches for the final set of estimators which can be NP-hard in the worst case.
In Section~\ref{sec:sim}, we compare the Stage~2 methodology with a blockwise modification of \cite{safikhani2020} that is implemented in the R package {\tt VARDetect} \citep{bai2021multiple}.

Finally, we note that there are methods developed under piecewise stationary extensions
of the low-rank plus sparse VAR($1$) model proposed in \cite{basu2019}, see \cite{bai2021}.
While they additionally permit a low rank structure in the parameter matrices,
the spectrum of $\mbf X_t$ is assumed to be uniformly bounded which rules out
pervasive (serial) correlations in the data
and thus is distinguished from the piecewise stationary factor-adjusted VAR model considered in this paper.

\clearpage

\section{Further information on simulation studies}

\subsection{Data generating processes}
\label{sec:common:dgp}

We provide full details on how the data is generated for numerical experiments reported in Section~\ref{sec:sim}.
Firstly, the factor-driven component $\bm\chi_t$ is generated according to the following two models.
\begin{enumerate}[label = (C\arabic*)]
\item \label{c:one} $\bm\chi^{\k}_t$ admits a static factor model representation, as
\begin{align*}
\chi^{[k]}_{it} = \sum_{j = 1}^q (B^{\k}_{0, ij} + B^{\k}_{1, ij} L + B^{\k}_{2, ij} L^2) u_{jt}, \quad 0 \leq k \leq K_\chi,
\end{align*}
where $u_{jt} \sim_{\iid} \mc N(0, \sigma_j^2)$ 
with $(\sigma_1, \sigma_2) = (1, 0.5)$, 
and the MA coefficients are generated as
$(B^{\k}_{0, ij}, B^{\k}_{1, ij}, B^{\k}_{2, ij}) \sim_{\iid} \mc N_3(\mbf 0, \mbf I)$ 
for all $1 \le i \le p$ and $1 \le j \le q$ when $k = 0$.
Then sequentially for $k = 1, \ldots, K_\chi$, 
we draw $\Pi^{\k}_\chi \subset \{1, \ldots, p\}$ with $\vert \Pi^{\k}_\chi \vert = [0.5 p]$
such that for all $j$,
$(B^{\k}_{0, ij}, B^{\k}_{1, ij}, B^{\k}_{2, ij}) \sim_{\iid} \mc N_3(\mbf 0, \mbf I)$
when $i \in \Pi^{\k}_\chi$ while
$(B^{\k}_{0, ij}, B^{\k}_{1, ij}, B^{\k}_{2, ij}) 
= (B^{[k - 1]}_{0, ij}, B^{[k - 1]}_{1, ij}, B^{[k - 1]}_{2, ij})$
when $i \notin \Pi^{\k}_\chi$.

\item \label{c:two} $\bm\chi^{\k}_t$ does not admit a static factor model representation, as
\begin{align*}
  \chi^{[k]}_{it} = \sum_{j=1}^q \big\{a_{ij} (1-\alpha^{\k}_{ij}L)^{-1} \big\} u_{jt}, \quad 0 \leq k \leq K_\chi,
\end{align*}
where $u_{jt} \sim_{\iid} \mc N(0, 1)$ and
the coefficients $a_{ij}$ are drawn uniformly as $a_{ij} \sim_{\iid} \mc U[-1, 1]$ with $\mc U[a, b]$ denoting a uniform distribution.
The AR coefficients are generated as 
$\alpha^{\k}_{ij} \sim_{\iid} \mc U[-0.8, 0.8]$ when $k = 0$
and then sequentially for $k = 1, \ldots, K_\chi$, 
we draw $\Pi^{\k}_\chi \subset \{1, \ldots, p\}$ with $\vert \Pi^{\k}_\chi \vert = [0.5 p]$
such that for all $j$, we have
$\alpha^{\k}_{ij} = -\alpha^{[k - 1]}_{ij}$ when $i \in \Pi^{\k}_\chi$
and $\alpha^{\k}_{ij} = \alpha^{[k - 1]}_{ij}$ when $i \notin \Pi^{\k}_\chi$.
\end{enumerate}

For generating the piecewise stationary VAR($d$) process $\bm\xi_t$, we consider $\bm\Gamma^{\k} = \mbf I$, $\bm\vep_t \sim_{\iid} \mc N_p(\mbf 0, \mbf I)$ and $d \in \{1, 2\}$.
When $d = 1$,  we generate $\mc N = (\mc V, \mc E)$, a directed Erd\"{o}s-R\'{e}nyi random graph on the vertex set $\mc V = \{1, \ldots, p\}$ with the link probability $1/p$, set the entries of $\bm A_{1}^{[0]}$ as $A_{1, ii'}^{[0]} = 0.4$ if $(i, i') \in \mc E$ and $A_{1, ii'}^{[0]} = 0$ otherwise, then rescale it such that $\Vert \bm A_{1}^{[0]} \Vert = 1$.
When $d = 2$, we rescale the thus-generated $\bm A_{1}^{[0]}$ to have $\Vert \bm A_{1}^{[0]} \Vert = 0.5$ and similarly generate $\mbf A_2^{[0]}$ with $\Vert \bm A_{2}^{[0]} \Vert = 0.5$.
Then sequentially for $k = 1, \ldots, K_\xi$, we set $\bm A_{\ell}^{\k} = - \beta^k \bm A_{\ell}^{[k-1]}$ for $1 \le \ell \le d$ and some $\beta \in (0, 1]$.

\subsection{Complete simulation results}
\label{sec:sim:res}

Tables~\ref{Tab:com1} and~\ref{Tab:oracle} report the complete results obtained for the simulation studies described in Section~\ref{sec:sim}.
In particular, Table~\ref{Tab:com1} compares the performance of FVARseg against BCF \citep{barigozzi2018} on datasets generated as in (M1)--(M2) of Table~\ref{table:sim},
and Table~\ref{Tab:oracle} compares the Stage~2 methodology of FVARseg (i.e.\ Algorithm~\ref{alg:two} applied with $\wh{\Cp}_\chi = \emptyset$ and $\wh{\bm\Gamma}_{\xi, v}(\ell, G) = \wh{\bm\Gamma}_{x, v}(\ell, G)$), against VARDetect \citep{safikhani2020, bai2021multiple} on datasets generated under (M3) in Table~\ref{table:sim}.
All tuning parameters are selected as described in Section~\ref{sec:numeric:detail}.

Denoting by $\wh\Cp$ and $\Cp$ the sets of estimated and true change points, respectively, we report the distributions of $\wh K - K$ (with $\wh K = \vert \wh\Cp \vert$ and $K = \vert\Cp \vert$) and the (scaled) Hausdorff distance between $\wh\Cp$ and $\Cp$,
\begin{align}
\label{eq:hausdorff}
d_H(\wh\Cp, \Cp) = \frac{1}{n} \max\l\{
\max_{\wh\cp \in \wh\Cp} \min_{\cp \in \Cp} \vert \wh\cp - \cp \vert,
\max_{\cp \in \Cp} \min_{\wh\cp \in \wh\Cp} \vert \wh\cp - \cp \vert
\r\}
\end{align}
averaged over $100$ realisations, as well as the average computation time (in seconds) in Table~\ref{Tab:oracle}.

\begin{table}[htbp]
\caption{\small (M1)--(M2): Distributions of $\wh{K}_\chi - K_\chi$ and $\wh{K}_\xi - K_\xi$ 
and the average Hausdorff distance $d_H(\wh\Cp_\chi, \Cp_\chi)$ 
and $d_H(\wh\Cp_\xi, \Cp_\xi)$
returned by FVARseg and BCF \citep{barigozzi2018},
over $100$ realisations. 
We have $K_\xi = 2$ under both (M1) and (M2).}
\label{Tab:com1}
\centering
{ \small
\resizebox{\columnwidth}{!}{
\begin{tabular}{c c c c ccccc ccccc cc}
\toprule
&&& &  \multicolumn{5}{c}{$\wh{K}_\chi - K_\chi$} & \multicolumn{5}{c}{$\wh{K}_\xi - K_\xi$} & \multicolumn{2}{c}{$d_H$} \\ 
& $p$ & $K_\chi$ & Method & $\leq -2$ & $-1$ &  \textbf{0} & 1 & $\geq 2$ & $\leq -2$ & $-1$ &  \textbf{0} & 1 & $\geq 2$ & $\chi$ & $\xi$ \\ 
\cmidrule(lr){1-4}\cmidrule(lr){5-9} \cmidrule(lr){10-14}\cmidrule(lr){15-16}
\multirow{12}{*}{(M1)} & \multirow{4}{*}{$50$}
& \multirow{2}{*}{$0$} & FVARseg & 0 & 0 & \textbf{100} & 0 & 0 & 0 & 0 & \textbf{98} & 2 & 0 & 0.000 & 0.018 \\ 
& & & BCF & 0 & 0 & \textbf{95} & 5 & 0 &  &  &  &  &  & 0.007 &  \\ 
\cmidrule(lr){3-4}\cmidrule(lr){5-9} \cmidrule(lr){10-14}\cmidrule(lr){15-16}

& & \multirow{2}{*}{$3$} & FVARseg &  5 & 15 & \textbf{80} & 0 & 0 & 0 & 6 & \textbf{86} & 8 & 0 & 0.057 & 0.049 \\ 
& && BCF &  0 & 0 & \textbf{91} & 8 & 1 &  &  &  &  &  &  0.010 &  \\ 
\cmidrule(lr){2-4}\cmidrule(lr){5-9} \cmidrule(lr){10-14}\cmidrule(lr){15-16}

& \multirow{4}{*}{$100$}
& \multirow{2}{*}{$0$} & FVARseg & 0 & 0 & \textbf{100} & 0 & 0 & 0 & 0 & \textbf{100} & 0 & 0 & 0.000 & 0.018 \\ 
& & & BCF & 0 & 0 & \textbf{95} & 4 & 1 &  &  &  &  &  & 0.012 &  \\ 
\cmidrule(lr){3-4}\cmidrule(lr){5-9} \cmidrule(lr){10-14}\cmidrule(lr){15-16}

& & \multirow{2}{*}{$3$} & FVARseg & 4 & 10 & \textbf{86} & 0 & 0 & 0 & 2 & \textbf{90} & 6 & 2 & 0.041 & 0.032 \\ 
& & & BCF & 0 & 0 & \textbf{50} & 30 & 20 &  &  &  &  &  & 0.040 &  \\ 
\cmidrule(lr){2-4}\cmidrule(lr){5-9} \cmidrule(lr){10-14}\cmidrule(lr){15-16}

& \multirow{4}{*}{$150$}
& \multirow{2}{*}{$0$} & FVARseg &  0 & 0 & \textbf{100} & 0 & 0 & 0 & 0 & \textbf{100} & 0 & 0 & 0.000 & 0.018 \\ 
& & & BCF &  0 & 0 & \textbf{92} & 7 & 1 &  &  &  &  &  & 0.011 &  \\ 
\cmidrule(lr){3-4}\cmidrule(lr){5-9} \cmidrule(lr){10-14}\cmidrule(lr){15-16}

& & \multirow{2}{*}{$3$} & FVARseg & 4 & 10 & \textbf{86} & 0 & 0 & 0 & 2 & \textbf{93} & 5 & 0 & 0.046 & 0.030 \\ 
& & & BCF & 0 & 0 & \textbf{33} & 28 & 39 &  &  &  &  &  & 0.056 &  \\ 

\cmidrule(lr){1-4}\cmidrule(lr){5-9} \cmidrule(lr){10-14}\cmidrule(lr){15-16}
\multirow{12}{*}{(M2)} & \multirow{4}{*}{$50$}
& \multirow{2}{*}{$0$} & FVARseg & 0 & 0 & \textbf{98} & 2 & 0 & 1 & 2 & \textbf{81} & 14 & 2 & 0.006 & 0.054 \\ 
& & & BCF & 0 & 0 & \textbf{94} & 5 & 1 &  &  &  &  &  & 0.012 &  \\ 
\cmidrule(lr){3-4}\cmidrule(lr){5-9} \cmidrule(lr){10-14}\cmidrule(lr){15-16}

& & \multirow{2}{*}{$2$} & FVARseg &  0 & 0 & \textbf{99} & 1 & 0 & 7 & 20 & \textbf{57} & 12 & 4 & 0.006 & 0.141 \\ 
& && BCF &  0 & 1 & \textbf{91} & 8 & 0 &  &  &  &  &  &  0.013 &  \\ 
\cmidrule(lr){2-4}\cmidrule(lr){5-9} \cmidrule(lr){10-14}\cmidrule(lr){15-16}

& \multirow{4}{*}{$100$}
& \multirow{2}{*}{$0$} & FVARseg & 0 & 0 & \textbf{99} & 1 & 0 & 0 & 4 & \textbf{87} & 8 & 1 & 0.003 & 0.044 \\
& & & BCF & 0 & 0 & \textbf{94} & 5 & 1 &  &  &  &  &  & 0.009 &  \\ 
\cmidrule(lr){3-4}\cmidrule(lr){5-9} \cmidrule(lr){10-14}\cmidrule(lr){15-16}

& & \multirow{2}{*}{$2$} & FVARseg & 0 & 0 & \textbf{99} & 1 & 0 & 3 & 10 & \textbf{67} & 20 & 0 & 0.004 & 0.101 \\ 
& & & BCF & 0 & 0 & \textbf{90} & 10 & 0 &  &  &  &  &  & 0.013 &  \\ 
\cmidrule(lr){2-4}\cmidrule(lr){5-9} \cmidrule(lr){10-14}\cmidrule(lr){15-16}

& \multirow{4}{*}{$150$}
& \multirow{2}{*}{$0$} & FVARseg & 0 & 0 & \textbf{99} & 0 & 1 & 0 & 4 & \textbf{91} & 5 & 0 & 0.003 & 0.042 \\ 
& & & BCF & 0 & 0 & \textbf{92} & 7 & 1 &  &  &  &  &  & 0.009 &  \\ 
\cmidrule(lr){3-4}\cmidrule(lr){5-9} \cmidrule(lr){10-14}\cmidrule(lr){15-16}

& & \multirow{2}{*}{$2$} & FVARseg & 0 & 0 & \textbf{100} & 0 & 0 & 3 & 9 & \textbf{78} & 9 & 1 & 0.005 & 0.086 \\ 
& & & BCF & 0 & 0 & \textbf{91} & 8 & 1 &  &  &  &  &  & 0.011 &  \\ 
\bottomrule
\end{tabular}
}}
\end{table}

\begin{table}[htbp]
\caption{\small (M3): Distribution of $\wh{K}_\xi - K_\xi$ 
and the average Hausdorff distance $d_H(\wh\Cp_\xi, \Cp_\xi)$ 
returned by the Stage~2 of FVARseg and VARDetect \citep{bai2021multiple}, over $100$ realisations. 
We also report the average computation time (in seconds) from $10$ cores of Apple M1 Max with $16$GB of RAM on mac OS.} 
\label{Tab:oracle}
\centering
{\footnotesize
\begin{tabular}{cccc ccccc c c}
\toprule
&&& &   \multicolumn{5}{c}{$\wh{K}_\xi - K_\xi$} & & \\ 
$d$ & $p$ & $K_\xi$ & Method & $\leq -2$ & $-1$ &  \textbf{0} & 1 & $\geq 2$ & $d_H$ & time \\ 
\cmidrule(lr){1-4} \cmidrule(lr){5-9} \cmidrule(lr){10-11}
\multirow{18}{*}{1} & \multirow{6}{*}{$50$}
& \multirow{3}{*}{$0$} & FVARseg ($d = 1$) & 0 & 0 & \textbf{99} & 1 & 0 & 0.001 & 11.14 \\ 
& &  & FVARseg ($d = 2$) & 0 & 0 & \textbf{95} & 5 & 0 & 0.013 & 12.92 \\  
& & & VARDetect & 0 & 0 & \textbf{95} & 1 & 4 & 0.015 & 7.11 \\ 
\cmidrule(lr){3-4} \cmidrule(lr){5-9} \cmidrule(lr){10-11}
& & \multirow{3}{*}{$2$} & FVARseg ($d = 1$) & 0 & 0 & \textbf{98} & 2 & 0 & 0.012 & 20.20 \\ 
& &  & FVARseg ($d = 2$) & 0 & 0 & \textbf{91} & 8 & 1 & 0.019 & 22.55 \\  
& & & VARDetect & 59 & 33 & \textbf{6} & 2 & 0 & 0.307 & 16.50 \\ 
\cmidrule(lr){2-4} \cmidrule(lr){5-9} \cmidrule(lr){10-11}
& \multirow{6}{*}{$100$}
& \multirow{3}{*}{$0$} & FVARseg ($d = 1$) & 0 & 0 & \textbf{100} & 0 & 0 & 0.00 & 26.98 \\ 
& &  & FVARseg ($d = 2$) & 0 & 0 & \textbf{100} & 0 & 0 & 0.00 & 37.96 \\ 
& & & VARDetect & 0 & 0 & \textbf{89} & 5 & 6 & 0.03 & 78.87 \\ 
\cmidrule(lr){3-4} \cmidrule(lr){5-9} \cmidrule(lr){10-11}
& & \multirow{3}{*}{$2$} & FVARseg ($d = 1$) & 0 & 1 & \textbf{98} & 1 & 0 & 0.013 & 47.71 \\ 
& &  & FVARseg ($d = 2$) & 0 & 1 & \textbf{98} & 1 & 0 & 0.015 & 62.78 \\ 
& &  & VARDetect & 90 & 9 & \textbf{0} & 0 & 1 & 0.362 & 96.08 \\ 
\cmidrule(lr){2-4} \cmidrule(lr){5-9} \cmidrule(lr){10-11}
& \multirow{6}{*}{$150$}
& \multirow{3}{*}{$0$} & FVARseg ($d = 1$) & 0 & 0 & \textbf{100} & 0 & 0 & 0.00 & 55.82 \\ 
& &  & FVARseg ($d = 2$) & 0 & 0 & \textbf{100} & 0 & 0 & 0.00 & 87.33 \\   
& &  & VARDetect & 0 & 0 & \textbf{88} & 6 & 6 & 0.036 & 305.95 \\ 
\cmidrule(lr){3-4} \cmidrule(lr){5-9} \cmidrule(lr){10-11}
& & \multirow{3}{*}{$2$} & FVARseg ($d = 1$) & 0 & 1 & \textbf{98} & 1 & 0 & 0.014 & 98.96 \\ 
& &  & FVARseg ($d = 2$) & 0 & 2 & \textbf{97} & 1 & 0 & 0.017 & 146.40 \\  
& &  & VARDetect & 89 & 11 & \textbf{0} & 0 & 0 & 0.361 & 335.79 \\ 
\bottomrule
\multirow{18}{*}{2} & \multirow{6}{*}{$50$}
& \multirow{3}{*}{$0$} & FVARseg ($d = 2$) & 0 & 0 & \textbf{82} & 18 & 0 & 0.054 & 13.52 \\ 
& &  & FVARseg ($d = 1$) & 0 & 0 & \textbf{92} & 8 & 0 & 0.026 & 11.18 \\  
& &  & VARDetect & 0 & 0 & \textbf{96} & 2 & 2 & 0.012 & 37.26 \\ 
\cmidrule(lr){3-4} \cmidrule(lr){5-9} \cmidrule(lr){10-11}
& & \multirow{3}{*}{$2$} & FVARseg ($d = 2$) & 0 & 4 & \textbf{88} & 8 & 0 & 0.031 & 20.90 \\ 
& &  &  FVARseg ($d = 1$) & 0 & 10 & \textbf{86} & 4 & 0 & 0.045 & 17.56 \\ 
& &  & VARDetect & 81 & 6 & \textbf{9} & 1 & 3 & 0.326 & 30.31 \\ 
\cmidrule(lr){2-4} \cmidrule(lr){5-9} \cmidrule(lr){10-11}
& \multirow{6}{*}{$100$}
& \multirow{3}{*}{$0$} & FVARseg ($d = 2$) & 0 & 0 & \textbf{98} & 2 & 0 & 0.008 & 37.96 \\ 
& &  &  FVARseg ($d = 1$) & 0 & 0 & \textbf{97} & 3 & 0 & 0.009 & 27.05 \\  
& &  & VARDetect & 0 & 0 & \textbf{90} & 8 & 2 & 0.021 & 334.24 \\
\cmidrule(lr){3-4} \cmidrule(lr){5-9} \cmidrule(lr){10-11}
& & \multirow{3}{*}{$2$} & FVARseg ($d = 2$) & 0 & 12 & \textbf{87} & 1 & 0 & 0.044 & 57.77 \\ 
& &  & FVARseg ($d = 1$) & 0 & 27 & \textbf{73} & 0 & 0 & 0.08 & 39.57 \\  
& &  & VARDetect & 95 & 1 & \textbf{3} & 0 & 1 & 0.365 & 137.59 \\ 
\cmidrule(lr){2-4} \cmidrule(lr){5-9} \cmidrule(lr){10-11}
& \multirow{6}{*}{$150$}
& \multirow{3}{*}{$0$} & FVARseg ($d = 2$) & 0 & 0 & \textbf{97} & 3 & 0 & 0.01 & 89.06 \\ 
& &  & FVARseg ($d = 1$) & 0 & 0 & \textbf{100} & 0 & 0 & 0.00 & 56.14 \\  
& &  & VARDetect & 0 & 0 & \textbf{93} & 3 & 4 & 0.016 & 1063.33 \\ 
\cmidrule(lr){3-4} \cmidrule(lr){5-9} \cmidrule(lr){10-11}
& & \multirow{3}{*}{$2$} & FVARseg ($d = 2$) & 0 & 15 & \textbf{85} & 0 & 0 & 0.051 & 136.81 \\ 
& &  & FVARseg ($d = 1$) & 1 & 28 & \textbf{71} & 0 & 0 & 0.085 & 86.51 \\ 
& &  & VARDetect & 97 & 1 & \textbf{0} & 0 & 2 & 0.371 & 389.86 \\ 
\bottomrule
\end{tabular}}
\end{table}

\clearpage

\section{Pseudocodes for FVARseg}
\label{sec:alg}

Algorithms~\ref{alg:one} and~\ref{alg:two} provide pseudocodes for Stages~1 and~2 of FVARseg.

\begin{algorithm}[htbp]
\caption{Stage~1 of FVARseg}
\label{alg:one}
\DontPrintSemicolon
\SetAlgoLined
\KwIn{Data $\{\mbf X_t\}_{t = 1}^n$, 
lag window size $m$, 
Bartlett kernel $K(\cdot)$,
bandwidth $G$, 
$\eta \in (0, 1]$,
threshold $\kappa_{n, p}$}

\BlankLine
{\bf Step 0:} Set $\wh\Cp_\chi \leftarrow \emptyset$.

\BlankLine
{\bf Step 1:} At $\omega_l, \, 0 \le l \le m$, compute $T_{\chi, v}(\omega_l, G), \, G \le v \le n - G$,
in~\eqref{eq:common:test:stat} and identify 
$\mc I = \{G, \ldots, n - G\} \setminus \{v: \, \max_{0 \le l \le m} T_{\chi, v}(\omega_l, G) \le \kappa_{n, p}\}$.
 
\BlankLine
{\bf Step 2:} Let $\wh\cp = \arg\max_{v \in \mc I} \max_l T_{\chi, v}(\omega_l, G)$~and~$\omega(\wh\cp) = \arg\max_{\omega_l: \, 0 \le l \le m} T_{\wh\cp}(\omega_l, G)$.
If $T_{\chi, \wh\cp}(\omega(\wh\cp), G) \ge \max_{\wh\cp - \eta G < v \le \wh\cp + \eta G} T_{\chi, v}(\omega(\wh\cp), G)$, 
update $\wh\Cp_\chi \leftarrow \wh\Cp_\chi \cup \{\wh\cp\}$.

\BlankLine
{\bf Step 3:} Update $\mc I \leftarrow \mc I \setminus \{\wh\cp - G + 1, \ldots, \wh\cp + G\}$.

\BlankLine
{\bf Step 4:} Repeat Steps~2--3 until $\mc I$ is empty.

\BlankLine
\KwOut{$\wh{\Cp}_\chi$}
\end{algorithm}

\begin{algorithm}[htbp]
\caption{Stage~2 of FVARseg}
\label{alg:two}
\DontPrintSemicolon
\SetAlgoLined
\KwIn{Data $\{\mbf X_t\}_{t = 1}^n$, 
change point estimators from the common component $\wh\Cp_\chi$, 
$\lambda_{n, p}$ for~\eqref{eq:ds},
bandwidth $G$, $\eta \in (0, 1]$,
threshold $\pi_{n, p}$}

\BlankLine
{\bf Step 0:} Set $\wh\Cp_\xi \leftarrow \emptyset$ and $v_\circ \leftarrow G$.

\BlankLine
{\bf Step 1:} With $\wh{\bm\beta} = \wh{\bm\beta}_{v_\circ}(G)$,
scan $T_{\xi, v}(\wh{\bm\beta}, G)$ for $v \ge v_\circ$
and identify $\check\cp = \min\{v: \, v_\circ \le v \le n - G \text{ and }
T_{\xi, v}(\wh{\bm\beta}, G) > \pi_{n, p}\}$.
 
\BlankLine
{\bf Step 2:} Find $\wh\cp = \arg\max_{\check\cp \le v \le \min(\check\cp + G, n - G)} 
T_{\xi, v}(\wh{\bm\beta}, G)$ and update $\wh\Cp_\xi \leftarrow \wh\Cp_\xi \cup \{\wh\cp\}$.

\BlankLine
{\bf Step 3:} Update $v_\circ \leftarrow \min(\check\cp + 2G, \wh\cp + (\eta + 1) G)$.

\BlankLine
{\bf Step 4:} Repeat Steps~1--3 until $v_\circ > n - G$.

\BlankLine
\KwOut{$\wh{\Cp}_\xi$}
\end{algorithm}

\clearpage

\section{Generalised dynamic factor model}
\label{app:gdfm}

\subsection{GDFM as a representation}

\cite{fornilippi01} show that the necessary and sufficient condition for any $p$-dimensional, weakly stationary time series to admit the generalised dynamic factor model (GDFM) representation, is to have a finite number of the eigenvalues of its spectral density matrix diverge with $p$ (as in Assumption~\ref{assum:factor}) while the remaining ones are bounded for all $p$. In other words, GDFM itself (without the VAR model imposed on $\bm\xi_t$ as in this paper) can be regarded as a representation of high-dimensional time series rather than a model. Overall, GDFM provides the most general framework for high-dimensional time series factor modelling and it encompasses other factor models found in the literature such as static factor models \citep{forni2009opening}.


Static factor models are popularly adopted 
in both stationary \citep{stock2002forecasting, bai2003, fan2013large, barigozzi2020consistent}
and piecewise stationary \citep{barigozzi2018, li2019detection} time series modelling in high dimensions.
Under stationary factor models,
the factor-driven component permits a representation
$\chi_{it} = \bm\lambda_i^\top \mbf f_t$
with some finite-dimensional vector processes $\mbf f_t \in \R^r$ as the common factors;
here, `static' refers to that $\chi_{it}$ loads $\mbf f_t$ contemporaneously 
and does not preclude serial dependence therein.
The model in~\eqref{eq:gdfm} includes such a static factor model
by representing 
$\mbf f_t = \mc B_f(L) \mbf u_t$ with $\mc B_f(L) = \sum_{\ell = 0}^\infty \mbf B_{f, \ell} L^\ell$, 
$\mbf B_{f, \ell} \in \R^{r \times q}$ for some $r \ge q$
(see Remark~R of \cite{forni2009opening}).
On the other hand, some models that have a finite number of factors under~\eqref{eq:gdfm}
cannot be represented with $\mbf f_t$ of finite dimension,
the simplest example being the case where
$\chi_{it} = a_i (1 - b_i L)^{-1} u_t$ for some $b_i \in (-1, 1)$ \citep{forni2015}; see also~\ref{c:two} in Section~\ref{sec:sim}.

\cite{hallin2018optimal} observe that
principal component analysis (PCA), typically accompanying static factor models as an estimation tool,
does not enjoy the optimality property that guarantees their success in the i.i.d.\ case
in the presence of serial correlations,
unlike the dynamic PCA adopted for estimation under GDFM (see Section~\ref{sec:common:post}).

\subsection{VAR representation of GDFM}

For notational simplicity, let $q_k = q$ for all $0 \le k \le K_\chi$.
Suppose that each $(i, j)$th element of the filter $\mc B^{\k}(L)$ in~\eqref{eq:model}, 
say $\mc B^{\k}_{ij}(L) = \sum_{\ell = 0}^\infty B^{\k}_{\ell, ij} L^\ell$,
is a ratio of finite-order polynomials in $L$ such that
for some finite $s_1, s_2 \in \N$,
\begin{align*}
\mc B^{\k}_{ij}(L) = \frac{\mc B^{[k, 1]}_{ij}(L)}{\mc B^{[k, 2]}_{ij}(L)} \quad \text{with} \quad
\mc B^{[k, l]}_{ij}(L) = \sum_{\ell = 0}^{s_l} B^{[k, l]}_{\ell, ij} L^\ell, \, l = 1, 2,
\end{align*}
for all $1 \le j \le q_k$ and $0 \le k \le K_\chi$. 
Furthermore, assume the followings.
\begin{enumerate}[label = (\alph*)]
\item There exists $M_\chi > 0$ such that
\begin{align*}
\max_{0 \le k \le K_\chi} \max_{1 \le i \le p} \max_{1 \le j \le q} \max_{0 \le \ell \le s_1} 
\vert B^{[k, 1]}_{\ell, ij} \vert \le M_\chi.
\end{align*}

\item For all $0 \le k \le K_\chi$, $1 \le i \le p$ and $1 \le j \le q$, we have
$\mc B^{[k, 2]}_{ij}(z) \ne 0$ for all $\vert z \vert \le 1$.
\end{enumerate}

Under such assumptions, Section~4 of \cite{forni2015} establishes that
for generic values of the parameters $B^{[k, 1]}_{\ell, ij}$ and $B^{[k, 2]}_{\ell, ij}$
(outside a countable union of nowhere dense subsets),
$\bm\chi^{\k}_t$ admits a block-wise singular VAR representation 
\begin{align}
\label{eq:gdfm:var}
\bmx
\mc A_\chi^{[k, 1]}(L) & \mbf O & \ldots & \mbf O \\
\mbf O & \mc A_\chi^{[k, 2]}(L) & \ldots & \mbf O \\
& & \ddots & \\
\mbf O & \mbf O & \ldots & \mc A_\chi^{[k, N]}(L) 
\emx \; 
\bmx \bm\chi^{[k, 1]}_t \\ \bm\chi^{[k, 2]}_t \\ \vdots \\ \bm\chi^{[k, N]}_t \emx
= \bmx \mbf R^{[k, 1]} \\ \mbf R^{[k, 2]} \\ \vdots \\ \mbf R^{[k, N]} \emx \mbf u^{\k}_t 
= \mbf R^{\k} \mbf u^{\k}_t,
\end{align}
where $\bm\chi^{[k, h]}_t = (\chi^{\k}_{(q + 1)(h - 1) + i, t}, \, 1 \le i \le q + 1)^\top$
and $\mbf R^{\k} \in \R^{p \times q}$ is of rank $q_k$;
for convenience, we assume that $p = N(q + 1)$ for some $N \in \N$.
Here, each $\bm\chi^{[k, h]}_t$ admits a {\it finite-order} VAR representation determined by
$\mc A_\chi^{[k, h]}(L) = \mbf I - \sum_{\ell = 1}^{s} \mbf A_{\chi, \ell}^{[k, h]} L^\ell$
with its degree $s \le q s_1 + q^2 s_2$,
and $\det(\mc A_\chi^{[k, h]}(z)) \ne 0$ for all $\vert z \vert \le 1$.

The representation~\eqref{eq:gdfm:var} gives
the piecewise stationary factor-adjusted VAR model in~\eqref{eq:model}
the interpretation of decomposing high-dimensional time series
into two latent VAR processes with time-varying parameter matrices,
one of low rank (singular) accounting for dominant dependence
and the other modelling individual interdependence between the variables unaccounted for by the former.

\clearpage

\section{Information on the real dataset}
\label{app:ticker}

Table~\ref{table:data:info} provides the list of the $72$ companies included in 
the application presented in Section~\ref{sec:app}
along with their tickers and industry classifications .

\begin{table}[htbp]
\caption{Tickers and industry classifications of the $72$ companies.}
\label{table:data:info}
\centering
{\scriptsize
\begin{tabular}{lll lll}
\toprule
Ticker	&	Company name	&	Sector	&	Ticker	&	Company name	&	Sector	\\	
\cmidrule(lr){1-3} \cmidrule(lr){4-6}
AMZN	&	Amazon.com	&	Cons. Disc. 	&	AMGN	&	Amgen	&	Health Care 	\\	
CMCSA	&	Comcast	&	Cons. Disc. 	&	BAX	&	Baxter International 	&	Health Care 	\\	
DIS	&	Walt Disney	&	Cons. Disc. 	&	BMY	&	Bristol-Myers Squibb 	&	Health Care 	\\	
F	&	Ford Motor	&	Cons. Disc. 	&	JNJ	&	Johnson \& Johnson 	&	Health Care 	\\	
HD	&	Home Depot 	&	Cons. Disc. 	&	LLY	&	Lilly (Eli) \& Co. 	&	Health Care 	\\	
LOW	&	Lowes	&	Cons. Disc. 	&	MDT	&	Medtronic	 &	Health Care 	\\	
MCD	&	McDonalds	&	Cons. Disc. 	&	MRK	&	Merck \& Co. 	&	Health Care 	\\	
SBUX	&	Starbucks	&	Cons. Disc. 	&	PFE	&	Pfizer	&	Health Care 	\\	
TGT	&	Target	&	Cons. Disc. 	&	UNH	&	United Health	&	Health Care 	\\	
CL	&	Colgate-Palmolive	&	Cons. Stap. 	&	BA	&	Boeing Company	&	Industrials 	\\	
COST	&	Costco	&	Cons. Stap. 	&	CAT	&	Caterpillar	&	Industrials 	\\	
CVS	&	CVS Caremark	&	Cons. Stap. 	&	EMR	&	Emerson Electric	&	Industrials 	\\	
PEP	&	PepsiCo	&	Cons. Stap. 	&	FDX	&	FedEx	&	Industrials 	\\	
PG	&	Procter \& Gamble	&	Cons. Stap. 	&	GD	&	General Dynamics	&	Industrials 	\\	
WMT	&	Wal-Mart Stores	&	Cons. Stap. 	&	GE	&	General Electric	&	Industrials 	\\	
APA	&	Apache	&	Energy 	&	HON	&	Honeywell Intl	&	Industrials 	\\	
COP	&	ConocoPhillips	&	Energy 	&	LMT	&	Lockheed Martin	&	Industrials 	\\	
CVX	&	Chevron	&	Energy 	&	MMM	&	3M Company	&	Industrials 	\\	
HAL	&	Halliburton	&	Energy 	&	NSC	&	Norfolk Southern	&	Industrials 	\\	
NOV	&	National Oilwell Varco 	&	Energy 	&	UNP	&	Union Pacific	&	Industrials 	\\	
OXY	&	Occidental Petroleum 	&	Energy 	&	UPS	&	United Parcel Service	&	Industrials 	\\	
SLB	&	Schlumberger Ltd.	&	Energy 	&	DD	&	Du Pont	&	Materials 	\\	
XOM	&	Exxon Mobil	&	Energy 	&	FCX	&	Freeport-McMoran	&	Materials 	\\	
AIG	&	AIG	&	Financials 	&	CSCO	&	Cisco Systems	&	Technology	\\	
ALL	&	Allstate	&	Financials 	&	EBAY	&	eBay	&	Technology	\\	
AXP	&	American Express Co 	&	Financials 	&	AAPL	&	Apple	&	Technology 	\\	
BAC	&	Bank of America	&	Financials 	&	HPQ	&	Hewlett-Packard	&	Technology 	\\	
BK	&	Bank of New York	&	Financials 	&	IBM	&	IBM	&	Technology 	\\	
C	&	Citigroup	&	Financials 	&	INTC	&	Intel	&	Technology 	\\	
COF	&	Capital One Financial 	&	Financials 	&	MSFT	&	Microsoft	&	Technology 	\\	
GS	&	Goldman Sachs	&	Financials 	&	ORCL	&	Oracle	&	Technology 	\\	
JPM	&	JPMorgan Chase	&	Financials 	&	QCOM	&	QUALCOMM	&	Technology 	\\	
SPG	&	Simon Property	&	Financials 	&	T	&	AT\&T	&	Technology 	\\	
USB	&	U.S. Bancorp	&	Financials 	&	VZ	&	Verizon	&	Technology 	\\	
WFC	&	Wells Fargo 	&	Financials 	&	AEP	&	American Electric Power 	&	Utilities 	\\	
ABT	&	Abbott Laboratories 	&	Health Care 	&	EXC	&	Exelon 	&	Utilities 	\\	\bottomrule
\end{tabular}}
\end{table}

\section{Proofs}
\label{sec:pf}

%

\subsection{Preliminary lemmas}

In the following lemmas, we operate 
under Assumptions~\ref{assum:factor}, \ref{assum:idio}, \ref{assum:common}
and~\ref{assum:innov}.
For notational convenience, we assume that for each $k$, the filters $\mc B^{\k}(L)$ have $\mbf B^{\k}_\ell \in \R^{p \times q}$ with appropriate zero columns such that we can write $\bm\chi^{\k}_t = \mc B^{\k}(L) \mbf u_t$ even when $q_k < q$.

Recall that
\begin{align}
\label{eq:common:tv:spec}
\bm\Sigma_{\chi, v}(\omega, G) = \frac{1}{G} \sum_{k = L_\chi(v - G + 1)}^{L_\chi(v)} 
\{(\cp_{\chi, k + 1} \wedge v) - (\cp_{\chi, k} \vee (v - G))\} \bm\Sigma_\chi^{\k}(\omega),
\end{align}
with $L_\chi(v) = \max\{0 \le k \le K_\chi: \, \cp_{\chi, k} + 1 \le v\}$
denoting the index of the change point nearest to and strictly left of a time point $v$,
and we define $\bm\Sigma_{\xi, v}(\omega, G)$ and $L_\xi(v)$ analogously.
Then, the local spectral density matrix of $\mbf X_t$ is defined as
$\bm\Sigma_{x, v}(\omega, G) = \bm\Sigma_{\chi, v}(\omega, G) + \bm\Sigma_{\xi, v}(\omega, G)$.
Similarly, with $\bm\Gamma^{\k}_\chi(\ell) = \E(\bm\chi^{\k}_{t - \ell} (\bm\chi^{\k}_t)^\top)$
and $\bm\Gamma^{\k}_\xi(\ell) = \E(\bm\xi^{\k}_{t - \ell} (\bm\xi^{\k}_t)^\top)$,
we define the local ACV matrix of $\bm\chi_t$ as
\begin{align*}
\bm\Gamma_{\chi, v}(\ell, G) = \frac{1}{G} \sum_{k = L_\chi(v - G + 1)}^{L_\chi(v)} 
\{(\cp_{\chi, k + 1} \wedge v) - (\cp_{\chi, k} \vee (v - G))\} \bm\Gamma_\chi^{\k}(\ell),
\end{align*}
and analogously define $\bm\Gamma_{\xi, v}(\ell, G)$.
Then we define $\bm\Gamma_{x, v}(\ell, G) = \bm\Gamma_{\chi, v}(\ell, G) + \bm\Gamma_{\xi, v}(\ell, G)$.

\cite{zhang2021} extend the functional dependence measure 
introduced in \cite{wu2005} for high-dimensional, locally stationary time series. 
Denote by $\mc F_t = \{(\mbf u_v, \bm\vep_v), \, v \le t\}$
and $\mc G_\chi^{\k}(\cdot) = (g^{\k}_{\chi, 1}(\cdot), \ldots, g^{\k}_{\chi, p}(\cdot))^\top$ 
and $\mc G_\xi^{\k}(\cdot) = (g^{\k}_{\xi, 1}(\cdot), \ldots, g^{\k}_{\xi, p}(\cdot))^\top$ 
$\R^p$-valued measurable functions such that 
$\bm\chi^{\k}_t = \mc G_\chi^{\k}(\mc F_t)$ for $0 \le k \le K_\chi$,
and $\bm\xi^{\k}_t = \mc G_\xi^{\k}(\mc F_t)$ for $0 \le k \le K_\xi$.
Then, $\mbf X_t = \mc G(t/n, \mc F_t) = \mc G^{[L_\chi(t)]}_\chi(\mc F_t) + \mc G^{[L_\xi(t)]}_\xi(\mc F_t)$
and $X_{it} = g_i(t/n, \mc F_t) = g^{[L_\chi(t)]}_{\chi, i}(\mc F_t) + g^{[L_\xi(t)]}_{\xi, i}(\mc F_t)$.
Also let $\mc F_{t, \{0\}} = \{\ldots, (\mbf u_{-1}, \bm\vep_{-1}), (\mbf u^\prime_0, \bm\vep^\prime_0), 
(\mbf u_1, \bm\vep_1)^\top, \ldots, (\mbf u_t, \bm\vep_t)\}$ 
denote a coupled version of $\mc F_t$ with an independent copy
$(\mbf u^\prime_0, \bm\vep^\prime_0)$ replacing $(\mbf u_0, \bm\vep_0)$.
Then, the element-wise functional dependence measure is defined as
\begin{align*}
\delta_{t, \nu, i} = \sup_{z \in [0, 1]} \l\Vert g_i(z, \mc F_t) - g_i(z, \mc F_{t, \{0\}}) \r\Vert_\nu,
\end{align*}
the uniform functional dependence measure as
\begin{align*}
\delta_{t, \nu} = \sup_{z \in [0, 1]} 
\l\Vert \vert \mc G(z, \mc F_t) - \mc G(z, \mc F_{t, \{0\}}) \vert_\infty \r\Vert_\nu,
\end{align*}
the dependence adjusted norms as
\begin{align*}
\Vert \mbf X_{i \cdot} \Vert_{\nu, \alpha} 
= \sup_{\ell \ge 0} \, (\ell + 1)^\alpha \sum_{t = \ell}^\infty \delta_{t, \nu, i}
\quad \text{and} \quad
\Vert \vert \mbf X_{\cdot} \vert_\infty \Vert_{\nu, \alpha} 
= \sup_{\ell \ge 0} \, (\ell + 1)^\alpha \sum_{t = \ell}^\infty \delta_{t, \nu},
\end{align*}
and the overall and the uniform dependence adjusted norms as
\begin{align*}
\Psi_{\nu, \alpha} = \l( \sum_{i = 1}^p \Vert \mbf X_{i \cdot} \Vert_{\nu, \alpha}^{\nu/2} \r)^{2/\nu}
\quad \text{and} \quad
\Phi_{\nu, \alpha} = \max_{1 \le i \le p} \Vert \mbf X_{i \cdot} \Vert_{\nu, \alpha}.
\end{align*}

\begin{lem}
\label{lem:func:dep}
Let $\alpha \le \varsigma - 1$.
\begin{enumerate}[label = (\alph*)]
\item  
Under Assumption~\ref{assum:innov}~\ref{cond:moment}, we have
\begin{align*}
\Psi_{\nu, \alpha} \le C_{\nu, \Xi, \varsigma} M_\vep^{1/2} p^{2/\nu} \mu_\nu^{1/\nu}
\quad \text{and} \quad
\Vert \vert \mbf X_{\cdot} \vert_\infty \Vert_{\nu, \alpha} 
\le C_{\nu, \Xi, \varsigma} M_\vep^{1/2} \log^{1/2}(p) p^{1/\nu} \mu_\nu^{1/\nu}
\end{align*}
for some constant $C_{\nu, \Xi, \varsigma} > 0$ depending only on its subscripts
(varying from one occasion to another).
\item Under Assumption~\ref{assum:innov}~\ref{cond:moment}--\ref{cond:gauss}, 
we have $\Phi_{\nu, \alpha} \le C_{\nu, \Xi, \varsigma} M_\vep^{1/2} \mu_\nu^{1/\nu}$
for any $\nu$ for which $\Vert u_{jt} \Vert_\nu$ and $\Vert \vep_{it} \Vert_\nu$ exist.
\end{enumerate}
\end{lem}

\begin{proof}
By Minkowski inequality,
\begin{align*}
\delta_{t, \nu, i} \le \max_{0 \le k \le K_\chi} \Vert \mbf B^{\k}_{t, i\cdot} \mbf u_0 \Vert_\nu 
+ \max_{0 \le l \le K_\xi} \Vert \mbf D^{\lll}_{t, i\cdot} (\bm\Gamma^{\lll})^{1/2} \bm\vep_0 \Vert_\nu,
\quad \text{and}
\\
\delta_{t, \nu} \le \max_{0 \le k \le K_\chi} 
\Vert \vert \mbf B^{\k}_t \mbf u_0 \vert_\infty \Vert_\nu + 
\max_{0 \le l \le K_\xi}
\Vert \vert \mbf D^{\lll}_t (\bm\Gamma^{\lll})^{1/2} \bm\vep_0 \vert_\infty \Vert_\nu
\end{align*}
for all $t$.
Due to independence of $u_{jt}$, Assumption~\ref{assum:common}
and Lemma~D.3 of \cite{zhang2021}, there exists $C_\nu > 0$ that depends only on $\nu$
such that
\begin{align*}
\max_{0 \le k \le K_\chi} \Vert \mbf B^{\k}_{t, i\cdot} \mbf u_0 \Vert_\nu \le& 
\max_k \l\Vert \sum_{j = 1}^q \wt{B}^{\k}_{t, ij} u^{\k}_{j0} \r\Vert_\nu
\le C_\nu \max_k \vert \wt{B}^{\k}_{t, i\cdot} \vert_2 \; \mu_\nu^{1/\nu}
\le C_\nu \Xi (1 + t)^{-\varsigma} \; \mu_\nu^{1/\nu}
\end{align*}
for all $1 \le i \le p$, and
\begin{align*}
\max_{0 \le k \le K_\chi} \Vert \vert \mbf B^{\k}_t \mbf u_0 \vert_\infty \Vert_\nu
\le& C_\nu \log^{1/2}(p)
\max_k \l(\sum_{j = 1}^q \vert \mbf B^{\k}_{t, \cdot j} \vert_\infty^2\r)^{1/2} 
q^{1/\nu} \mu_\nu^{1/\nu}
\\
\le& C_\nu \log^{1/2}(p) \Xi (1 + t)^{-\varsigma} q^{1/\nu} \mu_\nu^{1/\nu}.
\end{align*}
Similarly, from Assumption~\ref{assum:idio} and independence of $\vep_{it}$, we have
\begin{align*}
& \max_{0 \le l \le K_\xi} \Vert \mbf D^{\lll}_{t, i\cdot} (\bm\Gamma^{\lll})^{1/2} \bm\vep_0 \Vert_\nu 
\le
C_\nu \max_l\vert \mbf D^{\lll}_{t, i\cdot} (\bm\Gamma^{\lll})^{1/2}\vert_2 \; \mu_\nu^{1/\nu}
\le 
C_\nu M_\vep^{1/2} \max_l \vert \mbf D^{\lll}_{t, i\cdot} \vert_2 \; \mu_\nu^{1/\nu}
\\
& \le
C_\nu M_\vep^{1/2} \max_l \l(\sum_{j = 1}^p\mbf (D^{\lll}_{t, ij})^2 \r)^{1/2} \mu_\nu^{1/\nu}
\le 
C_\nu M_\vep^{1/2} \Xi (1 + t)^{-\varsigma} \mu_\nu^{1/\nu}
\end{align*}
for all $1 \le i \le p$. Then,
\begin{align*}
\max_{0 \le l \le K_\xi} \Vert \vert \mbf D^{\lll}_t (\bm\Gamma^{\lll})^{1/2} \bm\vep_0 \vert_\infty \Vert_\nu^\nu
\le \max_l \sum_{i = 1}^p \Vert \mbf D^{\lll}_{t, i \cdot} (\bm\Gamma^{\lll})^{1/2} \bm\vep_0 \Vert_\nu^\nu
\le (C_\nu M_\vep^{1/2} \Xi (1 + t)^{-\varsigma})^\nu p \mu_\nu
\end{align*}
such that $\max_l \Vert \vert \mbf D^{\lll}_{t, i\cdot} (\bm\Gamma^{\lll})^{1/2} \bm\vep_0 \vert_\infty \Vert_\nu 
\le C_\nu M_\vep^{1/2} \Xi (1 + t)^{-\varsigma} p^{1/\nu} \mu_\nu^{1/\nu}$.
Then, for some constant $C_{\nu, \Xi} > 0$, we have
\begin{align*}
\delta_{t, \nu, i} \le C_{\nu, \Xi} M_\vep^{1/2} (1 + t)^{-\varsigma} \mu_\nu^{1/\nu},
\quad
\delta_{t, \nu} \le C_{\nu, \Xi} M_\vep^{1/2} \log^{1/2}(p) (1 + t)^{-\varsigma} p^{1/\nu} \mu_\nu^{1/\nu},
\end{align*}
and setting $\alpha \le \varsigma - 1$,
\begin{align*}
& \Phi_{\nu, \alpha}  \le C_{\nu, \Xi, \varsigma} M_\vep^{1/2} \mu_\nu^{1/\nu},
\quad 
\Psi_{\nu, \alpha} \le C_{\nu, \Xi, \varsigma} M_\vep^{1/2} p^{2/\nu} \mu_\nu^{1/\nu} \quad \text{and}
\\
& \Vert \vert \mbf X_{\cdot} \vert_\infty \Vert_{\nu, \alpha} 
\le C_{\nu, \Xi, \varsigma} M_\vep^{1/2} \log^{1/2}(p) p^{1/\nu} \mu_\nu^{1/\nu}.
\end{align*}
\end{proof}

\begin{lem}
\label{lem:decay:two} 
There exist some constants $C_{\Xi, \varsigma}, C_{\Xi, \varsigma, \vep} > 0$
which depend only on $\Xi$, $\varsigma$ and $M_\vep$
defined in Assumptions~\ref{assum:idio} and~\ref{assum:common}, such that
for all $h \ge 0$, 
\begin{align*}
\max_{h + 1 \le t \le n} \max_{1 \le i, i' \le p} \l\vert \E(\chi_{i, t - h} \chi_{i't}) \r\vert
\le C_{\Xi, \varsigma} (1 + h)^{-\varsigma}, \\
\max_{h + 1 \le t \le n} \max_{1 \le i, i' \le p} \l\vert \E(\xi_{i, t - h} \xi_{i't}) \r\vert
\le C_{\Xi, \varsigma, \vep} (1 + h)^{-\varsigma}.
\end{align*}
\end{lem}

\begin{proof}
Suppose that $L_\chi(t - \ell) = k$ and $L_\chi(t) = l$.
Then, for any $h \ge 0$, we have
\begin{align*}
& \l\vert \E(\chi_{i, t - h} \chi_{i't}) \r\vert = \l\vert
\E\l(\sum_{\ell, \ell' = 0}^\infty \sum_{j, j' = 1}^{q} 
B^{\k}_{\ell, ij} B^{\lll}_{\ell', i'j'} u_{j, t - \ell - h}u_{j', t -\ell'} \r) \r\vert
= \sum_{\ell = 0}^\infty \l\vert \sum_{j = 1}^q B^{\k}_{\ell, ij} B^{\k}_{\ell + h, i' j} \r\vert
\nn \\
&\le \sum_{\ell = 0}^\infty \vert \mbf B^{\k}_{\ell, i \cdot} \vert_2 \; \vert \mbf B^{\lll}_{\ell + h, i' \cdot} \vert_2 
\le \sum_{\ell = 0}^\infty \frac{\Xi^2}{(1 + \ell)^{\varsigma} (1 + \ell + h)^{\varsigma}}
\le \sum_{\ell = 0}^\infty \frac{\Xi^2}{(1 + \ell)^{\varsigma} (1 + h)^{\varsigma}}
\le C_{\Xi, \varsigma} (1 + h)^{-\varsigma}
\end{align*}
uniformly in $1 \le i, i' \le p$ and $t$ for some $C_{\Xi, \varsigma} > 0$ 
depending only on $\Xi$ and $\varsigma$, thanks to Assumption~\ref{assum:common}.
Similarly, assuming that $L_\xi(t - \ell) = k$ and $L_\xi(t) = l$, we have
\begin{align*}
\l\vert \E(\xi_{i, t - h} \xi_{i't}) \r\vert =& \l\vert
\E\l(\sum_{\ell, \ell' = 0}^\infty 
(\mbf D^{\k}_{\ell, i \cdot} (\bm\Gamma^{\k})^{1/2} \bm\vep_{t - \ell - h})
(\mbf D^{\lll}_{\ell', i' \cdot} (\bm\Gamma^{\lll})^{1/2} \bm\vep_{t - \ell'})\r)
\r\vert
\\
=& \sum_{\ell = 0}^\infty \l\vert 
\mbf D^{\k}_{\ell, i \cdot} (\bm\Gamma^{\k})^{1/2} (\bm\Gamma^{\lll})^{1/2} (\mbf D^{\lll}_{\ell + h, i' \cdot})^\top
\r\vert
\\
\le& M_\vep \sum_{\ell = 0}^\infty 
\vert \mbf D^{\k}_{\ell, i \cdot} \vert_2 \; \vert \mbf D^{\lll}_{\ell + h, i' \cdot} \vert_2 
\le M_\vep
\sum_{\ell = 0}^\infty \frac{\Xi^2}{(1 + \ell)^{\varsigma} (1 + \ell + h)^{\varsigma}}
\\
\le& M_\vep
\sum_{\ell = 0}^\infty \frac{\Xi^2}{(1 + \ell)^{\varsigma} (1 + h)^{\varsigma}}
\le C_{\Xi, \varsigma, \vep} (1 + h)^{-\varsigma}
\end{align*}
uniformly in $1 \le i, i' \le p$ and $t$
for some $C_{\Xi, \varsigma, \vep} > 0$ 
depending only on $\Xi$, $\varsigma$ and $M_\vep$,
from Assumption~\ref{assum:idio}~\ref{cond:idio:innov} and~\ref{cond:idio:coef}.
\end{proof}

The following lemma is a direct consequence of Lemma~\ref{lem:decay:two}.

\begin{lem}
\label{lem:decay} 
Denote by $\bm\Gamma_{x, v}(\ell, G)
= [\gamma_{x, v, ii'}(\ell, G), \, 1 \le i, i' \le p]$.
Then, there exists some constant $C_{\Xi, \varsigma, \vep}$
depending only on $\Xi$, $\varsigma$ and $M_\vep$
defined in Assumptions~\ref{assum:idio} and~\ref{assum:common}, such that
\begin{align*}
& \max_{G \le v \le n} \max_{1 \le i, i' \le p} \gamma_{x, v, ii'}(\ell, G) \le C_{\Xi, \varsigma, \vep} (1 + \vert \ell \vert)^{-\varsigma} \quad \text{and consequently,}
\\
& \max_{G \le v \le n} \max_{1 \le i, i' \le p} \sum_{\vert \ell \vert > m} 
\gamma_{x, v, ii'}(\ell, G) = O(m^{-\varsigma + 1}) = o(m^{-1}).
\end{align*}
\end{lem}

We adopt the notations 
$\bm\Sigma^{\k}_\chi(\omega) = [\sigma^{\k}_{\chi, ii'}(\omega), \, 1 \le i, i' \le p]$
and $\bm\Sigma^{\k}_\xi(\omega) = [\sigma^{\k}_{\xi, ii'}(\omega), \, 1 \le i, i' \le p]$
to denote the elements of the spectral density matrices,
and similarly
$\bm\Gamma^{\k}_\chi(\ell) = [\gamma^{\k}_{\chi, ii'}(\ell), \, 1 \le i, i' \le p]$
and $\bm\Gamma^{\k}_\xi(\ell) = [\gamma^{\k}_{\xi, ii'}(\ell), \, 1 \le i, i' \le p]$.

\begin{lem}
\label{lem:sigma:bound}
Denote by $\bm\Sigma_{x, v}(\omega, G) = [\sigma_{x, v, ii'}(\omega, G), 1 \le i, i' \le p]$.
Then, there exists $B_\sigma > 0$ such that
$\max_{G \le v \le n - G} \max_{1 \le i, i' \le p} \sup_{\pi \in [-\pi, \pi]} \sigma_{x, v, ii'}(\omega, G) \le B_\sigma$.
\end{lem}

\begin{proof}
By Lemma~\ref{lem:decay}, we can find $B_\sigma$ that depends only on $\Xi$, $\varsigma$ and $M_\vep$
defined in Assumptions~\ref{assum:idio} and~\ref{assum:common},
such that
\begin{align*}
\max_v \max_{i, i'} \sup_\omega \l\vert \sigma_{x, v, ii'}(\omega, G) \r\vert \le 
\frac{1}{2\pi} \max_v \max_{i, i'} \sum_{\ell = -\infty}^\infty  \l\vert \gamma_{x, v, ii'}(\ell, G) \r\vert
\le \frac{C_{\Xi, \varsigma, \vep}}{2\pi} \sum_{\ell = -\infty}^\infty \frac{1}{(1 + \vert \ell \vert)^\varsigma}
\le B_\sigma.
\end{align*}
\end{proof}

\begin{lem}
\label{lem:sigma:deriv}
For all $0 \le k \le K_\chi$ and $1 \le i, i' \le p$,
the functions $\omega \mapsto \sigma^{\k}_{\chi, ii'}(\omega)$
possess derivatives of any order and are of bounded variation,
i.e.\ there exists $B^\prime_\sigma > 0$ such that
$\sum_{l = 1}^N \vert \sigma^{\k}_{\chi, ii'}(\omega_l) - \sigma^{\k}_{\chi, ii'} (\omega_{l - 1}) \vert \le B^\prime_\sigma$
uniformly in $1 \le i, i' \le p$, $0 \le k \le K_\chi$, $N \in \N$ and any partition of $[-\pi, \pi]$,
$-\pi = \omega_0 < \omega_1 < \ldots < \omega_N = \pi$.
\end{lem}

\begin{proof}
From Lemma~\ref{lem:decay:two}, 
\begin{center}
$\max_{0 \le k \le K_\chi} \max_{1 \le i, i' \le p} \vert \gamma^{\k}_{\chi, ii'} (\ell) \vert \le C_{\Xi, \varsigma} (1 + \vert \ell \vert)^{-\varsigma}$
\end{center}
for all $\ell$, which implies that
$\sigma^{\k}_{\chi, ii'}(\omega) = (2\pi)^{-1} \sum_{\ell = -\infty}^\infty \gamma^{\k}_{\chi, ii'}(\ell) e^{-\iota \omega \ell}$
has derivatives of all orders. Moreover,
\begin{align*}
\l\vert \frac{d}{d\omega} \sigma^{\k}_{\chi, ii'}(\omega) \r\vert = \frac{1}{2\pi} \sum_{\ell = -\infty}^\infty
\l\vert (-\iota \ell) \gamma^{\k}_{\chi, ii'}(\ell) e^{-\iota \omega \ell} \r\vert
\le \frac{C_{\Xi, \varsigma}}{\pi} \sum_{\ell = 0}^\infty \frac{\ell}{(1 + \ell)^{\varsigma}} 
\le C^\prime_{\Xi, \varsigma}
\end{align*}
for some constant $C^\prime_{\Xi, \varsigma} > 0$ not depending on $1 \le i, i' \le p$,
$0 \le k \le K_\chi$ or $\omega \in [-\pi, \pi]$,
which entails the bounded variation of $\sigma^{\k}_{\chi, ii'}(\omega)$.
\end{proof}

\subsection{Proof of Proposition~\ref{prop:idio:eval}}

Let $\mc D^{\k}(z) = \sum_{\ell = 0}^\infty \mbf D^{\k}_\ell z^\ell$.
Under Assumption~\ref{assum:idio},  we can find a constant $M_\xi>0$ 
which depends only on $M_\vep$, $\Xi$ and $\varsigma$
such that, uniformly over $\omega \in [-\pi, \pi]$ and $0 \le k \le K_\xi$,
\begin{align*}
\mu^{\k}_{\xi, 1}(\omega) =& \Vert \bm\Sigma^{\k}_\xi(\omega) \Vert
= \frac{1}{2\pi} \Vert \mc D^{\k}(e^{-\iota \omega}) \bm\Gamma^{\k} (\mc D^{\k}(e^{-\iota \omega})^*) \Vert
\le \frac{M_\vep}{2\pi} 
\Vert \mc D^{\k}(e^{-\iota \omega}) \Vert_1 \; \Vert \mc D^{\k}(e^{-\iota \omega}) \Vert_\infty
\\
\le& \frac{M_\vep}{2\pi} 
\l(\max_{1 \le i \le p} \sum_{j = 1}^p \sum_{\ell = 0}^\infty \vert D^{\k}_{\ell, ij} \vert \r)
\l(\max_{1 \le j \le p} \sum_{i = 1}^p \sum_{\ell = 0}^\infty \vert D^{\k}_{\ell, ij} \vert \r)
\\
\le& \frac{M_\vep}{2\pi} 
\l(\max_i \sum_{j = 1}^p \sum_{\ell = 0}^\infty \frac{C_{ij}}{(1 + \ell)^{\varsigma}} \r)
\l(\max_j \sum_{i = 1}^p \sum_{\ell = 0}^\infty \frac{C_{ij}}{(1 + \ell)^{\varsigma}} \r)
\le \frac{\Xi^2 M_\vep}{2\pi} \l(\sum_{\ell = 0}^\infty \frac{1}{(1 + \ell)^\varsigma}\r)^2 \le M_\xi.
\end{align*}

\subsection{Proof of Theorem~\ref{thm:common}}

\begin{prop}
\label{prop:loc:spec}
Under the assumptions made in Theorem~\ref{thm:common}, we have
\begin{align*}
\max_{G \le v \le n} \sup_{\omega \in [-\pi, \pi]}
\frac{1}{p} \l\Vert
\wh{\bm\Sigma}_{x, v}(\omega, G) - \bm\Sigma_{\chi, v}(\omega, G)
\r\Vert
= O_p\l( \psi_n \vee \frac{1}{m} \vee \frac{1}{p} \r).
\end{align*}
\end{prop}
\begin{proof}
By Lemma~\ref{lem:spec:x}, we have
\begin{align*}
\E\l(\max_v \sup_\omega \l\vert \wh{\bm\Sigma}_{x, v}(\omega, G)
- \bm\Sigma_{x, v}(\omega, G) \r\vert_2^2\r) \le Cp^2\l(\psi_n^2 \vee \frac{1}{m^2}\r),
\end{align*}
and therefore $\max_v \sup_\omega p^{-1} \vert \wh{\bm\Sigma}_{x, v}(\omega, G)
- \bm\Sigma_{x, v}(\omega, G) \vert_2 = O_p(\psi_n \vee m^{-1})$
by Chebyshev's inequality.
Then, via Proposition~\ref{prop:idio:eval},
\begin{align*}
& \max_v \sup_\omega \frac{1}{p} \l\Vert
\wh{\bm\Sigma}_{x, v}(\omega, G) - \bm\Sigma_{\chi, v}(\omega, G) \r\Vert
\\
\le & \max_v \sup_\omega \frac{1}{p} \l\Vert
\wh{\bm\Sigma}_{x, v}(\omega, G) - \bm\Sigma_{x, v}(\omega, G) \r\Vert
+ \max_v \sup_\omega \frac{1}{p} \l\Vert \bm\Sigma_{\xi, v}(\omega, G) \r\Vert
= O_p\l(\psi_n \vee \frac{1}{m} \vee \frac{1}{p}\r).
\end{align*}
\end{proof}

For ease of notation, define
$\wh{\bm\Sigma}^\prime_{x, v}(\omega, G) = 
\wh{\bm\Sigma}_{x, v}(\omega, G) - \wh{\bm\Sigma}_{x, v + G}(\omega, G)$
and analogously define $\bm\Sigma^\prime_{x, v}(\omega, G)$, $\bm\Sigma^\prime_{\chi, v}(\omega, G)$
and $\bm\Sigma^\prime_{\xi, v}(\omega, G)$.
By definition and Assumption~\ref{assum:common:size}~\ref{cond:common:spacing}, 
$\bm\Sigma^\prime_{\chi, \cp_k}(\omega, G) = {\bm\Delta}_{\chi, k}(\omega)$.
Also, let $\omega(v) = \arg\max\{\omega_l, \, 0 \le l \le m: \, T_{\chi, v}(\omega_l, G) \}$
and $\omega_{\k}^\circ = \arg\max\{\omega_l, \, 0 \le l \le m: \, \Vert {\bm\Delta}_{\chi, k}(\omega_l) \Vert \}$.
Then, due to the Lipschitz continuity of $p^{-r^\prime_{k, 1}} \Vert {\bm\Delta}_{\chi, k}(\omega) \Vert$
(see Assumption~\ref{assum:common:size}~\ref{cond:common:jump}), 
we have 
\begin{align}
\label{eq:lipschitz}
\Vert {\bm\Delta}_{\chi, k}(\omega_{\k}^\circ) \Vert \ge p^{r^\prime_{k, 1}}
\l(p^{-r^\prime_{k, 1}} \Delta_{\chi, k} + O(m^{-1})\r) \ge (1 - \epsilon) \Delta_{\chi, k}
\end{align}
for some small enough constant $\epsilon \in (0, 1)$.
In what follows, we omit the subscript $\chi$ from $\wh\cp_{\chi, k}$ and $\cp_{\chi, k}$ for simplicity
and throughout the proof, we operate on the set 
$\mc M^\chi_{n, p} = \mc E^{(1)}_{n, p} \cap \bar{\mc E}^{(1)}_{n, p}$, 
where
\begin{align*}
\mc E^{(1)}_{n, p} = \l\{
\max_{G \le v \le n} \sup_{\omega \in [-\pi, \pi]} \frac{1}{p}
\l\Vert \wh{\bm\Sigma}_{x, v}(\omega, G) - \bm\Sigma_{\chi, v}(\omega, G) \r\Vert
\le \frac{M}{2}\l(\psi_n \vee \frac{1}{m} \vee \frac{1}{p}\r)
\r\}
\end{align*}
with $M$ as in Theorem~\ref{thm:common},
and $\bar{\mc E}^{(1)}_{n, p}$ is defined in~\eqref{eq:set:e:one:tilde} below.
By Proposition~\ref{prop:loc:spec}, we have $\p(\mc E^{(1)}_{n, p}) \to 1$ as $n, p \to \infty$
and similarly, $\p(\bar{\mc E}^{(1)}_{n, p}) \to 1$ by Lemma~\ref{lem:Qs},
such that $\p(\mc M^\chi_{n, p}) \to 1$.

\begin{proof}[Proof of Theorem~\ref{thm:common}~\ref{thm:common:one}]
On $\mc E^{(1)}_{n, p}$, we have
\begin{align}
\label{thm:common:one:eq:one}
\l\vert T_{\chi, v}(\omega(v), G) - \l\Vert \bm\Sigma^\prime_{\chi, v}(\omega(v), G) \r\Vert \r\vert \le Mp\l(\psi_n \vee \frac{1}{m} \vee \frac{1}{p}\r)
\end{align}
for all $G \le v \le n - G$.
From~\eqref{thm:common:one:eq:one}, it follows that for any $v$ satisfying $\min_{1 \le k \le K_\chi} \vert v - \cp_k \vert \ge G$,
we have $T_{\chi, v}(\omega(v), G) \le \kappa_{n, p}$
since $\bm\Sigma^\prime_{\chi, v}(\omega(v), G) = \mbf O$
due to Assumption~\ref{assum:common:size}~\ref{cond:common:spacing},
and such $v$ does not belong to $\mc I$.
Also, noting that by~\eqref{eq:lipschitz}, \eqref{thm:common:one:eq:one} and the definition of $\omega(\cp_k)$ and $\omega_{\k}^\circ$,
\begin{align*}
T_{\chi, \cp_k}(\omega(\cp_k), G)
\ge T_{\chi, \cp_k}(\omega_{\k}^\circ, G) \ge (1 - \epsilon) \Delta_{\chi, k} - Mp\l(\psi_n \vee \frac{1}{m} \vee \frac{1}{p}\r)
> \kappa_{n, p},
\end{align*}
we conclude that at least one change point is detected within distance $G$ from each $\cp_k, \, 1 \le k \le K_\chi$.
Next, suppose that $\wh\cp \in \wh\Cp_\chi$ satisfies $\vert \wh\cp - \cp_k \vert < G$.
From that 
$T_{\chi, \wh\cp}(\omega(\wh\cp), G) \ge 
T_{\chi, \cp_k}(\omega(\cp_k), G) \ge T_{\chi, \cp_k}(\omega_{\k}^\circ, G)$, we obtain 
\begin{align*}
& \frac{G - \vert \wh\cp - \cp_k \vert}{G} \Vert {\bm\Delta}_{\chi, k}(\omega(\wh\cp)) \Vert \ge
\Delta_{\chi, k} - 2 Mp\l(\psi_n \vee \frac{1}{m} \vee \frac{1}{p}\r), \quad \text{hence}
\\
& 2 Mp\l(\psi_n \vee \frac{1}{m} \vee \frac{1}{p}\r) \ge
\Delta_{\chi, k} - \frac{G - \vert \wh\cp - \cp_k \vert}{G} \Vert {\bm\Delta}_{\chi, k}(\omega(\wh\cp)) \Vert 
\ge \frac{\vert \wh\cp - \cp_k \vert}{G} \Delta_{\chi, k}, 
\\
& \therefore \quad \vert \wh\cp - \cp_k \vert \le \frac{G p M(\psi_n \vee m^{-1} \vee p^{-1})}{\Delta_{\chi, k}} \le \epsilon_0 G
\end{align*}
for some small constant $\epsilon_0 \in (0, 1/2)$ and large enough $n$
under Assumption~\ref{assum:common:size}~\ref{cond:common:spacing},
i.e.\ we detect at least one change point within $(\epsilon_0 G)$-distance from each $\cp_k, \, 1 \le k \le K_\chi$.
Finally, suppose that at some $v$ satisfying
$(1 - \epsilon_0) G \le \vert v - \cp_k \vert < G$,
we have $T_{\chi, v}(\omega(v), G)  > \kappa_{n, p}$.
Then by~\eqref{thm:common:one:eq:one} and the lower bound on $\kappa_{n, p}$,
\begin{align*}
2Mp\l(\psi_n \vee \frac{1}{m} \vee \frac{1}{p}\r) & < T_{\chi, v}(\omega(v), G) \le \l\Vert \bm\Sigma^\prime_{\chi, v}(\omega(v), G) \r\Vert + Mp\l(\psi_n \vee \frac{1}{m} \vee \frac{1}{p}\r)
\\
& \le \epsilon_0 \Vert {\bm\Delta}_{\chi, k}(\omega(v)) \Vert + Mp\l(\psi_n \vee \frac{1}{m} \vee \frac{1}{p}\r)
\\
& < (\epsilon_0 + \eta) \Vert {\bm\Delta}_{\chi, k}(\omega(v)) \Vert - Mp\l(\psi_n \vee \frac{1}{m} \vee \frac{1}{p}\r)
\le T_{\chi, v'} (\omega(v), G)
\end{align*}
where $v' = v + \text{sign}(\cp_k - v) \lfloor \eta G \rfloor$,
provided that $\eta > 2 \epsilon_0$,
i.e.\ such $v$ cannot be a local maximiser of $T_{\chi, v}(\omega(v), G)$ within its $\eta G$-radius.
This, combined with how $\mc I$ is updated at each iteration, makes sure that
only a single estimator is added to $\wh\Cp_\chi$ for each change point.
\end{proof}

\begin{proof}[Proof of Theorem~\ref{thm:common}~\ref{thm:common:two}]
WLOG, we consider the case when $\wh\cp_k \le \cp_k$; the following arguments apply analogously
to the case when $\wh\cp_k > \cp_k$.
We prove by contradiction that if
$\cp_k - \wh\cp_k > \rho^{\k}_{n, p}$,
we have $T_{\chi, \wh\cp_k}(\omega(\wh\cp_k), G) < T_{\chi, \cp_k}(\omega(\wh\cp_k), G)$
and thus $\wh\cp_k$ cannot be the local maximiser of $T_{\chi, v}(\omega(\wh\cp_k), G)$ within its $\eta G$-environment as required.

From Theorem~\ref{thm:common}~\ref{thm:common:one},
we have $\vert \wh\cp_k - \cp_k \vert \le \epsilon_0 G$.
Then from that $T_{\chi, \wh\cp_k}(\omega(\wh\cp_k), G) \ge T_{\chi, \cp_k}(\omega_{\k}^\circ, G)$
and by~\eqref{eq:lipschitz} and~\eqref{thm:common:one:eq:one}, we have
\begin{align}
\l\Vert \bm\Sigma^\prime_{\chi, \wh\cp_k}(\omega(\wh\cp_k), G) \r\Vert
\ge 
(1 - \epsilon) \Delta_{\chi, k} - 2Mp\l(\psi_n \vee \frac{1}{m} \vee \frac{1}{p}\r) 
\ge \frac{1}{2} \Delta_{\chi, k}
\label{thm:common:one:eq:two} 
\end{align}
for an arbitrarily small constant $\epsilon > 0$.
Noting that $\wh{\bm\Sigma}^\prime_{x, v}(\omega, G)$ and $\bm\Sigma^\prime_{\chi, v}(\omega, G)$
are Hermitian (and thus diagonalisable with real diagonal entries), we write
\begin{align*}
\l\Vert \wh{\bm\Sigma}^\prime_{x, \wh\cp_k}(\omega(\wh\cp_k), G) \r\Vert
= \l\vert \wh{\mbf g}_k^* \wh{\bm\Sigma}^\prime_{x, \wh\cp_k}(\omega(\wh\cp_k), G) \wh{\mbf g}_k \r\vert,
\quad
\l\Vert \bm\Sigma^\prime_{\chi, \cp_k}(\omega(\wh\cp_k), G) \r\Vert
= \l\vert \mbf g_k^* \bm\Sigma^\prime_{\chi, \cp_k}(\omega(\wh\cp_k), G) \mbf g_k \r\vert,
\end{align*}
for some $\wh{\mbf g}_k, \mbf g_k \in \mathbb{C}^p$ satisfying $\Vert \wh{\mbf g}_k \Vert = \Vert \mbf g_k \Vert = 1$, 
for all $1 \le k \le K_\chi$.
By Assumption~\ref{assum:common:size}~\ref{cond:common:spacing}, we have
\begin{align*}
\l\Vert \bm\Sigma^\prime_{\chi, \wh\cp_k}(\omega(\wh\cp_k), G) \r\Vert
= \l\Vert \frac{G - \vert \wh\cp_k - \cp_k \vert}{G}
\bm\Sigma^\prime_{\chi, \cp_k}(\omega(\wh\cp_k), G) \r\Vert
= \l\vert \mbf g_k^* \bm\Sigma^\prime_{\chi, \wh\cp_k}(\omega(\wh\cp_k), G) \mbf g_k \r\vert.
\end{align*}
Then on $\mc E^{(1)}_{n, p}$,
there exists $s_k$ with $\vert s_k \vert = 1$ and a constant $c_1 > 0$ such that
\begin{align}
& \max_{1 \le k \le K_\chi} \frac{\Delta_{\chi, k}}{p}
\l\Vert \wh{\mbf g}_k - s_k \mbf g_k \r\Vert
\nn \\
& \le \max_{1 \le k \le K_\chi} 
\frac{\Delta_{\chi, k}}{p}
\frac{\l\Vert \wh{\bm\Sigma}^\prime_{x, \wh\cp_k}(\omega(\wh\cp_k), G) - 
\bm\Sigma^\prime_{\chi, \wh\cp_k}(\omega(\wh\cp_k), G) \r\Vert}
{\mu_1\l(\bm\Sigma^\prime_{\chi, \wh\cp_k}(\omega(\wh\cp_k), G)\r) -
\mu_2\l(\bm\Sigma^\prime_{\chi, \wh\cp_k}(\omega(\wh\cp_k), G)\r)}
\nn \\
& \le \max_{1 \le k \le K_\chi} 
\frac{2b^{\k}_1(\omega(\wh\cp_k))}{a^{\k}_1(\omega(\wh\cp_k)) - b^{\k}_2(\omega(\wh\cp_k))} 
\cdot M\l(\psi_n \vee \frac{1}{m} \vee \frac{1}{p} \r)
\le c_1 \l(\psi_n \vee \frac{1}{m} \vee \frac{1}{p} \r),
\label{thm:common:one:eq:dk}
\end{align}
where the first inequality follows from Corollary~1 of \cite{yu2015}, 
and the second one from Assumption~\ref{assum:common:size}~\ref{cond:common:jump},
\eqref{thm:common:one:eq:one} and~\eqref{thm:common:one:eq:two}.
WLOG, suppose that $\mbf g_k^* \bm\Sigma^\prime_{\chi, \cp_k}(\omega(\wh\cp_k), G) \mbf g_k > 0$.
Then,
\begin{align*}
& \wh{\mbf g}_k^* \wh{\bm\Sigma}^\prime_{x, \wh\cp_k}(\omega(\wh\cp_k), G) \wh{\mbf g}_k
=
\mbf g_k^* \bm\Sigma^\prime_{\chi, \wh\cp_k}(\omega(\wh\cp_k), G) \mbf g_k
\\
& +
(\wh{\mbf g}_k - s_k \mbf g_k)^* \bm\Sigma^\prime_{\chi, \wh\cp_k}(\omega(\wh\cp_k), G) (s_k\mbf g_k)
+ \wh{\mbf g}_k ^* \bm\Sigma^\prime_{\chi, \wh\cp_k}(\omega(\wh\cp_k), G) (\wh{\mbf g}_k - s_k \mbf g_k)
\\
& + \wh{\mbf g}_k^* \l(\wh{\bm\Sigma}^\prime_{x, \wh\cp_k}(\omega(\wh\cp_k), G)
- \bm\Sigma^\prime_{\chi, \wh\cp_k}(\omega(\wh\cp_k), G) \r) \wh{\mbf g}_k
=: \mbf g_k^* \bm\Sigma^\prime_{\chi, \wh\cp_k}(\omega(\wh\cp_k), G) \mbf g_k +
I + II + III.
\end{align*}
From~\eqref{thm:common:one:eq:one}, \eqref{thm:common:one:eq:two}
and Assumption~\ref{assum:common:size}~\ref{cond:common:jump}, 
for small enough $\epsilon > 0$,
\begin{align*}
\vert III \vert \le \l\Vert \wh{\bm\Sigma}^\prime_{x, \wh\cp_k}(\omega(\wh\cp_k), G)
- \bm\Sigma^\prime_{\chi, \wh\cp_k}(\omega(\wh\cp_k), G) \r\Vert \le  M p \l(\psi_n \vee \frac{1}{m} \vee \frac{1}{p} \r)
\le \epsilon \l\Vert \bm\Sigma^\prime_{\chi, \wh\cp_k}(\omega(\wh\cp_k), G) \r\Vert.
\end{align*}
Also by~\eqref{thm:common:one:eq:dk}, 
\begin{align*}
\vert I \vert &\le 
\Vert \wh{\mbf g}_k - s_k \mbf g_k \Vert \; 
\l\Vert \bm\Sigma^\prime_{\chi, \wh\cp_k}(\omega(\wh\cp_k), G) \r\Vert
\le \l\Vert \bm\Sigma^\prime_{\chi, \wh\cp_k}(\omega(\wh\cp_k), G) \r\Vert \cdot 
\frac{c_1 (\psi_n \vee m^{-1} \vee p^{-1})}{p^{-1} \Delta_{\chi, k}}
\\
&\le \epsilon \l\Vert \bm\Sigma^\prime_{\chi, \wh\cp_k}(\omega(\wh\cp_k), G) \r\Vert
\end{align*}
for all $1 \le k \le K_\chi$ under Assumption~\ref{assum:common:size}~\ref{cond:common:jump},
and $II$ is bounded analogously. 
Putting together the bounds on $I$--$III$
together with the fact that
$\mbf g_k^* \bm\Sigma^\prime_{\chi, \cp_k}(\omega(\wh\cp_k), G) \mbf g_k
= \Vert \bm\Sigma^\prime_{\chi, \cp_k}(\omega(\wh\cp_k), G) \Vert$,
we have 
$\wh{\mbf g}_k^* \wh{\bm\Sigma}^\prime_{x, \wh\cp_k}(\omega(\wh\cp_k), G) \wh{\mbf g}_k
= (1 - 3\epsilon) \Vert \bm\Sigma^\prime_{\chi, \cp_k}(\omega(\wh\cp_k), G) \Vert > 0$,
and similarly we can show that
$\wh{\mbf g}_k^* \wh{\bm\Sigma}^\prime_{x, \cp_k}(\omega(\wh\cp_k), G) \wh{\mbf g}_k > 0$.
Then, we observe:
\begin{align}
& T_{\chi, \wh\cp_k}(\omega(\wh\cp_k), G) 
= \wh{\mbf g}_k^* \wh{\bm\Sigma}^\prime_{x, \wh\cp_k}(\omega(\wh\cp_k), G) \wh{\mbf g}_k 
\ge T_{\chi, \cp_k}(\omega(\wh\cp_k), G)
\ge \wh{\mbf g}_k^* \wh{\bm\Sigma}^\prime_{x, \cp_k}(\omega(\wh\cp_k), G) \wh{\mbf g}_k,
\nn
\end{align}
implies that
\begin{align}
& 0 < \wh{\mbf g}_k^* \l(\wh{\bm\Sigma}^\prime_{x, \wh\cp_k}(\omega(\wh\cp_k), G) 
- \wh{\bm\Sigma}^\prime_{x, \cp_k}(\omega(\wh\cp_k), G) \r)\wh{\mbf g}_k 
= \mbf g_k^* \l(\bm\Sigma^\prime_{\chi, \wh\cp_k}(\omega(\wh\cp_k), G) 
- \bm\Sigma^\prime_{\chi, \cp_k}(\omega(\wh\cp_k), G) \r)\mbf g_k
\nn \\
& + (\wh{\mbf g}_k - s_k \mbf g_k)^* \l(\bm\Sigma^\prime_{\chi, \wh\cp_k}(\omega(\wh\cp_k), G) 
- \bm\Sigma^\prime_{\chi, \cp_k}(\omega(\wh\cp_k), G) \r) (s_k\mbf g_k)
\nn \\
& + \wh{\mbf g}_k^* \l(\bm\Sigma^\prime_{\chi, \wh\cp_k}(\omega(\wh\cp_k), G) 
- \bm\Sigma^\prime_{\chi, \cp_k}(\omega(\wh\cp_k), G) \r) (\wh{\mbf g}_k - s_k \mbf g_k)
\nn \\
& + \wh{\mbf g}_k^* \l\{ \wh{\bm\Sigma}^\prime_{x, \wh\cp_k}(\omega(\wh\cp_k), G) 
- \wh{\bm\Sigma}^\prime_{x, \cp_k}(\omega(\wh\cp_k), G) - 
\E\l(\wh{\bm\Sigma}^\prime_{x, \wh\cp_k}(\omega(\wh\cp_k), G) 
- \wh{\bm\Sigma}^\prime_{x, \cp_k}(\omega(\wh\cp_k), G)\r) \r\} \wh{\mbf g}_k
\nn \\
& + \wh{\mbf g}_k^* \l\{ 
\E\l(\wh{\bm\Sigma}^\prime_{x, \wh\cp_k}(\omega(\wh\cp_k), G) 
- \wh{\bm\Sigma}^\prime_{x, \cp_k}(\omega(\wh\cp_k), G)\r) -
\l(\bm\Sigma^\prime_{x, \wh\cp_k}(\omega(\wh\cp_k), G) 
- \bm\Sigma^\prime_{x, \cp_k}(\omega(\wh\cp_k), G)\r) \r\} \wh{\mbf g}_k
\nn \\
& + \wh{\mbf g}_k^* \l(\bm\Sigma^\prime_{\xi, \wh\cp_k}(\omega(\wh\cp_k), G) 
-  \bm\Sigma^\prime_{\xi, \cp_k}(\omega(\wh\cp_k), G)\r) \wh{\mbf g}_k
=: \mc F_k + \mc R_{k1} + \mc R_{k2} + \mc R_{k3} + \mc R_{k4} + \mc R_{k5}.
\label{thm:common:one:eq:three}
\end{align}
First, note that by~\eqref{thm:common:one:eq:two},
\begin{align}
\mc F_k = - \frac{\vert \wh\cp_k - \cp_k \vert}{G} \Vert {\bm\Delta}_{\chi, k}(\omega(\wh\cp_k)) \Vert 
\le - \frac{\vert \wh\cp_k - \cp_k \vert}{2G} \Delta_{\chi, k},
\label{thm:common:one:eq:four}
\end{align}
and as applying the same arguments as those adopted in bounding $I$ and $II$ above,
we have $\max(\mc R_{k1}, \mc R_{k2}) \le \epsilon \vert \mc F_k \vert$.
Now we turn our attention to $\mc R_{k4}$. 
Let $\gamma_{x, ii'}(t, \ell) = \gamma^{[L_\chi(t)]}_{\chi, ii'}(\ell) + \gamma^{[L_\xi(t)]}_{\xi, ii'}(\ell)$
with $\gamma^{\k}_{\chi, ii'}(\ell) = \E(\chi^{\k}_{i, t - \ell} \chi^{\k}_t)$
and $\gamma^{\k}_{\xi, ii'}(\ell) = \E(\xi^{\k}_{i, t - \ell} \xi^{\k}_t)$, 
and $\wt\gamma_{x, ii'}(t, \ell) = \E(X_{i, t - \ell} X_{i't})$ when $\ell \ge 0$
and $\wt\gamma_{x, ii'}(t, \ell) = \E(X_{it} X_{i', t - |\ell|})$ when $\ell < 0$.
Then,
\begin{align*}
& \frac{2\pi G}{\cp_k - \wh\cp_k} \l[ \mc R_{k4} \r]_{i, i'}
= 
\sum_{\ell = -m}^m K\l(\frac{\ell}{m}\r) e^{-\iota \ell \omega(\wh\cp_k)} \frac{1}{\cp_k - \wh\cp_k} \times 
\\
& \l\{
\l( \sum_{t = \wh\cp_k - G + 1 + \vert \ell \vert}^{\cp_k - G + \vert \ell \vert} \wt\gamma_{x, ii'}(t, \ell) - \sum_{t = \wh\cp_k - G + 1}^{\cp_k - G} \gamma_{x, ii'}(t, \ell) \r) 
-\sum_{t = \wh\cp_k + 1}^{\cp_k} \l(\wt\gamma_{x, ii'}(t, \ell) - \gamma_{x, ii'}(t, \ell) \r) \r.
\\
& \l. - \l(\sum_{t = \wh\cp_k + 1 + \vert \ell \vert}^{\cp_k + \vert \ell \vert} \wt\gamma_{x, ii'}(t, \ell)
- \sum_{t = \wh\cp_k + 1}^{\cp_k} \gamma_{x, ii'}(t, \ell) \r)
+ \sum_{t = \wh\cp_k + G + 1}^{\cp_k + G} \l(\wt\gamma_{x, ii'}(t, \ell) - \gamma_{x, ii'}(t, \ell) \r)
\r\}
\\
& + \sum_{\ell = -m}^m \frac{\vert \ell \vert}{m} e^{-\iota \ell \omega(\wh\cp_k)} \frac{1}{\cp_k - \wh\cp_k}
\l(\sum_{t = \wh\cp_k - G + 1}^{\cp_k - G} - 2 \sum_{t = \wh\cp_k + 1}^{\cp_k} + \sum_{t = \wh\cp_k + G + 1}^{\cp_k + G} \r) \gamma_{x, ii'}(t, \ell)
\\
& + \sum_{\ell: \, \vert \ell \vert > m} \frac{\vert \ell \vert}{m} e^{-\iota \ell \omega(\wh\cp_k)} \frac{1}{\cp_k - \wh\cp_k}
\l(\sum_{t = \wh\cp_k - G + 1}^{\cp_k - G} - 2 \sum_{t = \wh\cp_k + 1}^{\cp_k} + \sum_{t = \wh\cp_k + G + 1}^{\cp_k + G} \r) \gamma_{x, ii'}(t, \ell) =: IV + V + VI.
\end{align*}
By Lemma~\ref{lem:decay:two}, 
there exists constant $C_{\Xi, \varsigma, \vep}, C_{\Xi, \varsigma, \vep}^\prime > 0$
that do not depend on $i, i'$ such that
\begin{align*}
\vert V \vert &\le \frac{4C_{\Xi, \varsigma, \vep}}{m} \sum_{\ell = -m}^m \frac{\vert \ell \vert}{(1 + \vert \ell \vert)^{\varsigma}}
\le \frac{4C_{\Xi, \varsigma, \vep}}{m} \sum_{\ell = -m}^m \frac{1}{(1 + \vert \ell \vert)^{\varsigma - 1}} 
\le \frac{C_{\Xi, \varsigma, \vep}^\prime}{m}, \quad \text{and}
\\
\vert VI \vert &\le \frac{4C_{\Xi, \varsigma, \vep}}{m} \sum_{\ell: \, \vert \ell \vert > m} \frac{1}{(1 + \vert \ell \vert)^{\varsigma}}
\le \frac{C_{\Xi, \varsigma, \vep}^{\prime}}{m^{\varsigma - 1}}.
\end{align*}
Similarly, noting that there are at most a single change point
within any $(\cp_k - \wh\cp_k)$-interval,
we have
\begin{align*}
\vert IV \vert \le \frac{12C_{\Xi, \varsigma, \vep}}{\cp_k - \wh\cp_k} \sum_{\ell = -m}^m \frac{\vert \ell \vert}{(1 + \vert \ell \vert)^{\varsigma}} \le \frac{C_{\Xi, \varsigma, \vep}^{\prime\prime}}{m}
\end{align*}
for some constant $C_{\Xi, \varsigma, \vep}^{\prime\prime} > 0$, from that $\cp_k - \wh\cp_k > c_0 \rho^{\k}_{n, p}$.
Collecting the bounds on $IV$, $V$ and $VI$, we have for some constant $C > 0$,
\begin{align*}
\vert \mc R_{k4} \vert \le \frac{Cp\vert \wh\cp_k - \cp_k \vert}{Gm} \le \epsilon \vert \mc F_k \vert
\end{align*}
under Assumption~\ref{assum:common:size}~\ref{cond:common:jump}. 
Also, we observe that by Proposition~\ref{prop:idio:eval},
\begin{align*}
\vert \mc R_{k5} \vert \le 
\l\Vert \bm\Sigma^\prime_{\xi, \wh\cp_k}(\omega(\wh\cp_k), G)
- \bm\Sigma^\prime_{\xi, \cp_k}(\omega(\wh\cp_k), G) \r\Vert
\le \frac{4 M_\xi (\cp_k - \wh\cp_k)}{G} \le \epsilon \vert \mc F_k \vert.
\end{align*}
Turning our attention to $\mc R_{k3}$, note that
\begin{align*}
\frac{G}{\vert \wh\cp_k - \cp_k \vert + m} \l[ \wh{\bm\Sigma}^\prime_{x, \wh\cp_k}(\omega(\wh\cp_k), G) - \wh{\bm\Sigma}^\prime_{x, \cp_k}(\omega(\wh\cp_k), G) \r]_{i, i'}
= \frac{1}{2\pi} \l\{ Q^{(1)}_{k, ii'}(\omega(\wh\cp_k), \cp_k - \wh\cp_k, -G) \r.
\\
\l. - Q^{(2)}_{k, ii'}(\omega(\wh\cp_k), \cp_k - \wh\cp_k, 0)
- Q^{(1)}_{k, ii'}(\omega(\wh\cp_k), \cp_k - \wh\cp_k, 0)
+ Q^{(2)}_{k, ii'}(\omega(\wh\cp_k), \cp_k - \wh\cp_k, G)\r\},
\end{align*}
where the definitions of $Q^{(r)}_{k, ii'}(\omega, h, H)$, $r = 1, 2$,
can be found in~\eqref{eq:def:Q:one}.
Then by Lemma~\ref{lem:Qs} and Chebyshev's inequality,
there exists some constant $c_1 > 0$ such that
$\p(\bar{\mc E}^{(1)}_{n, p}) \to 1$, where
\begin{align}
\bar{\mc E}^{(1)}_{n, p} = \l\{ \max_{1 \le k \le K_\chi} \max_{h \in I_k} \sup_{\omega \in [-\pi, \pi]} \frac{w_k}{p} \l\Vert 
\wh{\bm\Sigma}^\prime_{x, \cp_k + h}(\omega, G) 
- \wh{\bm\Sigma}^\prime_{x, \cp_k}(\omega, G) \r.\r.
\nn \\
\l.\l. - \E\l(\wh{\bm\Sigma}^\prime_{x, \cp_k + h}(\omega, G) 
- \wh{\bm\Sigma}^\prime_{x, \cp_k}(\omega, G)\r)
\r\Vert \le c_1 \wt\psi(\delta) \r\}
\label{eq:set:e:one:tilde}
\end{align}
for some $1 \le \delta \le G$,
where $w_k = (p^{-1} \Delta_{\chi, k})^{-1}$ and 
$I_k$ is defined in the lemma.
Setting $\delta = c_0w_k^{-2}\rho^{\k}_{n, p}$ (which itself does not depend on $k$), 
we have on $\bar{\mc E}^{(1)}_{n, p}$, 
\begin{align*}
& \vert \mc R_{k3} \vert \le 
\frac{c_1 p(\vert \wh\cp_k - \cp_k \vert + m)}{G} \times
\l\{
\begin{array}{l}
\min\l( \frac{m(G K_\chi)^{2/\nu}}{(c_0\rho^{\k}_{n, p})^{1 - 2/\nu}}, \sqrt{\frac{m\log(G K_\chi)}{c_0\rho^{\k}_{n, p}}} \r) 
\\
\quad \text{under Assumption~\ref{assum:innov}~\ref{cond:moment}},
\\
\sqrt{\frac{m\log(G K_\chi)}{c_0\rho^{\k}_{n, p}}} 
\\
\quad \text{under Assumption~\ref{assum:innov}~\ref{cond:gauss}}
\end{array}\r.
\\
& \le \frac{2 c_1 \vert \wh\cp_k - \cp_k \vert \; \Delta_{\chi, k}}{(c_0 \wedge \sqrt{c_0}) G} < (1 - 4\epsilon) \vert \mc F_k \vert
\end{align*}
for $c_0$ large enough which, combined with the bounds on 
$\mc R_{kl}, \, l = 1, 2, 4, 5$, contradicts the first inequality in~\eqref{thm:common:one:eq:four}.
As these statements are deterministic on $\mc E^{(1)}_{n, p} \cap \bar{\mc E}^{(1)}_{n, p}$,
the above arguments apply to all $1 \le k \le K_\chi$
which concludes the proof.
\end{proof}

\subsubsection{Supporting results}

In what follows, we operate 
under the assumptions made in Theorem~\ref{thm:common}.

\begin{lem}
\label{lem:spec:x} 
Let $\wh{\bm\Sigma}_{x, v}(\omega, G) = [\wh\sigma_{x, v, ii'}(\omega, G), \, 1 \le i, i' \le p]$.
There exists a constant $C > 0$ not dependent on $1 \le i, i' \le p$
such that
\begin{align*}
\E\l( \max_{G \le v \le n} \sup_{\omega \in [-\pi, \pi]} 
\l\vert \wh{\sigma}_{x, v, ii'}(\omega, G) - \sigma_{x, v, ii'}(\omega, G) \r\vert^2 \r)
\le C\l(\psi_n \vee \frac{1}{m}\r)^2
\end{align*}
with $\psi_n$ defined in~\eqref{eq:psi}.               
\end{lem}

\begin{proof}
Noting that
\begin{align}
& \E\l( \max_v \sup_{\omega} \l\vert \wh{\sigma}_{x, v, ii'}(\omega, G) - 
\sigma_{x, v, ii'}(\omega, G) \r\vert^2 \r) 
\nn \\
\le & 2 \E\l( \max_v \sup_{\omega} \l\vert \wh{\sigma}_{x, v, ii'}(\omega, G) - \E(\wh\sigma_{x, v, ii'}(\omega, G)) \r\vert^2 \r) 
\nn 
\\
& + 2 \max_v \sup_{\omega} \l\vert \E(\wh{\sigma}_{x, v, ii'}(\omega, G)) - \sigma_{x, v, ii'}(\omega, G) \r\vert^2,
\label{lem:spec:x:eq:one}
\end{align}
we first address the first term in the RHS of \eqref{lem:spec:x:eq:one}.
In Lemma~\ref{lem:func:dep},
for $\varsigma > 2$ (as assumed in 
Assumptions~\ref{assum:idio}~\ref{cond:idio:coef} and~\ref{assum:common}),
we can always set $\alpha = \varsigma - 1 > 1/2 - 2/\nu$.
Then, from the finiteness of $\Phi_{\nu, \alpha}$ shown therein and 
by Theorems 4.1 and~4.2 of \cite{zhang2021},
there exist universal constants $C_1, C_2, C_3 > 0$
and constants $C_\alpha, C_{\nu, \alpha} > 0$ that depend only on their subscripts,
such that for any $z > 0$,
\begin{align*}
& \p\l( \max_v \sup_\omega \l\vert \wh\sigma_{x, v, ii'}(\omega, G) 
- \E(\wh{\sigma}_{x, v, ii'}(\omega, G)) \r\vert \ge z \r) \le
\\
& \l\{\begin{array}{ll}
\frac{C_{\nu, \alpha} nm^{\nu/2} \log^{\nu + 1}(G) \Phi_{\nu, \alpha}^\nu}{(Gz)^{\nu/2}}
+ C_1 m n \exp\l(-\frac{Gz^2}{C_\alpha \Phi_{4, \alpha}^4 m}\r) 
& \text{under Assumption~\ref{assum:innov}~\ref{cond:moment}},
\\
C_2 nm \exp\l[- C_3 \min\l( \frac{G z^2}{m \Phi_{2, 0}^4}, \frac{Gz}{m \Phi_{2, 0}^2} \r) \r] 
& \text{under~Assumption~\ref{assum:innov}~\ref{cond:gauss}.}
\end{array}\r.
\end{align*}
Noting that for any positive random variable $Y$, we have $\E(Y) = \int_0^\infty \p(Y > y) dy$,
we have
$\E( \max_v \sup_\omega \vert \wh{\sigma}_{x, v, ii'}(\omega, G) - 
\E(\wh\sigma_{x, v, ii'}(\omega, G)) \vert^2)
\le C \psi_n^2$ for some constant $C > 0$ independent of $i, i'$, thanks to Lemma~\ref{lem:func:dep}.

Turning our attention to the second term in the RHS of~\eqref{lem:spec:x:eq:one},
let $\wh{\bm\Gamma}_{x, v}(\ell, G) = [\wh\gamma_{x, v, ii'}(\ell, G), \, 1 \le i, i' \le p]$
and define $\wh{\bm\Gamma}_{\chi, v}(\ell, G)$, $\wh\gamma_{\chi, v, ii'}(\ell, G)$,
$\wh{\bm\Gamma}_{\xi, v}(\ell, G)$ and $\wh\gamma_{\xi, v, ii'}(\ell, G)$, analogously. Then,
\begin{align*}
& \l\vert \E(\wh{\gamma}_{x, v, ii'}(\ell, G)) - \gamma_{x, v, ii'}(\ell, G) \r\vert
\le \\
& \underbrace{\l\vert \E(\wh{\gamma}_{\chi, v, ii'}(\ell, G)) - \gamma_{\chi, v, ii'}(\ell, G) \r\vert}_{I}
+
\underbrace{\l\vert \E(\wh{\gamma}_{\xi, v, ii'}(\ell, G)) - \gamma_{\xi, v, ii'}(\ell, G) \r\vert}_{II}.
\end{align*}
Then by Lemma~\ref{lem:decay:two}, for all $\ell$, $1 \le i, i' \le p$ and $G \le v \le n - G$,
\begin{align*}
& I = \frac{1}{G} \l\vert \sum_{t = v - G + 1 + \ell}^v \E(\chi_{i, t - \ell}\chi_{i't}) -
\sum_{k = L_\chi(v - G + 1)}^{L_\chi(v)} \{(\cp_{\chi, k + 1} \wedge v) - (\cp_{\chi, k} \wedge (v - G) )\}
\gamma_{\chi, ii'}^{\k}(\ell) \r\vert
\\
& \le \frac{C_{\Xi, \varsigma} (L_\chi(v) - L_\chi(v - G + 1) + 1) \vert \ell \vert}{G} (1 + \vert \ell \vert)^{-\varsigma}
\le \frac{2 C_{\Xi, \varsigma}}{G} (1 + \vert \ell \vert)^{-\varsigma + 1},
\end{align*}
noting that $L_\chi(v) - L_\chi(v - G + 1) \le 1$ under Assumption~\ref{assum:common:size}~\ref{cond:common:spacing}.
Similarly, we yield $II \le 2 G^{-1} C_{\Xi, \varsigma, \vep} (1 + \vert \ell \vert)^{-\varsigma + 1}$
under Assumption~\ref{assum:idio:size}~\ref{cond:idio:spacing}.
Then,
\begin{align}
& \max_{i, i'} \max_v \sup_\omega 
2\pi \l\vert \E(\wh{\sigma}_{x, v, ii'}(\omega, G)) - \sigma_{x, v, ii'}(\omega, G) \r\vert
\nn \\
& \le 
\max_{i, i'} \max_v \sum_{\ell = -m}^m \l\vert \E(\wh\gamma_{x, v, ii'}(\ell, G)) -
\gamma_{x, v, ii'}(\ell, G) \r\vert + \max_{i, i'} \max_v \sum_{\ell = -m}^m \frac{\vert \ell \vert}{m} \l\vert \gamma_{x, v, ii'}(\ell, G) \r\vert
\nn \\
& 
+ \max_{i, i'} \max_v \sum_{\vert \ell \vert > m} \l\vert \gamma_{x, v, ii'} (\ell, G) \r\vert
=: III + IV + V.
\label{eq:spec:x:bias}
\end{align}
From the bounds on $I$ and $II$ (which hold uniformly over $1 \le i, i' \le p$ and $G \le v \le n$)
and that $\varsigma > 2$, there exists $C^\prime_{\Xi, \varsigma, \vep} > 0$ such that
$III \le C^\prime_{\Xi, \varsigma, \vep} G^{-1} = o(m^{-1})$.
Also from Lemma~\ref{lem:decay} , there exists $C^{\prime\prime}_{\Xi, \varsigma, \vep} > 0$
such that
\begin{align*}
IV \le 2 C_{\Xi, \varsigma, \vep} \sum_{\ell = 1}^m \frac{\ell}{m (1 + \ell)^\varsigma}
\le \frac{2C_{\Xi, \varsigma, \vep}}{m} \sum_{\ell = 1}^m \frac{1}{(1 + \ell)^{\varsigma - 1}}
\le \frac{C^{\prime\prime}_{\Xi, \varsigma, \vep}}{m},
\end{align*}
and $V = O(m^{-\varsigma + 1}) = o(m^{-1})$.
Combining the bounds on $III$--$V$, the proof is complete. 
\end{proof}

For $H \in \{0, \pm G\}$ and $1 \le k \le K_\chi$, define
\begin{align}
Q^{(1)}_{k, ii'}(\omega, h, H) &= \frac{1}{|h| + m} 
\l\{
\sum_{\ell = 0}^m K\l(\frac{\ell}{m}\r) e^{-\iota \ell \omega}
\sum_{t = (\cp_k - h) \wedge \cp_k + H + |\ell| + 1}^{(\cp_k - h) \vee \cp_k + H + |\ell|} X_{i, t - |\ell|} X_{i't} \r.
\nn \\
& + \l. \sum_{\ell = -m}^{-1} K\l(\frac{\ell}{m}\r) e^{-\iota \ell \omega}
\sum_{t = (\cp_k - h) \wedge \cp_k + H + |\ell| + 1}^{(\cp_k - h) \vee \cp_k + H + |\ell|} X_{it} X_{i', t - \vert \ell \vert}\r\},
\nn \\
Q^{(2)}_{k, ii'}(\omega, h, H) &= \frac{1}{|h| + m} \l\{\sum_{\ell = 0}^m K\l(\frac{\ell}{m}\r) e^{-\iota \ell \omega}
\sum_{t = (\cp_k - h) \wedge \cp_k + H + 1}^{(\cp_k - h) \vee \cp_k + H} X_{i, t - |\ell|} X_{i't} \r. 
\nn \\
& + \l. \sum_{\ell = -m}^{-1} K\l(\frac{\ell}{m}\r) e^{-\iota \ell \omega}
\sum_{t = (\cp_k - h) \wedge \cp_k + H + 1}^{(\cp_k - h) \vee \cp_k + H} X_{it} X_{i', t - |\ell|} \r\}.
\label{eq:def:Q:one}
\end{align}

\begin{lem}
\label{lem:Qs}
There exists a constant $C > 0$ not dependent on $1 \le i, i' \le p$ such that 
for some $\delta \in \{m, \ldots, G\}$,
\begin{align*}
& \E\l( \max_{1 \le k \le K_\chi} \max_{1 \le r \le 2} \max_{H \in \{0, \pm G\}} \max_{h \in I_k}
\sup_{\omega \in [-\pi, \pi]} w_k^2 \l\vert Q^{(r)}_{k, ii'}(\omega, h, H) - \E(Q^{(r)}_{k, ii'}(\omega, h, H)) \r\vert^2 \r) \le C \wt\psi(\delta)^2,
\end{align*}
where $w_k = (\Delta_{\chi, k}/p)^{-1}$,
$I_k = \{h: \, w_k^2 \delta \le \vert h \vert \le G\}$ and
\begin{align*}
\wt\psi(\delta) = 
\frac{m(G K_\chi)^{2/\nu}}{\delta^{1 - 2/\nu}} \vee \sqrt{\frac{m\log(G K_\chi)}{\delta}},
\end{align*}
under Assumption~\ref{assum:innov}~\ref{cond:moment}, and
\begin{align*}
\wt\psi(\delta) = \sqrt{\frac{m\log(G K_\chi)}{\delta}}
\end{align*}
and under Assumption~\ref{assum:innov}~\ref{cond:gauss}.
\end{lem}

\begin{proof}
Under Assumption~\ref{assum:innov}~\ref{cond:moment},
Proposition~6.2 of \cite{zhang2021}, combined with the arguments adopted
in the proof of their Theorem~4.1 (most notably, their Equation~(B.15)) 
and Bonferroni correction, obtains that
there exist universal constants $C_1 > 0$ and $C_{\nu, \alpha}, C_\alpha > 0$
that depend only on their subscripts, such that
\begin{align*}
&\p\l(\max_{1 \le k \le K_\chi} \max_{1 \le r \le 2} \max_{H \in \{0, \pm G\}} \max_{h \in I_k}
\sup_{\omega \in [-\pi, \pi]} w_k \l\vert Q^{(r)}_{k, ii'}(\omega, h, H) - \E(Q^{(r)}_{k, ii'}(\omega, h, H)) \r\vert \ge z \r)
\\
& \le \frac{12C_{\nu, \alpha} G K_\chi m^{\nu/2 - 1} (4m + 1) \Psi_{\nu, \alpha}^4}
{\delta^{\nu/2 - 1} z^{\nu/2}}
+ C_1(4m + 1) G K_\chi \exp\l(-\frac{\delta z^2}{C_\alpha m \Psi_{4, \alpha}^4}\r)
\end{align*}
thanks to Lemma~\ref{lem:func:dep}.
Then, as in the proof of Lemma~\ref{lem:spec:x}, we can 
find $C > 0$ independent of $i, i'$ and show the first part of the claim
by Lemma~\ref{lem:func:dep}.
Similarly, under Assumption~\ref{assum:innov}~\ref{cond:gauss},
Lemma~\ref{lem:func:dep} and Theorem~6.3 of \cite{zhang2021} show that
there exists a universal constant $C_2 > 0$ such that
\begin{align*}
&\p\l(\max_{1 \le k \le K_\chi} \max_{1 \le r \le 2} \max_{H \in \{0, \pm G\}} \max_{h \in I_k}
\sup_{\omega \in [-\pi, \pi]} w_k \l\vert Q^{(r)}_{k, ii'}(\omega, h, H) - \E(Q^{(r)}_{k, ii'}(\omega, h, H)) \r\vert \ge z \r)
\\
& \le 24(4m + 1) G K_\chi \exp\l\{- C_2 \min\l( \frac{\delta z^2}{m \Psi_{2, 0}^4},
\frac{\delta z}{m \Psi_{2, 0}^2} \r) \r\},
\end{align*}
which completes the proof.
\end{proof}

\subsection{Proof of Theorem~\ref{thm:common:est}}

We provide a series of supporting results under the assumptions made in Theorem~\ref{thm:common:est}, leading to the proof of the claims.
In what follows, we operate in $\mc M^\chi_{n, p}$.
We define 
\begin{align*}
\check{\psi}_n &= \l\{\begin{array}{ll}
\frac{m (GK_\chi)^{2/\nu}}{G} \vee \sqrt{\frac{m\log(mK_\chi)}{G}} & \text{under Assumption~\ref{assum:innov}~\ref{cond:moment}},
\\
\sqrt{\frac{m\log(mK_\chi)}{G}} & \text{under Assumption~\ref{assum:innov}~\ref{cond:gauss},}
\end{array}\r.
\quad \text{and}
\\
\check{\vartheta}_{n, p} &= \l\{\begin{array}{ll}
\frac{m (K_\chi G p)^{2/\nu} \log^{7/2}(p)}{G} \vee 
\sqrt{\frac{m\log(np)}{G}}
& \text{under Assumption~\ref{assum:innov}~\ref{cond:moment}},
\\
\sqrt{\frac{m\log(np)}{G}}
& \text{under Assumption~\ref{assum:innov}~\ref{cond:gauss}.}
\end{array}\r.
\end{align*}

Also, let $\delta_{\chi, k} = \cp_{\chi, k + 1} - \cp_{\chi, k}$ for $0 \le k \le K_\chi$,
and $\wh\delta_{\chi, k} = \wh\cp_{\chi, k + 1} - \wh\cp_{\chi, k}$
for $0 \le k \le \wh{K}_\chi$,
such that $\wh{\bm\Sigma}_x^{\k}(\omega) = \wh{\bm\Sigma}_{x, \wh\cp_{\chi, k + 1}}(\omega, \wh\delta_{\chi, k})$.

\begin{prop}
\label{prop:spec:x:k}
\begin{enumerate}[label = (\alph*)]
\item \label{prop:spec:x:k:one}
There exists a constant $C > 0$ such that
\begin{align*}
\E\l( \max_{0 \le k \le K_\chi} \sup_{\omega \in [-\pi, \pi]}
\frac{1}{p} \l\vert \wh{\bm\Sigma}_{x, \wh\cp_{\chi, k + 1}}(\omega, \wh\delta_{\chi, k}) - 
\bm\Sigma_{x, \cp_{\chi, k + 1}}(\omega, \delta_{\chi, k} ) \r\vert_2^2 \r)
\le C \l(\check\psi_n \vee \frac{1}{m} \vee \frac{\rho_{n, p}}{G}\r)^2.
\end{align*}
\item \label{prop:spec:x:k:two} Also, we have 
\begin{align*}
\max_{0 \le k \le K_\chi} \sup_{\omega \in [-\pi, \pi]} 
\l\vert \wh{\bm\Sigma}_{x, \wh\cp_{\chi, k + 1}}( \omega, \wh\delta_{\chi, k}) - 
\bm\Sigma_{x, \cp_{\chi, k + 1}}(\omega, \delta_{\chi, k} ) \r\vert_\infty
= O_p\l(\check\vartheta_{n, p} \vee \frac{1}{m} \vee \frac{\rho_{n, p}}{G}\r).
\end{align*}
\end{enumerate}
\end{prop}

\begin{proof}[Proof of~\ref{prop:spec:x:k:one}]
Under Assumption~\ref{assum:common:size}~\ref{cond:common:spacing},
applying Proposition~6.2 and Theorem~6.3 of \cite{zhang2021} with their~(B.15),
there exist universal constants $C_1, C_2 > 0$ not dependent on $1 \le i, i' \le p$
and constants $C_\alpha, C_{\nu, \alpha} > 0$ that depend only on their subscripts,
such that for any $z > 0$,
\begin{align*}
& \p\l( \max_k \sup_\omega \l\vert \wh\sigma_{x, \wh\cp_{\chi, k + 1}, ii'}(\omega, \wh\delta_{\chi, k}) 
- \E\l(\wh{\sigma}_{x, \wh\cp_{\chi, k + 1}, ii'}(\omega, \wh\delta_{\chi, k})\r) \r\vert \ge z \r) \le
\\
& \l\{\begin{array}{ll}
\frac{C_{\nu, \alpha} K_\chi (4m + 1) m^{\nu/2 - 1} \Phi_{\nu, \alpha}^\nu}{G^{\nu/2 - 1} z^{\nu/2}}
+ C_1 K_\chi (4m + 1) \exp\l(-\frac{Gz^2}{C_\alpha \Phi_{4, \alpha}^4 m}\r) 
& \text{under Assumption~\ref{assum:innov}~\ref{cond:moment}},
\\
2 K_\chi (4m + 1) \exp\l[- C_2 \min\l( \frac{G z^2}{m \Phi_{2, 0}^4}, \frac{Gz}{m \Phi_{2, 0}^2} \r) \r] 
& \text{under~Assumption~\ref{assum:innov}~\ref{cond:gauss}}
\end{array}\r.
\end{align*}
thanks to Lemma~\ref{lem:func:dep},
which leads to
\begin{align*}
\E(\max_k \sup_\omega \vert \wh\sigma_{x, \wh\cp_{\chi, k+1}, ii'}(\omega, \wh\delta_{\chi, k}) -
\E(\wh\sigma_{x, \wh\cp_{\chi, k+1}, ii'}(\omega, \wh\delta_{\chi, k})) \vert^2)
\le C\check\psi_n^2.
\end{align*}
Next, we bound the bias term
\begin{align*}
& \max_k \max_{i, i'} \sup_\omega \l\vert \E\l(\wh\sigma_{x, \wh\cp_{\chi, k+1}, ii'}(\omega, \wh\delta_{\chi, k})\r) - 
\sigma_{x, \cp_{\chi, k+1}, ii'}\l( \omega, \delta_{\chi, k} \r) \r\vert
\\
\le  &
\max_k \max_{i, i'} \sup_\omega 
 \l\vert \E\l(\wh\sigma_{x, \wh\cp_{\chi, k+1}, ii'}(\omega, \wh\delta_{\chi, k})\r) 
- \sigma_{x, \wh\cp_{\chi, k+1}, ii'}( \omega, \wh\delta_{\chi, k} ) \r\vert
\\
& + \max_k \max_{i, i'} \sup_\omega 
\l\vert \sigma_{x, \wh\cp_{\chi, k+1}, ii'}( \omega, \wh\delta_{\chi, k} )  - 
\sigma_{x, \cp_{\chi, k+1}, ii'}( \omega, \delta_{\chi, k} ) \r\vert
=: I + II.
\end{align*}
We can show that $I = O(m^{-1})$ by the arguments analogous to
those adopted in bounding the RHS of~\eqref{eq:spec:x:bias}.
Also on $\mc M^\chi_{n, p}$, we have for fixed $\ell \ge 0$,
\begin{align*}
& \l\vert \gamma_{x, \wh\cp_{\chi, k + 1}, ii'}(\ell, \wh\delta_{\chi, k})
- \gamma_{x, \cp_{\chi, k + 1}, ii'}(\ell, \delta_{\chi, k}) \r\vert
\\
\le& \frac{1}{\wh\delta_{\chi, k}} \l(
\l\vert \sum_{t = \cp_{\chi, k} \wedge (\wh\cp_{\chi, k} + \ell) + 1}^{\cp_{\chi, k} \vee (\wh\cp_{\chi, k} + \ell)} \E(X_{i, t - \ell} X_{i't}) \r\vert +
\l\vert \sum_{t = \cp_{\chi, k + 1} \wedge\wh\cp_{\chi, k + 1} + 1}^{\cp_{\chi, k + 1} \vee \wh\cp_{\chi, k + 1}} \E(X_{i, t - \ell} X_{i't}) \r\vert \r)
\\
& + \l\vert \frac{1}{\wh\delta_{\chi, k}} - \frac{1}{\delta_{\chi, k}} \r\vert \;
\l\vert \sum_{t = \cp_{\chi, k} + 1}^{\cp_{\chi, k}} \E(X_{i, t - \ell} X_{i't}) \r\vert
\le 2C_{\Xi, \varsigma, \vep} (1 + \ell)^{-\varsigma} \cdot 
\frac{\vert \ell \vert + \rho^{\k}_{n, p} + \rho^{[k + 1]}_{n, p}}{G}
\end{align*}
by Lemma~\ref{lem:decay},
from which we obtain $II = O(\rho_{n, p}/G)$.
In summary, we obtain
\begin{align*}
\E\l(\max_k \sup_\omega \l\vert \wh\sigma_{x, \wh\cp_{\chi, k + 1}, ii'}(\omega, \wh\delta_{\chi, k})
- \sigma_{x, \cp_{\chi, k + 1}, ii'}(\omega, \delta_{\chi, k}) \r\vert^2 \r)
\le C\l(\check\psi_n \vee \frac{1}{m} \vee \frac{\rho_{n, p}}{G}\r)^2
\end{align*}
for some constant $C > 0$ that does not depend on $1 \le i, i' \le p$,
from which the conclusion follows.
\end{proof}

\begin{proof}[Proof of~\ref{prop:spec:x:k:two}]
Under Assumption~\ref{assum:common:size}~\ref{cond:common:spacing},
applying Theorems~6.1 and~6.3 of \cite{zhang2021} with their~(B.15),
there exist universal constants $C_1, C_2, C_3 > 0$
and constants $C_\alpha, C_{\nu, \alpha} > 0$ that depend only on their subscripts,
such that for any $z > 0$,
\begin{align*}
& \p\l( \max_k \sup_\omega
\l\vert \wh{\bm\Sigma}_{x, \wh\cp_{\chi, k + 1}}(\omega, \wh\delta_{\chi, k}) - 
\E(\wh{\bm\Sigma}_{x, \wh\cp_{\chi, k + 1}}(\omega, \wh\delta_{\chi, k})) \r\vert_\infty \ge z \r) \le
\\
& \l\{\begin{array}{l}
\frac{C_{\nu, \alpha} K_\chi (4m + 1) m^{\nu/2 - 1} (\log^{7/4}(p) p^{1/\nu})^\nu}{G^{\nu/2 - 1}z^{\nu/2}}
+ C_1 K_\chi(4m + 1) p^2 \exp\l(-\frac{Gz^2}{C_\alpha m \Phi_{4, \alpha}^4}\r) 
\\
\quad \text{under Assumption~\ref{assum:innov}~\ref{cond:moment}},
\\
C_2 K_\chi(4m + 1) p^2 \exp\l[- C_3 \min\l( \frac{G z^2}{m \Phi_{2, 0}^4}, \frac{Gz}{m \Phi_{2, 0}^2} \r) \r] 
\\
\quad \text{under~Assumption~\ref{assum:innov}~\ref{cond:gauss}}
\end{array}\r.
\end{align*}
thanks to Lemma~\ref{lem:func:dep}, such that 
$\vert \wh{\bm\Sigma}_{x, \wh\cp_{\chi, k + 1}}(\omega, \wh\delta_{\chi, k}) - 
\E(\wh{\bm\Sigma}_{x, \wh\cp_{\chi, k + 1}}(\omega, \wh\delta_{\chi, k})) \vert_\infty = O_p(\check\vartheta_{n, p})$.
We can bound the bias term
$\max_k \sup_\omega 
\vert \E(\wh{\bm\Sigma}_{x, \wh\cp_{\chi, k + 1}}(\omega, \wh\delta_{\chi, k}))
- \bm\Sigma_{x, \cp_{\chi, k + 1}}(\omega, \delta_{\chi, k}) \vert_\infty$
as in the proof of~\ref{prop:spec:x:k:one}, and the conclusion follows.
\end{proof}

We denote by $\mbf e^{\k}_{\chi, j}(\omega), \, 1 \le j \le q_k$,
the eigenvectors of $\bm\Sigma^{\k}_\chi(\omega)$
that correspond to $\mu^{\k}_{\chi, j}(\omega)$,
an let $\bm{\mc M}^{\k}_\chi(\omega) = \text{diag}(\mu^{\k}_{\chi, j}(\omega), \, 1 \le j \le q_k)$
and $\mbf E^{\k}_\chi(\omega) = [\mbf e^{\k}_{\chi, j}(\omega), \, 1 \le j \le q_k]$
for all $0 \le k \le K_\chi$.
Similarly, $\wh{\bm{\mc M}}^{\k}_x(\omega_l) \in \R^{q_k \times q_k}$ is a diagonal matrix with the $q_k$ largest eigenvalues of $\wh{\bm\Sigma}_x^{\k}(\omega_l)$ in its diagonal and $\wh{\mbf E}^{\k}_x(\omega_l) \in \R^{p \times q_k}$ consists of the corresponding $q_k$ eigenvectors.
\begin{lem}
\label{lem:dk}
There exists a unitary, diagonal matrix 
$\mbf O_k(\omega) \in \C^{q_k \times q_k}$ for each $\omega \in [-\pi, \pi]$
and $0 \le k \le K_\chi$, such that
\begin{align*}
\max_{0 \le k \le K_\chi} \sup_{\omega \in [-\pi, \pi]} 
\l\vert \wh{\mbf E}^{\k}_{x}(\omega) - \mbf E^{\k}_\chi(\omega) \mbf O_k(\omega) \r\vert_2 = 
O_p\l( \check\psi_n \vee \frac{1}{m} \vee \frac{\rho_{n, p}}{G} \vee \frac{1}{p}\r). 
\end{align*}
\end{lem}
\begin{proof}
By Propositions~\ref{prop:idio:eval} and~\ref{prop:spec:x:k}~\ref{prop:spec:x:k:one},
we have
\begin{align}
& \max_k \sup_{\omega} \frac{1}{p} \l\Vert 
\wh{\bm\Sigma}_{x, \wh\cp_{\chi, k + 1}}(\omega, \wh\delta_{\chi, k}) 
- \bm\Sigma^{\k}_\chi(\omega) \r\Vert
\le 
\max_k \sup_\omega \frac{1}{p} \l\Vert 
\wh{\bm\Sigma}_{x, \wh\cp_{\chi, k + 1}}(\omega, \wh\delta_{\chi, k}) 
- \bm\Sigma_{x, \cp_{\chi, k + 1}}(\omega, \delta_{\chi, k})\r\Vert
\nn \\
& + \max_k \sup_\omega \frac{1}{p} \l\Vert \bm\Sigma_{\xi, \cp_{\chi, k + 1}}(\omega, \delta_{\chi, k})\r\Vert
= O_p\l( \check\psi_n \vee \frac{1}{m} \vee \frac{\rho_{n, p}}{G} \vee \frac{1}{p} \r).
\label{lem:dk:eq:cov}
\end{align}
Then by Theorem~2 of \cite{yu2015}, 
there exist such $\mbf O_k(\omega)$ satisfying
\begin{align*}
\l\vert \wh{\mbf E}^{\k}_{x}(\omega) - \mbf E^{\k}_\chi(\omega) \mbf O_k(\omega) \r\vert_2 \le 
\frac{2^{3/2} q^{1/2} \l\Vert \wh{\bm\Sigma}_{x, \wh\cp_{\chi, k+1}}(\omega, \wh\delta_{\chi, k}) - \bm\Sigma^{\k}_\chi(\omega) \r\Vert}
{\mu^{\k}_{\chi, q_k}(\omega)}
\end{align*}
for all $\omega$ and $k$ which, 
combined with~\eqref{lem:dk:eq:cov} and Assumption~\ref{assum:factor:two},
concludes the proof.
\end{proof}

\begin{lem}
\label{lem:evals}
\begin{align*}
\max_{0 \le k \le K_\chi} \sup_{\omega \in [-\pi, \pi]}
\l\Vert \l(\frac{\wh{\bm{\mc M}}^{\k}_x(\omega)}{p}\r)^{-1} - 
\l(\frac{\bm{\mc M}^{\k}_\chi(\omega)}{p}\r)^{-1} \r\Vert
&= O_p\l( \check\psi_n \vee \frac{1}{m} \vee \frac{\rho_{n, p}}{G} \vee \frac{1}{p} \r).
\end{align*}
\end{lem}

\begin{proof}
Let $\wh\mu^{\k}_{x, j}(\omega)$ denote the $j$th largest eigenvalue of 
$\wh{\bm\Sigma}_{x, \wh\cp_{\chi, k + 1}}(\omega, \wh\delta_{\chi, k})$
(and thus the $j$th diagonal element of $\wh{\bm{\mc M}}^{\k}_x(\omega)$).
As a consequence of~\eqref{lem:dk:eq:cov} and Weyl's inequality, 
for all $1 \le j \le q_k$ and $\omega \in [-\pi, \pi]$, 
\begin{align}
\label{eq:first:q:eval}
\max_k \sup_\omega \frac{1}{p} \l\vert \wh\mu^{\k}_{x, j}(\omega) - \mu^{\k}_{\chi, j}(\omega) \r\vert
&\le \max_k \sup_\omega \frac{1}{p} \l\Vert \wh{\bm\Sigma}_{x, \wh\cp_{\chi, k + 1}}(\omega, \wh\delta_{\chi, k}) - \bm\Sigma^{\k}_\chi(\omega) \r\Vert
\nn \\
&= O_p\l( \check\psi_n \vee \frac{1}{m} \vee \frac{\rho_{n, p}}{G} \vee \frac{1}{p} \r).
\end{align}
Also from Assumption~\ref{assum:factor:two}, 
there exists $\alpha^{\k}_{q_k}(\omega)$ such that 
$p^{-1} \mu^{\k}_{\chi, q_k}(\omega) \ge \alpha_{q_k}(\omega)$ and 
thus $p^{-1} \wh{\mu}^{\k}_{x, q_k}(\omega) \ge \alpha_{q_k}(\omega) + 
O_p(\check\psi_n \vee m^{-1} \vee G^{-1}\rho_{n, p} \vee p^{-1})$,
which implies that the matrix $p^{-1}{\bm{\mc M}^{\k}_{\chi}(\omega)}$ 
is invertible and the inverse of $p^{-1}\wh{\bm{\mc M}}^{\k}_x(\omega)$ 
exists with probability tending to one as $n, p \to \infty$. Therefore,
\begin{align*}
\l\Vert \l(\frac{\bm{\mc M}^{\k}_{\chi}(\omega)}{p}\r)^{-1} \r\Vert = \frac{p}{\mu^{\k}_{\chi, q_k}(\omega)}
\quad \text{and} \quad
\l\Vert \l(\frac{\wh{\bm{\mc M}}^{\k}_x(\omega)}{p}\r)^{-1} \r\Vert = 
\frac{1}{p^{-1} \mu^{\k}_{\chi, q_k}(\omega)(1 + o_p(1))}.
\end{align*}
Then from~\eqref{eq:first:q:eval}, we have for all $\omega$,
\begin{align*}
& \l\Vert \l(\frac{\wh{\bm{\mc M}}^{\k}_x(\omega)}{p}\r)^{-1} - 
\l(\frac{\bm{\mc M}^{\k}_\chi(\omega)}{p}\r)^{-1} \r\Vert
= \sqrt{p^2 \sum_{j = 1}^{q_k} 
\l(\frac{1}{\wh\mu_{x, j}^{\k}(\omega)} - \frac{1}{\mu^{\k}_{\chi, j}(\omega)}\r)^2}
\\
\le& \sum_{j = 1}^{q_k} \frac{p^{-1}\vert \wh\mu^{\k}_{x, j}(\omega) - \mu^{\k}_{\chi, j}(\omega)\vert}
{p^{-1}\wh\mu^{\k}_{x, j}(\omega) \cdot p^{-1}\mu^{\k}_{\chi, j}(\omega)}
= O_p\l( q_k\l(\check\psi_n \vee \frac{1}{m} \vee \frac{\rho_{n, p}}{G} \vee \frac{1}{p} \r) \r)
\end{align*}
where $O_p$ holds uniformly over $\omega$ and $k$.
\end{proof}

Let $\bm\varphi_i$ denote a vector whose $i$th element is one
and the rest are set to be zero; 
its length are determined by the context. 
\begin{lem}
\label{lem:evec:vec}
\begin{align*}
\sqrt{p} \max_{0 \le k \le K_\chi} \max_{1 \le i \le p} \sup_{\omega \in [-\pi, \pi]}
\l\vert \bm\varphi_i^\top \l(\wh{\mbf E}^{\k}_{x}(\omega) - 
\mbf E^{\k}_{\chi}(\omega)\mbf O_k(\omega) \r) \r\vert_2
= O_p\l( \check\vartheta_{n, p} \vee \frac{1}{m} \vee \frac{\rho_{n, p}}{G} \vee \frac{1}{\sqrt p}\r).
\end{align*}
\end{lem}

\begin{proof}
By Propositions~\ref{prop:idio:eval} and~\ref{prop:spec:x:k}~\ref{prop:spec:x:k:two}, we have
\begin{align}
& \frac{1}{\sqrt p} \max_k \max_i \sup_\omega
\l\vert \bm\varphi_i^\top \l(\wh{\bm\Sigma}_{x, \wh\cp_{\chi, k + 1}}(\omega, \wh\delta_{\chi, k}) 
- \bm\Sigma^{\k}_{\chi}(\omega)\r) \r\vert_2
\nn \\
& \le 
\frac{1}{\sqrt p} \max_k \max_i \sup_\omega
\l\vert \bm\varphi_i^\top \l(\wh{\bm\Sigma}_{x, \wh\cp_{\chi, k + 1}}(\omega, \wh\delta_{\chi, k}) 
- \bm\Sigma_{x, \cp_{\chi, k + 1}}(\omega, \delta_{\chi, k}) \r) \r\vert_2
\nn \\
& + \frac{1}{\sqrt p} \max_k \max_i \sup_\omega
\l\Vert \bm\Sigma_{\xi, \cp_{\chi, k + 1}}(\omega, \delta_{\chi, k}) \r\Vert
= O_p\l( \check\vartheta_{n, p} \vee \frac{1}{m} \vee \frac{\rho_{n, p}}{G} \vee \frac{1}{\sqrt p}\r).
\label{lem:evec:vec:eq:one}
\end{align}

Then, by~\eqref{lem:evec:vec:eq:one}, Assumption~\ref{assum:factor:two}
and Lemmas~\ref{lem:sigma:bound}, \ref{lem:dk} and~\ref{lem:evals}, we have 
\begin{align*}
& \sqrt p \max_{k, i} \sup_\omega
\l\vert \bm\varphi_i^\top \l(\wh{\mbf E}^{\k}_{x}(\omega) 
- \mbf E^{\k}_{\chi}(\omega)\mbf O_k(\omega)\r) \r\vert_2 
\\
=& \frac 1 {\sqrt p} \max_{k, i} \sup_\omega \l\vert
\bm\varphi_i^\top \l[ \wh{\bm\Sigma}_{x, \wh\cp_{\chi, k+1}}(\omega, \wh\delta_{\chi, k}) \wh{\mbf E}^{\k}_{x}(\omega) 
\l( \frac{\wh{\bm{\mc M}}^{\k}_x(\omega)}{p} \r)^{-1}
\r. \r.
\\
& \qquad \qquad \qquad \qquad \l.\l. - \bm\Sigma^{\k}_{\chi}(\omega) \mbf E^{\k}_{\chi}(\omega) 
\l( \frac{\bm{\mc M}_{\chi}(\omega)}{p}\r)^{-1} \mbf O_k(\omega) \r] \r\vert_2 
\\
\le& \max_{k, i} \sup_\omega \l\{
\frac 1 {\sqrt p} \l\vert
\bm\varphi_i^\top\l(\wh{\bm\Sigma}_{x, \wh\cp_{\chi, k+1}}(\omega, \wh\delta_{\chi, k}) - 
\bm\Sigma^{\k}_\chi(\omega)\r) \r\vert_2 \,
\l\Vert\l(\frac{\wh{\bm{\mc M}}^{\k}_x(\omega)}{p}\r)^{-1}\r\Vert \r.
\\
& + 
\frac 1 {\sqrt p} \l\vert \bm\varphi_i^\top \bm\Sigma^{\k}_\chi(\omega) \r\vert_2\,
\l\Vert\l(\frac{\wh{\bm{\mc M}}^{\k}_x(\omega)}{p}\r)^{-1} - 
\l(\frac{\bm{\mc M}^{\k}_{\chi}(\omega)}{p}\r)^{-1}\r\Vert 
\\ 
& + \l. \frac 1 {\sqrt p} \l\vert \bm\varphi_i^\top\bm\Sigma^{\k}_\chi(\omega) \r\vert_2 \,
\l\Vert \l(\frac{\bm{\mc M}^{\k}_{\chi}(\omega)}{p} \r)^{-1} \r\Vert\,
\l\Vert \wh{\mbf E}^{\k}_{x}(\omega) - \mbf E^{\k}_{\chi}(\omega)\mbf O_k(\omega)\r\Vert \r\}
\\
&= O_p\l( \check\vartheta_{n, p} \vee \frac{1}{m} \vee \frac{\rho_{n, p}}{G} \vee \frac{1}{\sqrt p}\r).
\end{align*} 
\end{proof}

\begin{lem}
\label{lem:evec:size}
Let $\mbf e^{\k}_{\chi, j}(\omega)$ (resp. $\wh{\mbf e}^{\k}_{x, j}(\omega)$)
the $j$th column of $\mbf E^{\k}_{\chi}(\omega)$ (resp. $\wh{\mbf E}^{\k}_{x}(\omega)$)
and $e^{\k}_{\chi, ij}(\omega)$ (resp. $\wh e^{\k}_{x, ij}(\omega)$) denote its $i$th element.
\begin{enumerate}[label = (\alph*)]
\item \label{lem:evec:size:one} $\max_{0 \le k \le K_\chi} \max_{1 \le j \le q_k} 
\sup_{\omega \in [-\pi, \pi]}
\frac{1}{\sqrt{\mu^{\k}_{\chi, j}(\omega)}} \max_{1 \le i \le p} \vert e^{\k}_{\chi, ij}(\omega) \vert = O(1)$.

\item \label{lem:evec:size:three} If $\check\vartheta_{n, p} \to 0$ and $\rho_{n, p}/G \to 0$ 
as $n, p \to \infty$, we have
\begin{align*}
\max_{0 \le k \le K_\chi} \max_{1 \le j \le q_k} \sup_{\omega \in [-\pi, \pi]}  
\frac{1}{\sqrt{\wh\mu^{\k}_{\chi, j}(\omega)}} \max_{1 \le i \le p} \l\vert \wh{e}^{\k}_{x, ij}(\omega) \r\vert = O_p(1).
\end{align*}
\end{enumerate}
\end{lem}

\begin{proof}
Note that by 
the arguments adopted in the proof of Lemma~\ref{lem:sigma:bound},
$\sigma_{x, \cp_{\chi, k + 1}, ii'}(\omega, \delta_{\chi, k}) = \sigma^{\k}_{\chi, ii'}(\omega) + 
\sigma_{\xi, \cp_{\chi, k + 1}, ii'}(\omega, \delta_{\chi, k})$
and $\max_{k, i, i'} \sup_\omega \sigma_{x, \cp_{\chi, k + 1}, ii'}(\omega, \delta_{\chi, k}) \le B_\sigma < \infty$. 
Then from that
$\sigma^{\k}_{\chi, ii}(\omega)
= \sum_{j = 1}^{q_k} \mu^{\k}_{\chi, j}(\omega) \vert e^{\k}_{\chi, ij}(\omega) \vert^2 \le B_\sigma$,
the claim~\ref{lem:evec:size:one} follows. 
Next, by Proposition~\ref{prop:spec:x:k}~\ref{prop:spec:x:k:two} and Lemma~\ref{lem:sigma:bound},
\begin{align*}
\max_{k, i} \sum_{j = 1}^{q_k} \wh\mu^{\k}_{x, j}(\omega) \vert \wh e^{\k}_{x, ij}(\omega) \vert^2 
\le
\max_{k, i} \sup_\omega \wh\sigma_{x, \wh\cp_{\chi, k + 1}, ii}(\omega, \wh\delta_{\chi, k})   
\le B_\sigma + O_p\l(\check\vartheta_{n, p} \vee \frac{1}{m} \vee \frac{\rho_{n, p}}{G}\r)
\end{align*}
which, combined with~\eqref{eq:first:q:eval}, leads to~\ref{lem:evec:size:three}.
\end{proof} 

\begin{proof}[Proof of Theorem~\ref{thm:common:est}]
First, note that
\begin{align*}
& \max_k \sup_\omega 
\l\vert \wh{\bm\Sigma}^{\k}_{\chi}(\omega) - \bm\Sigma^{\k}_{\chi}(\omega) \r\vert_\infty
=
\max_k \max_{i, i'} \sup_\omega 
\l\vert \bm\varphi_i^\top \l( \wh{\bm\Sigma}^{\k}_{\chi}(\omega) - \bm\Sigma^{\k}_{\chi}(\omega) \r)
\bm\varphi_{i'} \r\vert 
\\
\le & \max_k \max_{i, i'} \sup_\omega 
\l\{
\l\vert \bm\varphi_i^\top \l(\wh{\mbf E}^{\k}_{x}(\omega) - \mbf E^{\k}_{\chi}(\omega) \mbf O_k(\omega) \r)
\wh{\bm{\mc M}}^{\k}_x(\omega) (\wh{\mbf E}^{\k}_{x}(\omega))^* \bm\varphi_{i'} \r\vert
\r.
\\
& +
\l\vert \bm\varphi_i^\top \mbf E^{\k}_{\chi}(\omega) \mbf O_k(\omega)
\l(\wh{\bm{\mc M}}^{\k}_x(\omega) - \bm{\mc M}^{\k}_\chi(\omega)\r) (\wh{\mbf E}^{\k}_{x}(\omega))^* \bm\varphi_{i'} \r\vert
\\
& +
\l.
\l\vert \bm\varphi_i^\top \mbf E^{\k}_{\chi}(\omega) \mbf O_k(\omega)
\bm{\mc M}^{\k}_{\chi}(\omega) \l(\wh{\mbf E}^{\k}_{x}(\omega) - 
\mbf E^{\k}_\chi(\omega) \mbf O_k(\omega)\r)^* \bm\varphi_{i'} \r\vert \r\} =: I + II + III.
\end{align*}
By Lemmas~\ref{lem:evals}, \ref{lem:evec:vec}, \ref{lem:evec:size} and
Cauchy-Schwarz inequality,
\begin{align*}
& I = \max_k \max_{i, i'} \sup_\omega 
\l\vert \sum_{j = 1}^{q_k} \wh{\mu}^{\k}_{x, j}(\omega) (\wh{e}^{\k}_{x, i'j})^* \cdot
\bm\varphi_i^\top \l(\wh{\mbf E}^{\k}_{x}(\omega) - \mbf E^{\k}_{\chi}(\omega) \mbf O_k(\omega) \r) \bm\varphi_j \r\vert
\\
& \le \max_k \sup_\omega 
\sqrt{p} \max_i \l\vert \bm\varphi_i^\top \l(\wh{\mbf E}^{\k}_{x}(\omega) - \mbf E^{\k}_{\chi}(\omega) \mbf O_k(\omega) \r) \r\vert_2 \cdot 
\sqrt{\frac{1}{p} \max_{i'} \sum_{j = 1}^{q_k} 
(\wh{\mu}^{\k}_{x, j}(\omega))^2 \vert \wh{e}^{\k}_{x, i'j}(\omega) \vert^2}
\\
&= O_p\l( \l(\check\vartheta_{n, p} \vee \frac{1}{m} \vee \frac{\rho_{n, p}}{G} \vee \frac{1}{\sqrt p} \r) \cdot 
\max_k \sqrt{\frac{1}{p} \sum_{j = 1}^{q_k} \wh{\mu}^{\k}_{x, j}(\omega)} \r)
= O_p\l( \check\vartheta_{n, p} \vee \frac{1}{m} \vee \frac{\rho_{n, p}}{G} \vee \frac{1}{\sqrt p} \r)
\end{align*}
under Assumption~\ref{assum:factor:two},
and $III$ can be handled analogously.
By~\eqref{eq:first:q:eval} and Lemma~\ref{lem:evec:size},
\begin{align*}
II \le & \max_k \sup_\omega 
\sqrt{p} \max_i \l\vert \bm\varphi_i^\top \mbf E^{\k}_{\chi}(\omega) \r\vert_2 \cdot
\frac{1}{p} \l\Vert \wh{\bm{\mc M}}^{\k}_x(\omega) - \bm{\mc M}^{\k}_\chi(\omega) \r\Vert \cdot
\sqrt{p} \max_{i'} \l\vert \bm\varphi_{i'}^\top \wh{\mbf E}^{\k}_{x}(\omega) \r\vert_2
\\
=& O_p\l( \l(\check\vartheta_{n, } \vee \frac{1}{m} \vee \frac{\rho_{n, p}}{G} \vee \frac{1}{p} \r) 
\cdot 
\sqrt{\sum_{j = 1}^{q_k} \frac{p}{\mu^{\k}_{\chi, j}(\omega)} \cdot
\sum_{j = 1}^{q_k} \frac{p}{\wh\mu^{\k}_{x, j}(\omega)}}\r)
\\
=& O_p\l( q_k\l(\check\vartheta_{n, } \vee \frac{1}{m} \vee \frac{\rho_{n, p}}{G} \vee \frac{1}{p} \r)\r)
\end{align*}
under Assumption~\ref{assum:factor:two}.
In summary, we have
\begin{align}
\label{eq:thm:common:est:one}
\max_{0 \le k \le K_\chi} \sup_{\omega \in [-\pi, \pi]} \l\vert \wh{\bm\Sigma}^{\k}_\chi(\omega) - \bm\Sigma^{\k}_\chi(\omega) \r\vert_\infty
= O_p\l( \check{\vartheta}_{n, p} \vee \frac{1}{m} \vee \frac{\rho_{n, p}}{G} \vee \frac{1}{\sqrt p} \r).
\end{align}
Next, let $\bm\Sigma^{\k}_\chi(\omega) = [\sigma^{\k}_{\chi, ii'}(\omega), \, 1 \le i, i' \le p]$,
$\wh{\bm\Sigma}^{\k}_\chi(\omega) = [\wh{\sigma}^{\k}_{\chi, ii'}(\omega), \, 1 \le i, i' \le p]$,
$\bm\Gamma^{\k}_\chi(\ell) = [\gamma^{\k}_{\chi, ii'}(\ell), \, 1 \le i, i' \le p]$ and
$\wh{\bm\Gamma}^{\k}_\chi(\ell) = [\wh{\gamma}^{\k}_{\chi, ii'}(\ell), \, 1 \le i, i' \le p]$.
Note that
\begin{align*}
& \max_k \max_{i, i'} \max_\ell \l\vert \wh{\gamma}^{\k}_{\chi, ii'}(\ell) - \gamma^{\k}_{\chi, ii'}(\ell) \r\vert
= \l\vert \frac{2\pi}{2m + 1} \sum_{l = -m}^m \wh\sigma^{\k}_{\chi, ii'}(\omega_l) 
e^{\iota \omega_l \ell}
- \int_{-\pi}^\pi \sigma^{\k}_{\chi, ii'}(\omega) e^{\iota \omega \ell} d\omega \r\vert
\\
& \le \max_k \max_{i, i'} \max_\ell
\frac{2\pi}{2m + 1} \sum_{l = -m}^m \l\vert \wh\sigma^{\k}_{\chi, ii'}(\omega_l) - \sigma^{\k}_{\chi, ii'}(\omega_l) \r\vert
\\
& + \max_k \max_{i, i'} \max_\ell
\l\vert \frac{2\pi}{2m + 1} \sum_{l = -m}^m \sigma^{\k}_{\chi, ii'}(\omega_k) e^{\iota \omega_k \ell}
- \int_{-\pi}^\pi \sigma^{\k}_{\chi, ii'}(\omega) e^{\iota \omega \ell d\omega} \r\vert =: III + IV
\end{align*}
where by~\eqref{eq:thm:common:est:one}, 
\begin{align*}
III \le 2\pi \max_k \sup_\omega \l\vert \wh{\bm\Sigma}_{\chi}^{\k}(\omega) - \bm\Sigma^{\k}_\chi(\omega) \r\vert
= O_p\l(\check\vartheta_{n, } \vee \frac{1}{m} \vee \frac{\rho_{n, p}}{G} \vee \frac{1}{p} \r).
\end{align*}
Next, we can find $\{\omega_l^*\}_{l = -m}^{m - 1}$
and $\{\omega_l^\circ\}_{l = -m}^{m - 1}$
with $\omega_l^*, \omega_l^\circ \in [\omega_l, \omega_{l + 1}]$,
such that
\begin{align*}
& \l\vert \frac{2\pi}{2m + 1} \sum_{l = -m}^m \sigma^{\k}_{\chi, ii'}(\omega_k) e^{\iota \omega_k \ell}
- \int_{-\pi}^\pi \sigma^{\k}_{\chi, ii'}(\omega) e^{\iota \omega \ell d\omega} \r\vert  
\\
& \le  
\frac{2\pi}{2m + 1} \sum_{l = -m}^{m - 1} \max_{\omega_l \le \omega \le \omega_{l + 1}}
\l\vert \sigma^{\k}_{\chi, ii'}(\omega_l) e^{\iota \omega_l \ell} - \sigma^{\k}_{\chi, ii'}(\omega) e^{\iota \omega \ell} \r\vert
\\
& \le 
\frac{2\pi}{2m + 1} \sum_{l = -m}^{m - 1} \max_{\omega_l \le \omega \le \omega_{l + 1}}
\l\vert \sigma^{\k}_{\chi, ii'}(\omega_l) - \sigma^{\k}_{\chi, ii'}(\omega) \r\vert
\\
& + 
\frac{2\pi \max_{1 \le i, i' \le p} \sup_\omega \vert \sigma_{\chi, ii'}(\omega) \vert}{2m + 1} 
\sum_{l = -m}^{m - 1} \max_{\omega_l \le \omega \le \omega_{l + 1}}
\l\vert e^{\iota \omega_l \ell} - e^{\iota \omega \ell} \r\vert 
\\
& \le 
\frac{2\pi}{2m + 1} \sum_{l = -m}^{m - 1} 
\l( \l\vert \sigma^{\k}_{\chi, ii'}(\omega_l) - \sigma^{\k}_{\chi, ii'}(\omega_l^*) \r\vert +
\l\vert \sigma^{\k}_{\chi, ii'}(\omega_{l + 1}) - \sigma^{\k}_{\chi, ii'}(\omega_l^*) \r\vert \r) 
\\
& + 
\frac{2\pi B_\sigma}{2m + 1} 
\sum_{l = -m}^{m - 1} \l( \l\vert e^{\iota \omega_l \ell} - e^{\iota \omega_l^\circ \ell} \r\vert 
+ \l\vert e^{\iota \omega_{l + 1} \ell} - e^{\iota \omega_l^\circ \ell} \r\vert \r)
=: V + VI,
\end{align*}
where the last inequality follows from Lemma~\ref{lem:sigma:bound}.
Then by Lemma~\ref{lem:sigma:deriv},
$V = O(m^{-1})$ uniformly over $1 \le i, i' \le p$ and $0 \le k \le K_\chi$.
Also, as the exponential function has bounded variation,
$VI = O(m^{-1})$ uniformly in $0 \le \ell \le d$ for some finite $d$.
Putting together the bounds on $V$ and $VI$ gives the bound on $IV$.
\end{proof}

\subsection{Proof of Theorem~\ref{thm:idio} and Corollary~\ref{cor:idio}}

Recall that
\begin{align}
\label{eq:common:tv:acv:est}
\wh{\bm\Gamma}_{\chi, v}(\ell, G) &= \frac{1}{G} \sum_{k = \wh{L}_\chi(v - G + 1)}^{\wh L_\chi(v)}
\{(\wh\cp_{\chi, k + 1} \wedge v) - (\wh\cp_{\chi, k} \vee (v - G))\} \wh{\bm\Gamma}_\chi^{\k}(\ell) 
\end{align} 
with $\wh L_\chi(v) = \max\{0 \le k \le \wh K_\chi: \, \wh\cp_{\chi, k} + 1 \le v\}$.

\begin{prop}
\label{prop:idio:est}
Under the assumptions made in Theorem~\ref{thm:common:est}, we have on $\mc M^\chi_{n, p}$,
\begin{align*}
\max_{G \le v \le n} \max_{0 \le \ell \le d} \l\vert 
\wh{\bm\Gamma}_{\xi, v}(\ell, G) - \bm\Gamma_{\xi, v}(\ell, G) \r\vert_\infty
= O_p\l(\vartheta_{n, p} \vee \frac{1}{m} \vee \frac{\rho_{n, p}}{G} \vee \frac{1}{\sqrt p}\r).
\end{align*}
\end{prop}
\begin{proof}
By definition, we have
\begin{align*}
\max_{v, \ell} \l\vert \wh{\bm\Gamma}_{\xi, v}(\ell, G) - \bm\Gamma_{\xi, v}(\ell, G) \r\vert_\infty
\le & \max_{v, \ell} \l\vert \wh{\bm\Gamma}_{x, v}(\ell, G) - \bm\Gamma_{x, v}(\ell, G) \r\vert_\infty
\\
&+
\max_{v, \ell} \l\vert \wh{\bm\Gamma}_{\chi, v}(\ell, G) - \bm\Gamma_{\chi, v}(\ell, G) \r\vert_\infty
=:  I + II
\end{align*}
where, from Lemma~\ref{lem:acv:x:max}, we have
$I = O_p(\bar{\vartheta}_{n, p})$ with $\bar{\vartheta}_{n, p}$ defined therein.
Also, 
\begin{align*}
&II \le \max_{v, \ell} \l\vert \frac{1}{G} \sum_{k = \wh L_\chi(v - G + 1)}^{\wh L_\chi(v)}
\l\{(\wh\cp_{\chi, k + 1} \wedge v) - (\wh\cp_{\chi, k} \vee (v - G)) \r\}
\l(\wh{\bm\Gamma}^{\k}_\chi(\ell) - \bm\Gamma^{\k}_\chi(\ell) \r) \r\vert_\infty 
\\
&\qquad + \max_{v, \ell} \frac{1}{G} \l\vert \sum_{k = \wh L_\chi(v - G + 1)}^{\wh L_\chi(v)} 
\l\{(\wh\cp_{\chi, k + 1} \wedge v) - (\wh\cp_{\chi, k} \vee (v - G)) \r\} \bm\Gamma^{\k}_\chi(\ell) - \r.
\\
&\l. \sum_{k = L_\chi(v - G + 1)}^{L_\chi(v)} 
\l\{(\cp_{\chi, k + 1} \wedge v) - (\cp_{\chi, k} \vee (v - G)) \r\} \bm\Gamma^{\k}_\chi(\ell) 
\r\vert_\infty 
= O_p\l(\check{\vartheta}_{n, p} \vee \frac{1}{m} \vee \frac{\rho_{n, p}}{G} \vee \frac{1}{\sqrt p}\r)
\end{align*}
on $\mc M^\chi_{n, p}$, from Theorem~\ref{thm:common:est}.
The conclusion follows by noting that $\check\vartheta_{n, p} \vee \bar{\vartheta}_{n, p} =
O(\vartheta_{n, p})$.
\end{proof}

A consequence of Proposition~\ref{prop:idio:est} is that 
$\p(\mc E^{(2)}_{n, p}) \to 1$, where
\begin{align*}
\mc E^{(2)}_{n, p} = \l\{
\max_{G \le v \le n} \max_{0 \le \ell \le d} 
\l\vert \wh{\bm\Gamma}_{\xi, v}(\ell, G) - \bm\Gamma_{\xi, v}(\ell, G) \r\vert_\infty
\le M\l(\vartheta_{n, p} \vee \frac{1}{m} \vee \frac{\rho_{n, p}}{G} \vee \frac{1}{\sqrt p}\r)
\r\}
\end{align*}
with $M$ as in~\eqref{eq:lambda}.

\begin{prop}
\label{prop:idio:beta}
Under the assumptions made in Theorem~\ref{thm:common:est}, 
with $\lambda_{n, p}$ chosen as in~\eqref{eq:lambda}, 
we have on $\mc M^\chi_{n, p} \cap \mc E^{(2)}_{n, p}$,
\begin{align*}
\l\vert \bbG^{\k} \l(\wh{\bm\beta}_v(G) - \bm\beta^{\k}\r) \r\vert_\infty \le 2 \lambda_{n, p}
\quad \text{and} \quad 
\l\Vert \wh{\bm\beta}_v(G) \r\Vert_1 \le \l\Vert \bm\beta^{\k} \r\Vert_1
\end{align*}
for all $\cp_{\xi, k} + G \le v \le \cp_{\xi, k + 1}$ and $0 \le k \le K_\xi$.
\end{prop}

\begin{proof}
We first note that solving~\eqref{eq:ds} is equivalent to solving the problem column-wise, i.e.\
\begin{align*}
\wh{\bm\beta}_{v, \cdot j}(G) = {\arg\min}_{\bm\beta \in \R^{pd}} \vert \bm\beta \vert_1
\quad \text{subject to} \quad
\l\vert \wh{\bbG}_v(G) \bm\beta - \wh{\bbg}_{v, \cdot j}(G) \r\vert_\infty \le \lambda_{n, p}
\quad \text{for } 1 \le j \le p,
\end{align*}
(see e.g.\ Lemma~1 of \cite{cai2011}),
where $\bm\beta_{\cdot j}$ denotes the $j$th column of any $\bm\beta \in \R^{(dp) \times p}$.
Next, we show for any $\cp_{\xi, k} + G \le v \le \cp_{\xi, k + 1}$, 
$\bm\beta^{\k}$ is a feasible solution to~\eqref{eq:ds}, for all $0 \le k \le K_\xi$. This follows from that
\begin{align*}
\l\vert \wh{\bbG}_v(G) \bm\beta^{\k} - \wh{\bbg}_v(G) \r\vert_\infty
= \l\vert \l(\wh{\bbG}_v(G) -\bbG^{\k}\r) \bm\beta^{\k} - 
\l(\wh{\bbg}_v(G) - \bbg^{\k}\r) \r\vert_\infty
\\
\le \l\Vert \bm\beta^{\k} \r\Vert_1 \; \l\vert \wh{\bbG}_v(G) -\bbG^{\k} \r\vert_\infty 
+ \l\vert \wh{\bbg}_v(G) - \bbg^{\k} \r\vert_\infty 
\le \lambda_{n, p}
\end{align*}
on $\mc E^{(2)}_{n, p}$.
Then, $\vert \wh{\bm\beta}_{v, \cdot j}(G) \vert_1 \le \vert \bm\beta^{\k}_{\cdot j} \vert_1$
for $\cp_{\xi, k} + G \le v \le \cp_{\xi, k + 1}$ and consequently,
$\Vert \wh{\bm\beta}_v(G) \Vert_1 \le \Vert \bm\beta^{\k} \Vert_1$.
From this, we have
\begin{align*}
&\max_k \max_v \l\vert \bbG^{\k} \l(\wh{\bm\beta}_v(G) - \bm\beta^{\k}\r) \r\vert_\infty
\le 
\max_k \max_v \l\vert \l(\wh{\bbG}_v(G) \wh{\bm\beta}_v(G) - \wh{\bbg}_v(G) \r) \r.
\\
& \qquad \l.
+ \l(\bbG^{\k} - \wh{\bbG}_v(G)\r) \wh{\bm\beta}_v(G)
+ \l(\wh{\bbg}_v(G) - \bbg^{\k} \r) \r\vert_\infty
\le 2 \lambda_{n, p}.
\end{align*}
\end{proof}

In the remainder of this section, we omit $\xi$ from $\cp_{\xi, k}$ and $\wh\cp_{\xi, k}$ for simplicity.
In what follows, we operate on $\mc M^\chi_{n, p} \cap \mc E^{(2)}_{n, p} \cap \bar{\mc E}^{(2)}_{n, p}$ 
with $\bar{\mc E}^{(2)}_{n, p}$ defined in~\eqref{eq:set:e:two:tilde} below
which, due to Theorem~\ref{thm:common}, Proposition~\ref{prop:idio:est} and Lemma~\ref{lem:Qs:max},
satisfies $\p(\mc M^\chi_{n, p} \cap \mc E^{(2)}_{n, p} \cap \bar{\mc E}^{(2)}_{n, p}) \to 1$.

\begin{proof}[Proof of Theorem~\ref{thm:idio}~\ref{thm:idio:one}]
In the first iteration of Algorithm~\ref{alg:two} with $v_\circ = G$, the estimator 
$\wh{\bm\beta} = \wh{\bm\beta}_{v_\circ}(G)$ satisfies 
\begin{align}
\label{eq:thm:idio:pf:beta}
\l\vert \bbG^{[0]}\l(\wh{\bm\beta} - \bm\beta^{[0]}\r) \r\vert_\infty \le 2 \lambda_{n, p}
\end{align}
and $\Vert \wh{\bm\beta} \Vert_1 \le \Vert \bm\beta^{[0]} \Vert_1$,
due to Proposition~\ref{prop:idio:beta}.
Then for all $v \le \cp_1 - G$, we have
\begin{align}
T_{\xi, v}(\wh{\bm\beta}, G) & \le 
\l\vert \l(\wh{\bbG}_v - \bbG^{[0]}\r) \wh{\bm\beta} \r\vert_\infty
+ \l\vert  \wh{\bbg}_v - \bbg^{[0]} \r\vert_\infty 
+ \l\vert \l(\wh{\bbG}_{v + G} - \bbG^{[0]}\r) \wh{\bm\beta} \r\vert_\infty
+ \l\vert  \wh{\bbg}_{v + G} - \bbg^{[0]} \r\vert_\infty
\nn \\
& \le 2\lambda_{n, p} < \pi_{n, p}.
\label{eq:thm:idio:pf:no}
\end{align}
On the other hand, we have
\begin{align}
& T_{\xi, \cp_1}(\wh{\bm\beta}, G) \ge 
\l\vert \bbG^{[1]} \l(\bm\beta^{[1]} - \bm\beta^{[0]}\r) \r\vert_\infty 
- \l\{\l\vert \l(\wh{\bbG}_{\cp_1} - \bbG^{[0]}\r) \wh{\bm\beta} \r\vert_\infty
+ \l\vert  \wh{\bbg}_{\cp_1} - \bbg^{[0]} \r\vert_\infty
+ \l\vert \bbG^{[0]} \l(\wh{\bm\beta} - \bm\beta^{[0]}\r) \r\vert_\infty \r.
\nn \\
& + \l. \l\vert \l(\wh{\bbG}_{\cp_1 + G} - \bbG^{[1]}\r) \wh{\bm\beta} \r\vert_\infty
+ \l\vert  \wh{\bbg}_{\cp_1 + G} - \bbg^{[1]} \r\vert_\infty
+ \l\vert \bbG^{[1]} (\bbG^{[0]})^{-1} \cdot \bbG^{[0]} \l(\wh{\bm\beta} - \bm\beta^{[0]}\r) \r\vert_\infty \r\}
\nn \\
& \ge \l\vert {\bm\Delta}_{\xi, k} \r\vert_\infty - 
2\l(2 + \Vert \bbG^{[1]}(\bbG^{[0]})^{-1}\Vert_1\r) \lambda_{n, p} > \pi_{n, p}
\label{eq:thm:idio:pf:first}
\end{align}
under Assumption~\ref{assum:idio:size}.
The above~\eqref{eq:thm:idio:pf:no}--\eqref{eq:thm:idio:pf:first} 
guarantee that in the first iteration, $\check\cp$ satisfies $\cp_1 - G < \check\cp \le \cp_1$,
which in turn leads to $\vert \wh\cp_1 - \cp_1 \vert < G$.

Next, we consider the case $\wh\cp_1 \le \cp_1$.
For some $v$ satisfying $\cp_1 - G < v \le \cp_1$, we have
\begin{align}
& T_{\xi, v}(\wh{\bm\beta}, G) = \l\vert
\frac{G - \vert v - \cp_1 \vert}{G} \bbG^{[1]} \l(\bm\beta^{[1]} - \bm\beta^{[0]} \r) + 
\l( \wh{\bbG}_v - \bbG^{[0]}\r) \wh{\bm\beta}
- \l(\wh{\bbg}_v - \bbg^{[0]}\r)
+ \bbG^{[0]}\l(\wh{\bm\beta} - \bm\beta^{[0]}\r) \r.
\nn \\
& - \l( \wh{\bbG}_{v + G} - \frac{\vert v - \cp_1 \vert}{G}\bbG^{[0]} - \frac{G - \vert v - \cp_1 \vert}{G} \bbG^{[1]}\r) \wh{\bm\beta}
+ \l( \wh{\bbg}_{v + G} - \frac{\vert v - \cp_1 \vert}{G}\bbg^{[0]} - \frac{G - \vert v - \cp_1 \vert}{G} \bbg^{[1]}\r)
\nn \\
& \l. - \l(\frac{\vert v - \cp_1 \vert}{G} + \frac{G - \vert v - \cp_1 \vert}{G} \bbG^{[1]}(\bbG^{[0]})^{-1}\r) \bbG^{[0]} \l(\wh{\bm\beta} - \bm\beta^{[0]}\r)
\r\vert_\infty.
\label{eq:thm:idio:pf:decomp}
\end{align}
From~\eqref{eq:thm:idio:pf:beta}, \eqref{eq:thm:idio:pf:decomp}
and Proposition~\ref{prop:idio:est}, it follows that
\begin{align}
& T_{\xi, \wh\cp_1}(\wh{\bm\beta}, G) 
\le \frac{G - \vert \wh\cp_1 - \cp_1 \vert}{G}\l\vert {\bm\Delta}_{\xi, k} \r\vert_\infty + 
\l(6 + 2\Vert \bbG^{[1]} (\bbG^{[0]})^{-1} \Vert_1\r) \lambda_{n, p}, 
\nn \\
& T_{\xi, \cp_1}(\wh{\bm\beta}, G) 
\ge \l\vert {\bm\Delta}_{\xi, k} \r\vert_\infty -
\l(6 + 2\Vert \bbG^{[1]} (\bbG^{[0]})^{-1} \Vert_1\r) \lambda_{n, p}.
\label{eq:thm:idio:pf:loc}
\end{align}
By definition of $\wh\cp_1$, we have
$T_{\xi, \wh\cp_1}(\wh{\bm\beta}, G) \ge T_{\xi, \cp_1}(\wh{\bm\beta}, G)$ such that
\begin{align*}
\frac{\vert \wh\cp_1 - \cp_1 \vert}{G} \l\vert {\bm\Delta}_{\xi, k} \r\vert_\infty \le 
4\l(3 + \Vert \bbG^{[1]} (\bbG^{[0]})^{-1} \Vert_1\r) \lambda_{n, p},
\end{align*}
i.e.\ $\vert \wh\cp_1 - \cp_1 \vert \le \epsilon_0 G$ for some small constant $\epsilon_0 \in (0, 1/2)$
and $n$ large enough under Assumption~\ref{assum:idio:size}~\ref{cond:idio:jump}.
When $\wh\cp_1 > \cp_1$, in place of~\eqref{eq:thm:idio:pf:decomp},
we have the following alternative decomposition of $T_{\xi, v}(\wh{\bm\beta}, G)$
for any $v$ satisfying $\cp_1 < v \le \cp_1 + G$:
\begin{align*}
& T_{\xi, v}(\wh{\bm\beta}, G) = \l\vert
- \frac{G - \vert v - \cp_1 \vert}{G} \bbG^{[1]} \l(\bm\beta^{[1]} - \bm\beta^{[0]} \r) + 
\l( \wh{\bbG}_v - \frac{G - \vert v - \cp_1 \vert}{G} \bbG^{[0]} - \frac{\vert v - \cp_1 \vert}{G} \bbG^{[1]}\r) \wh{\bm\beta} \r.
\\
&- \l(\wh{\bbg}_v - \frac{G - \vert v - \cp_1 \vert}{G} \bbg^{[0]} - \frac{\vert v - \cp_1 \vert}{G} \bbg^{[1]}\r)
+ \l(\frac{G - \vert v - \cp_1 \vert}{G} \bbG^{[0]} + \frac{\vert v - \cp_1 \vert}{G} \bbG^{[1]}\r)
\l(\wh{\bm\beta} - \bm\beta^{[0]}\r) 
\\
& \l. - \l( \wh{\bbG}_{v + G} - \bbG^{[1]}\r) \wh{\bm\beta}
+ \l( \wh{\bbg}_{v + G} - \bbg^{[1]}\r)
- \bbG^{[1]}(\bbG^{[0]})^{-1} \cdot \bbG^{[0]}\l(\wh{\bm\beta} - \bm\beta^{[0]}\r)
\r\vert_\infty
\end{align*}
using which, analogous arguments apply.

After the first iteration, we update $v_\circ$ as
$v_\circ = \min(\check\cp + 2G, \wh\cp_1 + (\eta + 1) G)$ with $\eta > \epsilon_0$ such that
$\cp_1 + G \le v_\circ \le \cp_2$, which ensures that
$\vert \bbG^{[1]}(\wh{\bm\beta}_{v_\circ} - \bm\beta^{[1]}) \vert_\infty \le 2\lambda_{n, p}$ by Proposition~\ref{prop:idio:beta}.
Repeatedly applying the same arguments as those adopted for $\wh\cp_1$,
the conclusion follows.
\end{proof}

\begin{proof}[Proof of Theorem~\ref{thm:idio}~\ref{thm:idio:two}]
For some $1 \le k \le K_\xi$ satisfying the condition in~\ref{thm:idio:two},
suppose that $\wh\cp_k \le \cp_k$;
the following arguments apply analogously
to the case when $\wh\cp_k > \cp_k$.
In what follows, $\wh{\bm\beta}$ denotes the estimator of $\bm\beta^{[k - 1]}$ used at the iteration where $\wh\cp_k$ is added to $\wh\Cp_\xi$
which, by construction, satisfies
\begin{align}
\label{eq:thm:idio:two:pf:beta}
\l\vert \bbG^{[k - 1]} \l(\wh{\bm\beta} - \bm\beta^{[k - 1]} \r) \r\vert_\infty \le 2 \lambda_{n, p}
\quad \text{and} \quad
\l\Vert \wh{\bm\beta} \r\Vert_1 \le \l\Vert \bm\beta^{[k - 1]} \r\Vert_1,
\end{align}
see Proposition~\ref{prop:idio:beta}.
By definition, we can find $\bm\varphi_l \in \R^{pd}$ and $\bm\varphi_r \in \R^p$,
each a vector of zeros except for a single element set to be one, such that
\begin{align*}
T_{\xi, \wh\cp_k}(\wh{\bm\beta}, G) = \l\vert
\bm\varphi_l^\top \l( \wh{\bbG}_{\wh\cp_k} \wh{\bm\beta} - \wh{\bbg}_{\wh\cp_k}
- \wh{\bbG}_{\wh\cp_k + G} \wh{\bm\beta} + \wh{\bbg}_{\wh\cp_k + G}
\r) \bm\varphi_r \r\vert.
\end{align*}
Then, the first statement in~\eqref{eq:thm:idio:pf:loc} can be re-written as
\begin{align}
& T_{\xi, \wh\cp_k}(\wh{\bm\beta}, G) 
\le \frac{G - \vert \wh\cp_k - \cp_k \vert}{G}
\l\vert \bm\varphi_l^\top {\bm\Delta}_{\xi, k} \bm\varphi_r \r\vert + 
\l(6 + 2\Vert \bbG^{[k]} (\bbG^{[k - 1]})^{-1} \Vert_1\r) \lambda_{n, p}, \quad \text{such that}
\nn \\
& \frac{G - \vert \wh\cp_k - \cp_k \vert}{G}\l\vert \bm\varphi_l^\top {\bm\Delta}_{\xi, k} \bm\varphi_r \r\vert 
\ge \l\vert {\bm\Delta}_{\xi, k} \r\vert_\infty - 4\l(3 + \Vert \bbG^{\k} (\bbG^{[k - 1]})^{-1} \Vert_1\r) \lambda_{n, p}, \nn \\
& \therefore \quad \l\vert \bm\varphi_l^\top {\bm\Delta}_{\xi, k} \bm\varphi_r \r\vert 
\ge \frac{1}{2} \l\vert {\bm\Delta}_{\xi, k} \r\vert_\infty
\label{eq:thm:idio:two:pf:size}
\end{align}
for $n$ large enough.
WLOG, suppose that $\bm\varphi_l^\top {\bm\Delta}_{\xi, k} \bm\varphi_r > 0$.
Then from that $\vert \wh\cp_k - \cp_k \vert \le \epsilon_0 G$,
\begin{align*}
\bm\varphi_l^\top \l( \wh{\bbG}_{\wh\cp_k} \wh{\bm\beta} - \wh{\bbg}_{\wh\cp_k}
- \wh{\bbG}_{\wh\cp_k + G} \wh{\bm\beta} + \wh{\bbg}_{\wh\cp_k + G}
\r) \bm\varphi_r
\ge \frac{1}{2} \bm\varphi_l^\top {\bm\Delta}_{\xi, k} \bm\varphi_r
- \l(6 + 2\Vert \bbG^{\k} (\bbG^{[k - 1]})^{-1} \Vert_1\r) \lambda_{n, p} > 0
\end{align*}
and similarly, 
\begin{align*}
\bm\varphi_l^\top \l( \wh{\bbG}_{\cp_k} \wh{\bm\beta} - \wh{\bbg}_{\cp_k}
- \wh{\bbG}_{\cp_k + G} \wh{\bm\beta} + \wh{\bbg}_{\cp_k + G}
\r) \bm\varphi_r > 0.
\end{align*}
Observing that
\begin{align*}
T_{\xi, \wh\cp_k}(\wh{\bm\beta}, G)
& = \bm\varphi_l^\top \l( \wh{\bbG}_{\wh\cp_k} \wh{\bm\beta} - \wh{\bbg}_{\wh\cp_k}
- \wh{\bbG}_{\wh\cp_k + G} \wh{\bm\beta} + \wh{\bbg}_{\wh\cp_k + G}
\r) \bm\varphi_r
\\
&  \ge T_{\xi, \cp_k}(\wh{\bm\beta}, G)
\ge \bm\varphi_l^\top \l( \wh{\bbG}_{\cp_k} \wh{\bm\beta} - \wh{\bbg}_{\cp_k}
- \wh{\bbG}_{\cp_k + G} \wh{\bm\beta} + \wh{\bbg}_{\cp_k + G}
\r) \bm\varphi_r,
\end{align*}
we obtain
\begin{align}
& \mc F_k := \frac{\vert \wh\cp_k - \cp_k \vert}{G}  \bm\varphi_l^\top {\bm\Delta}_{\xi, k} \bm\varphi_r \le 
\nn \\
& \qquad \bm\varphi_l^\top 
\l(\wh{\bbG}_{\wh\cp_k} - \wh{\bbG}_{\cp_k} - 
\wh{\bbG}_{\wh\cp_k + G} + \wh{\bbG}_{\cp_k + G}
+ \frac{\vert \wh\cp_k - \cp_k \vert}{G}(\bbG^{\k} - \bbG^{[k - 1]})
\r) \wh{\bm\beta} \; \bm\varphi_r
\nn \\
& \qquad  - \bm\varphi_l^\top \l(\wh{\bbg}_{\wh\cp_k} - \wh{\bbg}_{\cp_k} - 
\wh{\bbg}_{\wh\cp_k + G} + \wh{\bbg}_{\cp_k + G}
+ \frac{\vert \wh\cp_k - \cp_k \vert}{G}(\bbg^{\k} - \bbg^{[k - 1]})
\r) \bm\varphi_r
\nn \\
& \qquad - \frac{\vert \wh\cp_k - \cp_k \vert}{G} 
\bm\varphi_l^\top\l(\bbG^{[k - 1]} - \bbG^{\k}\r)
\l(\wh{\bm\beta} - \bm\beta^{[k - 1]}\r) \bm\varphi_r
=: \mc R_{k1} + \mc R_{k2} + \mc R_{k3}.
\label{eq:thm:idio:two:contradict}
\end{align}
We adopt the proof by contradiction: Supposing that
$\vert \wh\cp_k - \cp_k \vert > c_0 \varrho^{\k}_{n, p}$,
we show that
the above inequality in~\eqref{eq:thm:idio:two:contradict} does not hold and consequently,
it cannot hold that
$T_{\xi, \wh\cp_k}(\wh{\bm\beta}, G) \ge T_{\xi, \cp_k}(\wh{\bm\beta}, G)$. 
By~\eqref{eq:thm:idio:two:pf:beta}, we have
\begin{align*}
\vert \mc R_{k3} \vert \le \frac{2}{G} \vert \wh\cp_k - \cp_k \vert 
(1 + \Vert \bbG^{\k}(\bbG^{[k - 1]})^{-1} \Vert_1) \lambda_{n, p}
\le \epsilon \mc F_k
\end{align*}
for an arbitrarily small constant $\epsilon \in (0, 1)$
due to Assumption~\ref{assum:idio:size}~\ref{cond:idio:jump}.
In order to control $\mc R_{k1}$ and $\mc R_{k2}$,
we note that since $\{\cp_{\xi, k} - 2G + 1, \ldots, \cp_{\xi, k} + 2G\} \cap \Cp_\chi = \emptyset$
(and therefore $\{\cp_{\xi, k} - \lfloor 3G/2 \rfloor + 1, \ldots, \cp_{\xi, k} + \lfloor 3G/2 \rfloor \} \cap \wh\Cp_\chi = \emptyset$
on $\mc M^\chi_{n, p}$),
we have
\begin{align*}
& \l\vert \wh{\bm\Gamma}_{\xi, \wh\cp_k + G}(\ell, G) - \wh{\bm\Gamma}_{\xi, \cp_k + G}(\ell, G) - \frac{\vert \wh\cp_k - \cp_k \vert}{G}\l(\bm\Gamma^{[k - 1]}_\xi(\ell) - \bm\Gamma^{\k}_\xi(\ell)\r) \r\vert_\infty
\\
=&  \l\vert \wh{\bm\Gamma}_{x, \wh\cp_k + G}(\ell, G) - \wh{\bm\Gamma}_{x, \cp_k + G}(\ell, G)
- \frac{\vert \wh\cp_k - \cp_k \vert}{G}\l(\bm\Gamma_{x, \wh\cp_k + G}(\ell, G) - \bm\Gamma_{x, \cp_k + G}(\ell, G)\r) \r\vert_\infty
\\
\le&  \l\vert \wh{\bm\Gamma}_{x, \wh\cp_k + G}(\ell, G) - \wh{\bm\Gamma}_{x, \cp_k + G}(\ell, G) -
\E\l(\wh{\bm\Gamma}_{x, \wh\cp_k + G}(\ell, G) - \wh{\bm\Gamma}_{x, \cp_k + G}(\ell, G)\r) \r\vert_\infty 
\\
& + \l\vert \E\l(\wh{\bm\Gamma}_{x, \wh\cp_k + G}(\ell, G) - \wh{\bm\Gamma}_{x, \cp_k + G}(\ell, G)\r) - \frac{\vert \wh\cp_k - \cp_k \vert}{G}\l(\bm\Gamma_{x, \wh\cp_k + G}(\ell, G) - \bm\Gamma_{x, \cp_k + G}(\ell, G)\r) \r\vert_\infty
\\
\le&  \frac{\vert \wh\cp_k - \cp_k \vert}{G} \l\{\l\vert \mbf Q^{(1)}_k(\ell, \cp_k - \wh\cp_k, 0) - \E\l(\mbf Q^{(1)}_k(\ell, \cp_k - \wh\cp_k, 0)\r) \r\vert_\infty \r.
\\
& +  \l\vert \mbf Q^{(2)}_k(\ell, \cp_k - \wh\cp_k, G) - \E\l(\mbf Q^{(2)}_k(\ell, \cp_k - \wh\cp_k, G)\r) \r\vert_\infty 
\\
& + \l\vert \E\l(\mbf Q^{(1)}_k(\ell, \cp_k - \wh\cp_k, 0)\r) - \bm\Gamma_{x, \cp_k + \ell}(\ell, \cp_k - \wh\cp_k) \r\vert_\infty
\\
& + \l. \l\vert \E\l(\mbf Q^{(2)}_k(\ell, \cp_k - \wh\cp_k, G)\r) - \bm\Gamma_{x, \cp_k + G}(\ell, \cp_k - \wh\cp_k) \r\vert_\infty
\r\} =: \mc R_{k4} + \mc R_{k5} + \mc R_{k6} + \mc R_{k7},
\end{align*}
with $\mbf Q^{(r)}_k(\ell, h, H)$ defined in~\eqref{eq:def:Q:two}.
It is easily seen that $\mc R_{k7} = 0$, while
\begin{align*}
\frac{G}{\vert \wh\cp_k - \cp_k \vert} \mc R_{k6} \le  \frac{\vert \ell \vert}{\vert \wh\cp_k - \cp_k \vert} \l\vert \bm\Gamma_{x, \cp_k}(\ell, G) \r\vert_\infty \le \frac{C_{\Xi, \varsigma, \vep} (1 + \vert \ell \vert)^{-\varsigma + 1}}{\vert \wh\cp_k - \cp_k \vert} \le \frac{c_1}{\vert \wh\cp_k - \cp_k \vert}
\end{align*}
from Lemma~\ref{lem:decay:two} for some constant $c_1 > 0$.
Also by Lemma~\ref{lem:Qs:max}, we have $\p(\bar{\mc E}^{(2)}_{n, p}) \to 1$ where, with 
$w_k = \vert {\bm\Delta}_{\xi, k} \vert_\infty$, 
$I_k = \{h: \, w_k^2 \delta \le \vert h \vert \le G\}$ and $\wt{\vartheta}(\delta)$ defined in the lemma,
\begin{align}
\bar{\mc E}^{(2)}_{n, p} = \l\{
\max_{1 \le k \le K_\xi} \max_{1 \le r \le 2} \max_{H \in \{0, \pm G\}} \max_{h \in I_k}  \max_{0 \le \ell \le d} 
w_k \l\vert \mbf Q^{(r)}_k(\ell, h, H) - \E(\mbf Q^{(r)}_k(\ell, h, H)) \r\vert_\infty \le c_2 \wt{\vartheta}(\delta) \r\}
\label{eq:set:e:two:tilde}
\end{align}
for some $c_2 > 0$, such that we obtain
\begin{align*}
\max\l(\mc R_{k4}, \mc R_{k5} \r)
\le \frac{\vert \wh\cp_k - \cp_k \vert}{G} \cdot
c_2 \vert {\bm\Delta}_{\xi, k} \vert_\infty \wt{\vartheta}(\delta)
\end{align*}
on $\bar{\mc E}^{(2)}_{n, p}$. 
We can similarly show that
\begin{align*}
\l\vert \wh{\bm\Gamma}_{\xi, \wh\cp_k}(\ell, G) - \wh{\bm\Gamma}_{\xi, \cp_k}(\ell, G) \r\vert_\infty
\le \frac{\vert \wh\cp_k - \cp_k \vert}{G} \cdot
c_2 \vert {\bm\Delta}_{\xi, k} \vert_\infty \wt{\vartheta}(\delta).
\end{align*}
Setting $\delta = c_0 w_k^{-2} \varrho^{\k}_{n, p}$ and
putting together the bounds on $\mc R_{kr}, \, 4 \le r \le 7$, 
we can choose a large enough $c_0$ such that,
\begin{align*} 
\vert \mc R_{k1} \vert + \vert \mc R_{k2} \vert \le 
\frac{\vert \wh\cp_k - \cp_k \vert}{G} \l(1 + \Vert \bm\beta^{[k - 1]} \Vert_1\r) \l( \frac{c_1}{\vert \wh\cp_k - \cp_k \vert} + 2c_2 \vert {\bm\Delta}_{\xi, k} \vert_\infty \wt{\vartheta}(\delta) \r) 
\\
\le \frac{\vert \wh\cp_k - \cp_k \vert}{G}
\l(\frac{c_1}{c_0 \varrho^{\k}_{n, p}} + \frac{2c_2 \vert {\bm\Delta}_{\xi, k} \vert_\infty}{\min(c_0, \sqrt{c_0})} \r)
< (1 - \epsilon) \mc F_k
\end{align*}
from~\eqref{eq:thm:idio:two:pf:size}.
This, together with the bound on $\vert \mc R_{k3} \vert$, 
shows that the inequality in~\eqref{eq:thm:idio:two:contradict} does not hold,
and thus we prove the claim.
Since all the arguments are conditional on $\mc E^{(2)}_{n, p} \cap \bar{\mc E}^{(2)}_{n, p}$, 
which in turn are formulated uniformly over $1 \le k \le K_\xi$, the proof is complete.
\end{proof}

\begin{proof}[Proof of Corollary~\ref{cor:idio}]
For the proof of~\ref{cor:idio:station}, we first note that
under the stationarity of $\bm\chi_t$, 
we have $\wh K_\chi = 0$ on $\mc M^\chi_{n, p}$ such that $\rho_{n, p} = 0$.
Therefore, we have $\p(\mc E^{(2)\prime}_{n, p}) \to 1$ where
\begin{align*}
\mc E^{(2)\prime}_{n, p} = \l\{
\max_{G \le v \le n} \max_{0 \le \ell \le d} 
\l\vert \wh{\bm\Gamma}_{\xi, v}(\ell, G) - \bm\Gamma_{\xi, v}(\ell, G) \r\vert_\infty
\le M\l(\vartheta_{n, p} \vee \frac{1}{m} \vee \frac{1}{\sqrt p}\r) \r\}.
\end{align*}
Operating on $\mc M^\chi_{n, p} \cap \mc E^{(2)\prime}_{n, p} \cap \bar{\mc E}^{(2)}_{n, p}$,
analogous arguments as those adopted in Theorem~\ref{thm:idio} apply.

For the proof of~\ref{cor:idio:no}, we proceed similarly as in the case of~\ref{cor:idio:station}
except that now we have $\p(\mc E^{(2)\prime\prime}_{n, p}) \to 1$ thanks to Lemma~\ref{lem:acv:x:max},
where
\begin{align*}
\mc E^{(2)\prime\prime}_{n, p} = \l\{
\max_{G \le v \le n} \max_{-d \le \ell \le d} 
\l\vert \wh{\bm\Gamma}_{\xi, v}(\ell, G) - \bm\Gamma_{\xi, v}(\ell, G) \r\vert_\infty
\le M \bar{\vartheta}_{n, p} \r\}.
\end{align*}
\end{proof}

\subsubsection{Supporting results}

In what follows, we operate 
under the assumptions made in Theorem~\ref{thm:idio}.
We define $\Gamma_{\chi, v}(\ell, G)$ analogously as $\bm\Gamma_{\xi, v}(\ell, G)$
with $\cp_{\chi, k}$ in place of $\cp_{\xi, k}$ and let
$\bm\Gamma_{x, v}(\ell, G) = \bm\Gamma_{\chi, v}(\ell, G) + \bm\Gamma_{\xi, v}(\ell, G)$.

\begin{lem}
\label{lem:acv:x:max}
Recall the definition of $\bar{\vartheta}_{n, p}$ in~\eqref{eq:bar:vartheta}. Then,
\begin{align*}
\max_{G \le v \le n} \max_{0 \le \ell \le d}
\l\vert \wh{\bm\Gamma}_{x, v}(\ell, G) - \bm\Gamma_{x, v}(\ell, G) \r\vert_\infty
= O_p\l(\bar{\vartheta}_{n, p}\r).
\end{align*}
\end{lem}

\begin{proof}
By Theorems~3.1 and~3.2 of \cite{zhang2021},
there exist universal constants $C_1, C_2 > 0$
and constants $C_\alpha, C_{\nu, \alpha} > 0$ that depend only on their subscripts,
such that for any $z > 0$,
\begin{align*}
& \p\l( \max_v \max_\ell \l\vert \wh{\bm\Gamma}_{x, v}(\ell, G) 
- \E(\wh{\bm\Gamma}_{x, v}(\ell, G)) \r\vert_\infty \le z \r) \le
\\
& \l\{\begin{array}{l}
\frac{C_{\nu, \alpha} n d^{\nu/4} \log^{\nu + 1}(G) (\log^{3/2}(p) p^{1/\nu})^\nu}{(Gz)^{\nu/2}}
+ C_1 np^2 \exp\l(-\frac{Gz^2}{C_\alpha m \Phi_{4, \alpha}^4}\r) \\
\quad \text{under Assumption~\ref{assum:innov}~\ref{cond:moment}}, 
\\
2 np^2d \exp\l[- C_2 \min\l( \frac{G z^2}{\Phi_{2, 0}^4}, \frac{Gz}{\Phi_{2, 0}^2} \r) \r] 
\\
\quad \text{under~Assumption~\ref{assum:innov}~\ref{cond:gauss},}
\end{array}\r.
\end{align*}
such that $\max_v \max_\ell \vert \wh{\bm\Gamma}_{x, v}(\ell, G) 
- \E(\wh{\bm\Gamma}_{x, v}(\ell, G)) \vert_\infty = O_p(\bar{\vartheta}_{n, p})$,
thanks to Lemma~\ref{lem:func:dep}.
As for the bias term, applying the arguments adopted in the proof of Lemma~\ref{lem:spec:x}, 
it is shown that
\begin{align*}
\max_v \max_\ell \l\vert \E(\wh{\bm\Gamma}_{x, v}(\ell, G)) - \bm\Gamma_{x, v}(\ell, G) \r\vert_\infty
= O\l(\frac{(1 + \vert \ell \vert)^{-\varsigma + 1}}{G}\r) = o(\bar{\vartheta}_{n, p}),
\end{align*}
which completes the proof.
\end{proof}

For $1 \le k \le K_\xi$, $H \in \{0, \pm G\}$ and $\ell \ge 0$, define
\begin{align}
\mbf Q^{(1)}_k(\ell, h, H) &= \frac{1}{|h|} 
\sum_{t = (\cp_{\xi, k} - h) \wedge \cp_{\xi, k} + H + \ell + 1}^{(\cp_{\xi, k} - h) \vee \cp_{\xi, k} + H + \ell} 
\mbf X_{t - \ell} \mbf X_t^\top,
\nn \\
\mbf Q^{(2)}_k(\ell, h, H) &= \frac{1}{|h|} 
\sum_{t = (\cp_{\xi, k} - h) \wedge \cp_{\xi, k} + H + 1}^{(\cp_{\xi, k} - h) \vee \cp_{\xi, k} + H} 
\mbf X_{t - \ell} \mbf X_t^\top.
\label{eq:def:Q:two}
\end{align}
\begin{lem}
\label{lem:Qs:max}
For some fixed $d \in \Z$ and $\delta \in \{d, \ldots, G\}$,
\begin{align*}
& \max_{1 \le k \le K_\xi} \max_{1 \le r \le 2} \max_{H \in \{0, \pm G\}} 
\max_{h \in I_k} \max_{0 \le \ell \le d} 
w_k \l\vert \mbf Q^{(r)}_k(\ell, h, H) - \E\l(\mbf Q^{(r)}_k(\ell, h, H)\r) \r\vert_\infty 
= O_p(\wt{\vartheta}(\delta)),
\end{align*}
where $w_k = \vert {\bm\Delta}_{\xi, k} \vert_\infty^{-1}$, 
$I_k = \{h: \, w_k^2 \delta \le \vert h \vert \le G\}$ and
\begin{align*}
\wt{\vartheta}(\delta) = \l\{\begin{array}{ll}
\frac{(K_\xi G)^{2/\nu} p^{2/\nu} \log^3(p)}{\delta^{1 - 2/\nu}} \vee
\sqrt{\frac{\log(G K_\xi p)}{\delta}} 
& \text{under Assumption~\ref{assum:innov}~\ref{cond:moment}},
\\
\sqrt{\frac{\log(G K_\xi p)}{\delta}} 
& \text{under Assumption~\ref{assum:innov}~\ref{cond:gauss}.}
\end{array}\r.
\end{align*}
\end{lem}

\begin{proof}
Applying Theorems~6.4 and~6.5 of \cite{zhang2021} with Bonferroni correction, 
there exist universal constant $C_1, C_2 > 0$
and constants $C_\alpha, C_{\nu, \alpha} > 0$
that depend only on their subscripts, such that for any $z > 0$,
\begin{align*}
& \p\l(\max_{1 \le k \le K_\xi} \max_{1 \le r \le 2} \max_{H \in \{0, \pm G\}} 
\max_{h \in I_k} \max_{0 \le \ell \le d} 
w_k \l\vert \mbf Q^{(r)}_k(\ell, h, H) - \E\l(\mbf Q^{(r)}_k(\ell, h, H)\r) \r\vert_\infty \ge z\r) \le 
\\
& \l\{\begin{array}{ll}
\frac{C_{\nu, \alpha} K_\xi G d^{\nu/4} (p^{1/\nu}\log^{3/2}(p))^\nu}
{\delta^{\nu/2 - 1} z^{\nu/2}} + C_1 K_\xi G d p^2 \exp\l(-\frac{\delta z^2}{C_\alpha \Phi_{4, \alpha}^4}\r) & \text{under Assumption~\ref{assum:innov}~\ref{cond:moment}},
\\
24 K_\xi G d p^2 \exp\l( - \frac{C_2\delta z^2}{\Phi_{2, 0}^4} \r)
& \text{under Assumption~\ref{assum:innov}~\ref{cond:gauss}}
\end{array}
\r.
\end{align*}
thanks to Lemma~\ref{lem:func:dep}, 
which completes the proof.
\end{proof}

\end{document}